%% file: main.tex
\documentclass
[
    fontsize = 11 pt,        % The size of the font.
    american,                % Support for American English.
    captions = tableheading, % Places the correct amount of space when the caption of a table is above the table.
    numbers = noenddot,      % Does not use a period at the end of numbered titles, such as sections or figures.
    abstracton,
    paper = letter,
    DIV = 13
]
{scrartcl}

\input{core/standard_packages.tex}

\input{packages_and_commands/additional_packages.tex}
\input{core/general_commands.tex}

\input{packages_and_commands/additional_commands.tex}

\title{Sampling repulsive Gibbs point processes using random graphs}

 \author{Tobias Friedrich$^{*}$ \and Andreas Göbel$^{*}$ \and Maximilian Katzmann$^{\dagger}$ \and Martin~S. Krejca$^{\ddag}$ \and Marcus Pappik$^{*}$}

 \publishers
 {%
 	$^{*}$\footnotesize Hasso Plattner Institute, University of Potsdam, Potsdam, Germany\\
 	\{tobias.friedrich, andreas.goebel, marcus.pappik\}@hpi.de \\[1 ex]

     $^{\dagger}$ Karlsruhe Institute of Technology, Karlsruhe, Germany \\
     maximilian.katzmann@kit.edu \\[1 ex]

 	$^{\ddag}$ Laboratoire d'Informatique (LIX), CNRS, École Polytechnique, Institut Polytechnique de Paris \\
 	martin.krejca@polytechnique.edu
 }

\begin{document}

	\maketitle
	\input{content/abstract}

\input{content/introduction}

	\input{content/preliminaries}

	\input{content/general_concentration}

\input{content/spin_systems}

	\input{content/point_processes}

	\input{content/sampling}

	\input{content/connective_constant}

	\section*{Aknowledgements}
	Andreas Göbel was funded by the project PAGES (project No. 467516565) of the German Research Foundation (DFG).

	\bibliographystyle{plainurl}% the mandatory bibstyle
	\bibliography{references.bib}
	
\end{document}

%% file: core/standard_packages.tex
\usepackage[T1]{fontenc}
\usepackage[utf8]{inputenc}
\usepackage
[
    babel = true, % Enables language-specific tuning.
]
{microtype}           % Uses the text space more efficiently.
\usepackage{csquotes} % Uses the correct quotes according to the current language.

\usepackage[dvipsnames]{xcolor} % Allows it to define colors. The option says that common names can be used.

\definecolor{stroke1}{HTML}{2574A9} % This color is used as the standard color to highlight things.

\usepackage{amsmath}
\usepackage{amssymb}
\usepackage{amsthm}
\usepackage{thmtools}
\usepackage{mathtools}
\usepackage{thm-restate}
\usepackage{dsfont}        % Yields far better blackboard-bold letters than \mathbb. Use \mathds in order to write such letters.
\usepackage{braceMnSymbol} % Adjusts overbraces and underbraces such that longer versions are put together seamlessly.

\usepackage
[
    ttscale = 0.85, % Scales the typewriter font.
]
{libertine} % The main font used in this thesis.
\usepackage
[
    libertine,    % Changes the math font to libertine (the main font).
    slantedGreek, % Makes all greek letters italic by default. If you want to use an upright greek letter, use ›\up‹ immediately followed by the letter’s name. For example, \upGamma displays an upright uppercase gamma.
    vvarbb,       % Changes the \mathbb font to another font. However, \mathbb remains ugly and should not be used. Use \mathds instead.
    libaltvw,     % Uses different characters for v und w that look far better than the default ones.
]
{newtxmath} % The main math font of this thesis. It fits well with the main font.
\usepackage{url} % Responsible for URL formatting.
\usepackage{bm}  % Allows to use sensible bold letters in math mode. This package has to go after the font packages. Otherwise it does not work correctly!

\usepackage{graphicx} % The standard package for including graphics into your document.
\usepackage
[
    subrefformat = simple, % Formats the label of the \subref command without parentheses.
    labelformat = simple,  % Formats the mark of a subfigure without parentheses.
]
{subcaption}         % Enables it to have subfigures inside of a single figure.

\usepackage{array}     % Improves the way that tables can be formatted.
\usepackage{booktabs}  % Adds lines (called ›rules‹) that can be used in tables and improves spacing.
\usepackage{longtable} % Allows to make tables that span multiple pages.
\usepackage{pdflscape} % Allows to change a page into landscape. This is handy if a table is very wide.

\usepackage[shortlabels]{enumitem} % Adds tons of useful features to enumeration environments.

\usepackage
[
    ruled,         % Creates lines at the top and at the bottom. Further, the caption is now above the algorithm.
    vlined,        % Shows the scope of a statement spanning multiple lines via a small vertical bar. Thus, no closing tags are needed.
    linesnumbered, % Shows line numbers.
]
{algorithm2e} % Allows to write pseudocode.

\usepackage{xspace}   % Adds the functionality that a space after a command will be shown as a space in the output.
\usepackage
[
    shortcuts, % Allows to use short symbols for non-breaking hyphens and dashes instead of lengthy commands.
]
{extdash}             % Adds non-breaking hyphens and dashes.
\usepackage{setspace} % Allows to easily chnage the spacing inside of the document.

\usepackage{xparse}    % Is used in order to define reasonable commands.
\usepackage{footnote}  % Allows it to extend the environments footnotes can be used in. It is said that this package is in conflict with ›hyperref‹. I did not note any troubles. However, if something is fishy, it is probably best to not use this package.
\usepackage{afterpage} % Adds the \afterpage command, which specifies that the provided argument shall be processed after the current page is finished.
\usepackage
[
    textsize = scriptsize, % Determines the text size of the TODO note.
]
{todonotes}            % Adds TODO notes to the document. These are small text areas inside of the margin of a page.

\input{packages_and_commands/additional_packages}

\usepackage
[
    bookmarks = true,                 % Generates boodmarks for the PDF.
    bookmarksopen = false,            % The bookmarks are closed by default.
    bookmarksnumbered = true,         % The bookmarks use the numbers of the corresponding headline.
    pdfstartpage = 1,                 % The first page seen when opening the PDF.
    colorlinks = true,                % The text of hyperlinks is colored instead of having a colored box around it.
    allcolors = stroke1,              % Every hyperlink uses the same color. If you want to change specific colors, use the commands below.
]
{hyperref} % The standard package that is used for creating hyperlinks inside of a document.

\usepackage
[
    noabbrev,   % The words in front of the labels are not abbreviated.
    nameinlink, % Extends the link of a reference to the word in front of it.
]
{cleveref} % This package must be included after ›hyperref‹. It creates clever references that know what they refer to.

%% file: packages_and_commands/additional_packages.tex
\usepackage{mfirstuc}
\usepackage{algorithmicx}
\usepackage{amsfonts}
\usepackage
[
    labelfont = bf,
    format = plain,
]
{caption}

\usepackage{mathtools}

\usepackage{bm}

\usepackage{bigints}

%% file: core/general_commands.tex
\date{}

\makeatletter
    \def\IfEmptyTF#1%
    {%
        \if\relax\detokenize{#1}\relax%
            \expandafter\@firstoftwo%
        \else%
            \expandafter\@secondoftwo%
        \fi%
    }
\makeatother

\NewDocumentCommand{\mathOrText}{m}
{%
    \ensuremath{#1}\xspace%
}

\let\originalleft\left
\let\originalright\right
\renewcommand{\left}{\mathopen{}\mathclose\bgroup\originalleft}
\renewcommand{\right}{\aftergroup\egroup\originalright}

\makeatletter
    \DeclareRobustCommand{\bfseries}%
    {%
        \not@math@alphabet\bfseries\mathbf%
        \fontseries\bfdefault\selectfont%
        \boldmath%
    }

\xspaceaddexceptions{]\}}

\allowdisplaybreaks

\crefname{ineq}{inequality}{inequalities}
\creflabelformat{ineq}{#2{\upshape(#1)}#3}

\crefname{term}{term}{terms}
\creflabelformat{term}{#2{\upshape(#1)}#3}

\crefname{cond}{condition}{conditions}
\creflabelformat{cond}{#2{\upshape(#1)}#3}

\crefname{assume}{assumption}{assumptions}
\creflabelformat{assume}{#2{\upshape(#1)}#3}

\deffootnote[1.2 em]{1.2 em}{0 em}{\makebox[1.2 em][l]{\textbf{\thefootnotemark}}}

\let\oldfootnote\footnote

\newlength{\spaceBeforeFootnote} % Denotes the space before the footnote mark in em.
\newlength{\spaceAfterFootnote}  % Denotes the space after the footnote mark in em.

\RenewDocumentCommand{\footnote}{o o o m}%
{%
    \IfNoValueTF{#1}%
    {%
        \oldfootnote{#4}%
    }%
    {%
        \setlength{\spaceBeforeFootnote}{\IfEmptyTF{#1}{0}{#1} em}%
        \IfNoValueTF{#2}%
        {%
            \hspace*{\spaceBeforeFootnote}\oldfootnote{#4}%
        }%
        {%
            \setlength{\spaceAfterFootnote}{\IfEmptyTF{#2}{0}{#2} em}%
            \hspace*{\spaceBeforeFootnote}\IfNoValueTF{#3}{\oldfootnote{#4}}{\oldfootnote[#3]{#4}}\hspace*{\spaceAfterFootnote}%
        }%
    }%
}

\makesavenoteenv{figure}
\makesavenoteenv{table}
\makesavenoteenv{tabular}

\declaretheoremstyle
[
   	spaceabove = \topsep,
   	spacebelow = \topsep,
   	headfont = \bfseries,
   	headformat = \textcolor{stroke1}{$\blacktriangleright$} \NAME~\NUMBER \NOTE,
   	notefont = \bfseries,
   	notebraces = {(}{)},
   	bodyfont = \normalfont,
   	postheadspace = 0.5 em,
   	qed = \textcolor{stroke1}{\bfseries$\blacktriangleleft$},
]
{myTheoremStyle}

\declaretheorem
[
   	style = myTheoremStyle,
   	name = Lemma,
    sharenumber = conjecture,
]
{lemma}
\declaretheorem
[
   	style = myTheoremStyle,
   	name = Corollary,
    sharenumber = conjecture,
]
{corollary}
\declaretheorem
[
   	style = myTheoremStyle,
   	name = Theorem,
    sharenumber = conjecture,
]
{theorem}
\declaretheorem
[
   	style = myTheoremStyle,
   	name = Definition,
    sharenumber = conjecture,
]
{definition}

\declaretheorem
[
    style = myTheoremStyle,
    name = Remark,
    sharenumber = conjecture,
]
{remark}
\declaretheorem
[
    style = myTheoremStyle,
    name = Observation,
    sharenumber = conjecture,
]
{observation}

\NewDocumentCommand{\functionTemplate}{m m m m o}%
{%
    \IfNoValueTF{#5}%
    {%
        \mathOrText{#1\left#2{#4}\right#3}%
    }%
    {%
        \mathOrText{#1#5#2{#4}#5#3}%
    }%
}

\newcommand*{\leftBracketType}{(}
\newcommand*{\rightBracketType}{)}

\NewDocumentCommand{\createFunction}{m m o o}%
{%
    \renewcommand*{\leftBracketType}{\IfNoValueTF{#3}{(}{#3}}%
    \renewcommand*{\rightBracketType}{\IfNoValueTF{#4}{)}{#4}}%
    \NewDocumentCommand{#1}{o o}%
    {%
        \IfNoValueTF{##1}%
        {%
            \mathOrText{#2}%
        }%
        {%
            \functionTemplate{#2}{\leftBracketType}{\rightBracketType}{##1}[##2]%
        }%
    }%
}

\DeclareDocumentCommand{\probabilisticFunctionTemplate}{m m O{} o}
{%
    \functionTemplate{#1}%
    {\lbrack}%
    {\rbrack}%
    {#2\IfEmptyTF{#3}{}{\ \IfNoValueTF{#4}{\left}{#4}\vert\ \vphantom{#2}#3\IfNoValueTF{#4}{\right.}{}}}%
    [#4]%
}

\newcommand*{\N}{\mathOrText{\mathds{N}}}

\newcommand*{\R}{\mathOrText{\mathds{R}}}

\newcommand*{\indicatorFunctionSymbol}{\mathds{1}}

\RenewDocumentCommand{\Pr}{m O{} o}%
{%
    \probabilisticFunctionTemplate{\mathrm{Pr}}{#1}[#2][#3]%
}

\NewDocumentCommand{\E}{m O{} o}%
{%
    \probabilisticFunctionTemplate{\mathrm{E}}{#1}[#2][#3]%
}

\NewDocumentCommand{\Var}{m O{} o}%
{%
    \probabilisticFunctionTemplate{\mathrm{Var}}{#1}[#2][#3]%
}

\DeclareDocumentCommand{\bigO}{m o}%
{%
    \functionTemplate{\mathrm{O}}{(}{)}{#1}[#2]%
}

\DeclareDocumentCommand{\smallO}{m o}%
{%
    \functionTemplate{\mathrm{o}}{(}{)}{#1}[#2]%
}

\DeclareDocumentCommand{\bigTheta}{m o}%
{%
    \functionTemplate{\upTheta}{(}{)}{#1}[#2]%
}

\DeclareDocumentCommand{\bigOmega}{m o}%
{%
    \functionTemplate{\upOmega}{(}{)}{#1}[#2]%
}

\DeclareDocumentCommand{\smallOmega}{m o}%
{%
    \functionTemplate{\upomega}{(}{)}{#1}[#2]%
}

\DeclareDocumentCommand{\eulerE}{o}%
{%
    \mathOrText{\mathrm{e}\IfNoValueTF{#1}{}{^{#1}}}%
}

\DeclareDocumentCommand{\poly}{m o}%
{%
    \functionTemplate{\mathrm{poly}}{(}{)}{#1}[#2]%
}

\createFunction{\id}{\mathrm{id}}

\NewDocumentCommand{\ind}{m o o}%
{%
    \IfNoValueTF{#2}%
    {%
        \mathOrText{\indicatorFunctionSymbol_{#1}}%
    }%
    {%
        \functionTemplate{\indicatorFunctionSymbol_{#1}}{(}{)}{#2}[#3]%
    }%
}

\DeclareDocumentCommand{\dom}{m o}%
{%
    \functionTemplate{\mathrm{dom}}{(}{)}{#1}[#2]%
}

\DeclareDocumentCommand{\rng}{m o}%
{%
    \functionTemplate{\mathrm{rng}}{(}{)}{#1}[#2]%
}

\DeclareDocumentCommand{\d}{o}%
{%
    \mathrm{d}\IfNoValueTF{#1}{}{^{#1}}%
}

\DeclareDocumentCommand{\set}{m m o}%
{%
    \mathOrText{\IfNoValueTF{#3}{\left}{#3}\{#1\ \IfNoValueTF{#3}{\left}{#3}\vert\ \vphantom{#1}#2\IfNoValueTF{#3}{\right.}{}\IfNoValueTF{#3}{\right}{#3}\}}%
}

%% file: packages_and_commands/additional_commands.tex
\crefname{observation}{observation}{observations}

\newcommand*{\vectorize}[1]{\mathOrText{\bm{#1}}}

\DeclareDocumentCommand{\bigOTilde}{m o}%
{%
	\functionTemplate{\widetilde{\mathrm{O}}}{(}{)}{#1}[#2]%
}

\newcommand*{\powerset}[1]{\mathOrText{2^{#1}}}

\newcommand*{\dtv}[2]{\mathOrText{d_{\text{tv}}}\left(#1, #2\right)}

\DeclareMathOperator*{\esssup}{ess\,sup}

\newcommand*{\compEvent}[1]{\mathOrText{\overline{#1}}}

\newcommand*{\concat}{\mathOrText{\circ}}

\newcommand*{\size}[1]{\mathOrText{\left\vert #1 \right\vert}}

\newcommand*{\absolute}[1]{\mathOrText{\left\vert #1 \right\vert}}

\newcommand*{\neighbors}[2]{\mathOrText{N_{#1}\left(#2\right)}}

\newcommand{\intD}{\mathOrText{\text{d}}}

\newcommand*{\tensor}{\mathOrText{\otimes}}

\newcommand*{\bigTensor}{\bigotimes}

\newcommand*{\zeroFunction}{\mathOrText{\bm{0}}}

\newcommand*{\error}{\mathOrText{\varepsilon}}

\newcommand*{\errorProb}{\mathOrText{\delta}}

\NewDocumentCommand{\EWrt}{m o O{} o}%
{%
	\probabilisticFunctionTemplate{\mathrm{E}\IfNoValueF{#2}{\mathOrText{_{#2}}}}{#1}[#3][#4]%
}

\NewDocumentCommand{\VarWrt}{m o O{} o}%
{%
	\probabilisticFunctionTemplate{\mathrm{Var}\IfNoValueF{#2}{\mathOrText{_{#2}}}}{#1}[#3][#4]%
}

\DeclareDocumentCommand{\innerProduct}{m m o}
{
	\mathOrText{\left\langle #1, #2 \right\rangle\IfNoValueF{#3}{_{#3}}}
}

\DeclareDocumentCommand{\norm}{m o}
{
	\mathOrText{ \left\lVert #1 \right\rVert \IfNoValueF{#2}{_{#2}}}
}

\DeclareDocumentCommand{\Lspace}{m o}
{
	\mathOrText{ L^{#1} \IfNoValueF{#2}{\left(#2\right)}}
}

\newcommand*{\graph}{\mathOrText{G}}
\newcommand*{\vertices}{\mathOrText{V}}
\newcommand*{\edges}{\mathOrText{E}}

\DeclareDocumentCommand{\degree}{o o}
{
	\mathOrText{ d \IfNoValueF{#1}{_{#1} \IfNoValueF{#2}{\left(#2\right)}}}
}

\newcommand*{\graphs}[1]{\mathOrText{\mathcal{G}_{#1}}}

\newcommand*{\vertexSpace}{\mathOrText{X}}
\newcommand*{\vertexSigmaAlgebra}{\mathOrText{\mathcal{A}}}
\newcommand*{\vertexDistribution}{\mathOrText{\xi}}
\newcommand*{\vertexProbabilitySpace}{\mathOrText{\mathcal{X}}}

\newcommand{\graphonSymbol}{\mathOrText{W}}
\DeclareDocumentCommand{\edgeProbability}{o o}
{
	\mathOrText{ \graphonSymbol \IfNoValueF{#2}{\left(#1, #2\right)}}
}

\DeclareDocumentCommand{\graphDistribution}{m m m o}
{
	\mathOrText{ \mathds{G}_{#3, #1} \IfNoValueF{#4}{\left(#4\right)}}
}

\newcommand*{\degreeError}{\mathOrText{\alpha}}

\newcommand*{\degreeErrorProb}{\mathOrText{q}}

\newcommand*{\countOnes}[1]{\mathOrText{\left\lvert#1\right\rvert_1}}

\newcommand*{\countEdges}[3]{\mathOrText{m_{#1}^{(#2)}\left(#3\right)}}

\DeclareDocumentCommand{\spinConfiguration}{o}
{
	\mathOrText{\sigma \IfNoValueF{#1}{\left(#1\right)}}
}

\newcommand*{\spinConfigurations}[1]{\mathOrText{\Sigma_{#1}}}

\newcommand*{\fugacitySymbol}{\mathOrText{\gamma}}

\DeclareDocumentCommand{\fugacity}{o}
{
	\mathOrText{ \fugacitySymbol \IfNoValueF{#1}{\left(#1\right)}}
}

\newcommand*{\criticalFugacity}[1]{\mathOrText{\fugacitySymbol_{\text{c}}\left(#1\right)}}

\newcommand*{\initialFugacity}{\mathOrText{\fugacitySymbol_0}}

\newcommand*{\edgeInteraction}{\mathOrText{\beta}}

\newcommand*{\GibbsDistributionSymbol}{\mathOrText{\mu}}

\DeclareDocumentCommand{\GibbsDistribution}{m m m m o}
{
	\mathOrText{\GibbsDistributionSymbol_{#1}^{\left(#2, #3, #4\right)} \IfNoValueF{#5}{\left(#5\right)}}
}

\newcommand*{\discretePartitionFunctionSymbol}{\mathOrText{Z}}

\DeclareDocumentCommand{\partitionFunction}{m o o}
{
	\mathOrText{ \discretePartitionFunctionSymbol_{#1} \IfNoValueF{#3}{\left(#2, #3\right)}}
}

\DeclareDocumentCommand{\hcPartitionFunction}{m o}
{
	\mathOrText{ \discretePartitionFunctionSymbol_{#1} \IfNoValueF{#2}{\left(#2\right)}}
}

\DeclareDocumentCommand{\hcGibbsDistribution}{m m o}
{
	\mathOrText{ \GibbsDistributionSymbol_{#1}^{\left(#2\right)} \IfNoValueF{#3}{\left(#3\right)}}
}

\newcommand*{\independentSets}[1]{\mathOrText{\mathcal{I}\left(#1\right)}}

\newcommand*{\configurationToSet}[1]{\mathOrText{S_{#1}}}

\DeclareDocumentCommand{\occupationRatio}{m m o o}
{
	\mathOrText{ R_{#1}^{\left(#2\right)} 
		\IfNoValueTF{#4}
		{
			\IfNoValueF{#3}{\left(#3\right)}
		}
		{
			\left(#3 \mid #4\right)
		}
	}
}

\newcommand{\probSpace}{\mathOrText{\Omega}}

\newcommand{\sigmaAlgebra}{\mathOrText{\mathcal{F}}}

\DeclareDocumentCommand{\probMeasureIdx}{o o}
{
	\mathOrText{ \mu \IfNoValueF{#1}{_{#1}} \IfNoValueF{#2}{\left(#2\right)}}
}

\DeclareDocumentCommand{\probMeasure}{o}
{
	\mathOrText{\probMeasureIdx[][#1]}
}

\DeclareDocumentCommand{\markovOperatorIdx}{o o o}
{
	\mathOrText{ P\IfNoValueF{#1}{_{#1}} \IfNoValueF{#3}{\left(#2, #3\right)}}
}

\DeclareDocumentCommand{\markovOperator}{o o}
{
	\mathOrText{\markovOperatorIdx[][#1][#2]}
}

\DeclareDocumentCommand{\dirichlet}{o o}
{
	\mathOrText{ \mathcal{E} \IfNoValueF{#2}{\left(#1, #2\right)}}
}

\newcommand*{\pointProcessSpace}{\mathOrText{\mathds{X}}}

\newcommand*{\BorelSymbol}{\mathOrText{\mathcal{B}}}

\DeclareDocumentCommand{\Borel}{o}
{
	\mathOrText{\BorelSymbol \IfNoValueF{#1}{\left(#1\right)}}
}

\DeclareDocumentCommand{\BorelOn}{m o}
{
	\mathOrText{\BorelSymbol_{#1} \IfNoValueF{#2}{\left(#2\right)}}
}

\DeclareDocumentCommand{\dist}{o o}
{
	\mathOrText{d \IfNoValueF{#2}{\left(#1, #2\right)}}
}

\newcommand*{\volumeMeasureSymbol}{\mathOrText{\nu}}

\DeclareDocumentCommand{\volumeMeasure}{o}
{
	\mathOrText{\volumeMeasureSymbol \IfNoValueF{#1}{\left(#1\right)}}
}

\DeclareDocumentCommand{\productVolumeMeasure}{m o}
{
	\mathOrText{\volumeMeasureSymbol^{#1} \IfNoValueF{#2}{\left(#2\right)}}
}

\DeclareDocumentCommand{\countingMeasures}{o}
{
	\mathOrText{\mathcal{N} \IfNoValueF{#1}{_{#1}}}
}

\DeclareDocumentCommand{\countingMeasure}{o}
{
	\mathOrText{\eta \IfNoValueF{#1}{\left(#1\right)}}
}

\newcommand*{\countingSigmaAlgebra}{\mathOrText{\mathcal{R}}}

\DeclareDocumentCommand{\countFunction}{m o}
{
	\mathOrText{N_{#1} \IfNoValueF{#2}{\left(#2\right)}}
}

\DeclareDocumentCommand{\pointSet}{o}
{
	\mathOrText{X \IfNoValueF{#1}{_{#1}}}
}

\DeclareDocumentCommand{\DiracMeasure}{m o}
{
	\mathOrText{\delta_{#1} \IfNoValueF{#2}{\left(#2\right)}}
}

\newcommand*{\PoissonIntensity}{\mathOrText{\kappa}}

\newcommand*{\region}{\mathOrText{\mathds{V}}}

\newcommand*{\gppFugacity}{\mathOrText{\lambda}}

\DeclareDocumentCommand{\hamiltonian}{o}
{
	\mathOrText{H \IfNoValueF{#1}{\left(#1\right)}}
}

\DeclareDocumentCommand{\potential}{o o}
{
	\mathOrText{\phi \IfNoValueF{#2}{\left(#1, #2\right)}}
}

\DeclareDocumentCommand{\gppPartitonFunction}{m o o}
{
	\mathOrText{\Xi_{#1} \IfNoValueF{#3}{\left(#2, #3\right)}}
}

\newcommand*{\pointProcess}{\mathOrText{P}}

\newcommand*{\GibbsPointProcessSymbol}{\mathOrText{P}}

\newcommand*{\GibbsPointProcess}[3]{\mathOrText{\GibbsPointProcessSymbol_{#1}^{\left(#2, #3\right)}}}

\DeclareDocumentCommand{\PoissonPointProcess}{o o}
{
	\mathOrText{Q \IfNoValueF{#1}{_{#1}} \IfNoValueF{#2}{\left(#2\right)}}
}

\newcommand*{\numPoints}{\mathOrText{k}}

\DeclareDocumentCommand{\canonicalDistribution}{m m m o}
{
	\mathOrText{ \zeta^{(#1)}_{#2, #3} \IfNoValueF{#4}{\left(#4\right)}}
}

\newcommand*{\canonicalProbabilitySpace}[1]{\mathOrText{ \mathcal{X}_{#1}}}

\DeclareDocumentCommand{\canonicalEdgeProbability}{m o o}
{
	\mathOrText{ \graphonSymbol_{#1} \IfNoValueF{#3}{\left(#2, #3\right)}}
}

\DeclareDocumentCommand{\generalizedTemperedness}{m o}
{
	\mathOrText{C_{#1} \IfNoValueF{#2}{\left(#2\right)}}
}

\newcommand*{\numSpaces}{\mathOrText{N}}

\newcommand*{\deviationSymbol}{\mathOrText{\Delta}}

\DeclareDocumentCommand{\deviation}{m m o o}
{
	\mathOrText{ \deviationSymbol_{#1}^{\left(#2\right)} \IfNoValueF{#4}{\left(#3, #4\right)}}
}

\DeclareDocumentCommand{\deviationFunc}{m o}
{
	\mathOrText{ \deviationSymbol_{#1} \IfNoValueF{#2}{\left(#2\right)}}
}

\newcommand*{\uniformProbabilitySpace}{\mathOrText{\mathcal{Y}}}

\newcommand*{\uniformDistributionSymbol}{\mathOrText{u}}

\DeclareDocumentCommand{\uniformDistribution}{o}
{
	\mathOrText{\uniformDistributionSymbol \IfNoValueF{#1}{\left(#1\right)}}
}

\DeclareDocumentCommand{\uniformDistributionOn}{m o}
{
	\mathOrText{\uniformDistributionSymbol_{#1} \IfNoValueF{#2}{\left(#2\right)}}
}

\DeclareDocumentCommand{\productUniformDistributionOn}{m m o}
{
	\mathOrText{\uniformDistributionSymbol^{#2}_{#1} \IfNoValueF{#3}{\left(#3\right)}}
}

\newcommand*{\timeSymbol}{\mathOrText{t}}

\newcommand*{\sampleGraphTime}[3]{\mathOrText{\timeSymbol_{#1, #2}\left(#3\right)}}

\DeclareDocumentCommand{\PapangelouIntensity}{o o o}
{
	\mathOrText{\zeta \IfNoValueF{#1}{_{#1}}\IfNoValueF{#3}{\left(#2, #3\right)}}
}

\newcommand*{\samplePointTime}[1]{\mathOrText{\timeSymbol_{#1}}}

\newcommand*{\evaluatePotentialTime}[1]{\mathOrText{\timeSymbol_{#1}}}

\newcommand*{\samplingError}{\mathOrText{\varepsilon}}

\newcommand*{\randomPoint}{\mathOrText{X}}

\newcommand*{\samplingGraph}{\mathOrText{G}}

\newcommand*{\sampledPointSet}{\mathOrText{X}}

\newcommand*{\modifiedSampledPointSet}{\mathOrText{Y}}

\DeclareDocumentCommand{\spinConfigurationModified}{o}
{
	\mathOrText{\tau \IfNoValueF{#1}{\left(#1\right)}}
}

\newcommand*{\modfiedSamplerOutput}[1]{\mathOrText{\widehat{P_{#1}}}}

\newcommand*{\approximationGraphs}[2]{\mathOrText{A^{\left(#1\right)}_{#2}}}

\DeclareDocumentCommand{\gppDensity}{o}
{
	\mathOrText{f\IfNoValueF{#1}{\left(#1\right)}}
}

\DeclareDocumentCommand{\modifiedSamplingDensity}{m o}
{
	\mathOrText{g_{#1} \IfNoValueF{#2}{\left(#2\right)}}
}

\DeclareDocumentCommand{\densityNormalizingFunction}{m m o}
{
	\mathOrText{\Psi^{\left(#2\right)}_{#1} \IfNoValueF{#3}{\left(#3\right)}}
}

\DeclareDocumentCommand{\countingMeasureToTuple}{o}
{
	\mathOrText{\varphi\IfNoValueF{#1}{\left(#1\right)}}
}

\newcommand{\kPWCC}[2]{\mathOrText{V_{#1}(#2)}}

\newcommand{\pwcc}[1]{\mathOrText{\Delta_{#1}}}

\newcommand{\rootVertex}{\mathOrText{r}}

\DeclareDocumentCommand{\paths}{o o}
{
	\mathOrText{\mathcal{P}
		\IfNoValueTF{#2}
		{
			\IfNoValueF{#1}{_{#1}}
		}{
			_{#1, #2}
		}
	}
}

\DeclareDocumentCommand{\sawTree}{o o}
{
	\mathOrText{T
		\IfNoValueTF{#2}
		{
			\IfNoValueF{#1}{_{#1}}
		}{
			_{#1, #2}
		}
	}
}

\DeclareDocumentCommand{\orderedPaths}{o o o}
{
	\mathOrText{\mathcal{P}
		\IfNoValueTF{#3}
		{
			\IfNoValueTF{#2}{
				_{#1}^{#2}
			}{
				\IfNoValueF{#1}{_{#1}}
			}
		}{
			_{#1, #3}^{#2}
		}
	}
}

\DeclareDocumentCommand{\WeitzTree}{o o o}
{
	\mathOrText{T
		\IfNoValueTF{#3}
		{
			\IfNoValueTF{#2}{
				_{#1}^{#2}
			}{
				\IfNoValueF{#1}{_{#1}}
			}
		}{
			_{#1, #3}^{#2}
		}
	}
}

\DeclareDocumentCommand{\layer}{m o o o}
{
	\mathOrText{L
		\IfNoValueTF{#4}
		{
			\IfNoValueTF{#3}{
				_{#2}^{#3}
			}{
				\IfNoValueF{#2}{_{#2}}
			}
		}{
			_{#2, #4}^{#3}
		}
		\left(#1\right)
	}
}

\newcommand{\neighborOrderings}{\mathOrText{F}}

\DeclareDocumentCommand{\neighborOrder}{m o}
{
	\mathOrText{f_{#1} \IfNoValueF{#2}{\left(#2\right)}}
}

\newcommand{\graphFamily}{\mathOrText{\mathcal{F}}}

\newcommand*{\connectiveConstant}{\mathOrText{\Delta}}

\newcommand*{\mixingConstant}{\mathOrText{\delta}}

%% file: content/abstract.tex
\begin{abstract}

We study computational aspects of repulsive Gibbs point processes, which are probabilistic models of interacting particles in a finite-volume region of space. We introduce an approach for reducing a Gibbs point process to the hard-core model, a well-studied discrete spin system. Given an instance of such a point process, our reduction generates a random graph drawn from a natural geometric model. We show that the partition function of a hard-core model on graphs generated by the geometric model concentrates around the partition function of the Gibbs point process. Our reduction allows us to use a broad range of algorithms developed for the hard-core model to sample from the Gibbs point process and approximate its partition function. This is, to the extend of our knowledge, the first approach that deals with pair potentials of unbounded range. We compare the resulting algorithms with recently established results and study further properties of the random geometric graphs with respect to the hard-core model.

\end{abstract}

%% file: content/introduction.tex
\section{Introduction}

Gibbs point processes are a tool for modelling a variety of phenomena that can be described as distributions of random spatial events \cite{baddeley2006case,moller2007modern}. Such phenomena include the location of stars in the universe, a sample of cells under the microscope, or the location of pores and cracks in the ground (see \cite{moller2003statistical,van2000markov} for more on applications of Gibbs point processes). In statistical physics, such point processes are frequently used as stochastic models for gases or liquids of interacting particles \cite{ruelle1999statistical}.

A Gibbs point process on a finite-volume region $\region$ is parameterized by a fugacity $\gppFugacity$ and a pair potential $\potential$ expressing the interactions between pairs of points. Every point configuration in the region is assigned a weight according to the pair interactions~$\potential$ between all pairs of points in the configuration. One can then think of a Gibbs point process as a Poisson point process of intensity $\gppFugacity$, where the density of each configuration is scaled proportionally to its weight. The density is normalized by the partition function, which is the integral of the weights over the configuration space (see \Cref{sec:GPP_definition} for a formal definition of the model). The most famous example of such a process is the \emph{hard-sphere model}, a model of a random packing of equal-sized spheres with radius $r$. The pair potential in the hard-sphere model defines hard-core interactions, i.e., configurations where two points closer than some distance $2r$ have weight zero, while all other configurations have weight one. In this article, we consider Gibbs point processes with \emph{repulsive} potentials, that is, pair potentials in which adding a point to a configuration does not increase its weight. The hard-sphere model, for example, does have a repulsive pair potential, however, we do not restrict ourselves to hard-core potentials and allow for soft-core interactions.

% Give Hard spheres as an example and informally define GPPs. Say what the partition function is. Talk about the focus on repulsive interactions.

The two most fundamental algorithmic tasks considered on Gibbs point processes are to sample from the model and to compute its partition function, which are closely related. Understanding for which potentials and fugacities these two tasks are tractable is an ambitious endeavour. Towards this goal, there has been a plethora of algorithmic results on Gibbs point processes spanning several decades. Notably, the Markov chain Monte Carlo method was developed for sampling an instance of the hard-sphere model with 224 particles \cite{Metropolis}. Since then, a variety of exact and approximate sampling algorithms for such point processes have been proposed in the literature, and their efficiency has been studied extensively both without \cite{garcia2000perfect,haggstrom1999characterization} and with rigorous running time guarantees \cite{moller1989rate,huber2012spatial,guo2021perfect,michelen2022strong,anand2023perfect}. The key objective of rigorous works is to identify a parameter regime for their respective model for which a randomized algorithm for sampling and approximating the partition function exists with running time polynomial in the volume of the region. In addition, deterministic algorithms for approximating the partition function have also appeared in the literature \cite{friedrich2021algorithms,jenssen2022quasipolynomial} with running time quasi-polynomial in the volume of the region.

This recent flurry of algorithmic results on Gibbs point processes can be attributed to progress in understanding the computational properties of discrete spin systems, such as the hard-core model. Within these works, two main approaches can be identified for transferring insights from discrete spin systems to Gibbs point processes. The first one, which includes results such as \cite{helmuth2022correlation,michelen2022strong,anand2023perfect}, considers properties proven to hold in discrete spin systems and translates them to the continuous setting of Gibbs point processes. More precisely, these works consider the notion of strong spatial mixing, which has been strongly connected to algorithmic properties of discrete spin systems~\cite{weitz2006counting,sinclair2017spatial,feng2022perfect}, and translate it to an analogous notion for Gibbs point processes to obtain algorithmic results. %\todo{I feel that discussing zero-freeness here diverges from the story.}
% While \cite{jenssen2022quasipolynomial} considers zero-freeness of the partition function, yet another property which has been shown to give fast algorithms in discrete spin systems \cite{barvinok2016combinatorics,patel2017deterministic}, to derive deterministic quasi-polynomial algorithms for Gibbs point processes.
A common pattern in these works is that once the parameter regime for which strong spatial mixing holds is established, one needs to prove from scratch that this implies efficient algorithms in Gibbs point processes. In addition, the definition of strong spatial mixing for Gibbs points processes assumes that the pair interactions of two particles is always of bounded range, i.e., if two particles are placed at distance greater than some constant $r \in \R_{\ge 0}$, they do not interact with each other.

The second approach, used in~\cite{friedrich2021spectral,friedrich2021algorithms}, is to discretize the model, i.e., reduce it to an instance of the hard-core model and then solve the respective algorithmic problem for the hard-core model. In this case, the algorithmic arsenal developed over the years for the hard-core model is now readily available for the instances resulting from this reduction. The main downside of these approaches is that they only apply to the hard-sphere model, a special case of bounded-range repulsive interactions.

\paragraph{Our contributions.} We introduce a natural approach for reducing repulsive point processes to the hard-core model. Given an instance $(\region,\gppFugacity,\potential)$ of such a point process, we generate a random graph by sampling $n\in\bigTheta{\region^2}$ point-vertices independently and uniformly at random in $\region$ and by connecting each pair of points with an edge drawn with an appropriate probability, which depends on $\potential$. We show that computational properties of the hard-core model on graphs generated by this model and with an appropriately scaled fugacity transfer to the originating Gibbs point process.

We first show that the partition function of the hard-core model on these graphs concentrates around the partition function of the Gibbs point process. Using existing algorithms for the hard-core model as a black box, our result immediately yields randomized approximation algorithms for the partition function of the point process in running time polynomial in the volume of $\region$.
Furthermore, we show that sampling an independent set from the generated hard-core model and returning the positions of its vertices in $\region$ results in an approximate sampler from the distribution of the Gibbs point process. Our approach, in contrast to all previous algorithmic work in the literature, does not require the pair potential $\potential$ of the point process to be of bounded range. This includes various models of interest in statistical physics, such as the (hard-core) Yukawa model~\cite{el2000line,rowlinson1989yukawa}, the Gaussian overlap model~\cite{berne1972gaussian}, the generalized exponential model~\cite{bacher2014explaining}, and the Yoshida--Kamakura model~\cite{yoshida1974liquid}.

Finally, we identify the parameter regime (in terms of $\gppFugacity$ and $\potential$) for which the generated hard-core instance exhibits strong spatial mixing. This parameter regime is identical to the best known regime for which repulsive point processes with bounded-range potentials have strong spatial mixing \cite{michelen2022strong}.

\subsection{Gibbs point processes}
\label{sec:GPP_definition}
We now formally define the notion of a Gibbs point process.
As usual in the theory of point processes, we assume the underlying space to be a complete, separable metric space $(\pointProcessSpace, \dist)$ equipped with the Borel algebra $\Borel = \Borel[\pointProcessSpace]$ and a reference measure $\volumeMeasure$ on $(\pointProcessSpace, \Borel)$ that assigns finite volume to bounded measurable sets.
In this work, we study Gibbs point processes $\GibbsPointProcess{\region}{\gppFugacity}{\potential}$ on bounded measurable regions $\region \subseteq \pointProcessSpace$ that are parameterized by a \emph{fugacity} parameter $\gppFugacity \in \R_{\ge 0}$ and a repulsive (i.e., non-negative), symmetric, measurable \emph{potential function} $\potential\colon \pointProcessSpace^2 \to \R_{\ge 0} \cup \{\infty\}$.
Such a process $\GibbsPointProcess{\region}{\gppFugacity}{\potential}$ is defined by a density with respect to a Poisson point process $\PoissonPointProcess[\gppFugacity]$ with intensity $\gppFugacity$ on $\pointProcessSpace$.
For every finite point configuration $\vectorize{x}=(x_1, \dots, x_\numPoints) \in \region^{\numPoints}$, this density is proportional to $\eulerE^{- \hamiltonian[x_1, \dots, x_\numPoints]}$, where $\hamiltonian$ is the \emph{Hamiltonian}
\[
	\hamiltonian[x_1 \dots, x_\numPoints] = \sum\nolimits_{\{i, j\} \in \binom{[\numPoints]}{2}} \potential[x_i][x_j] .
\]
More precisely, the density can be expressed explicitly as
\[
	\frac{\intD \GibbsPointProcess{\region}{\gppFugacity}{\potential}}{\intD \PoissonPointProcess[\gppFugacity]} (x_1, \dots, x_\numPoints) = \frac{\ind{\forall i \in [\numPoints] \colon x_i \in \region} \cdot \eulerE^{-\hamiltonian[x_1, \dots, x_k]}  \eulerE^{\gppFugacity \volumeMeasure[\region]}}{\gppPartitonFunction{\region}[\gppFugacity][\potential]} ,
\]
where the normalizing constant $\gppPartitonFunction{\region}[\gppFugacity][\potential]$ is the \emph{partition function}
\begin{align*}
	\gppPartitonFunction{\region}[\gppFugacity][\potential]
	&= 1 + \sum_{\numPoints \in \N_{\ge 1}} \frac{\gppFugacity^{\numPoints}}{\numPoints!} \int_{\region^{\numPoints}} \eulerE^{- \hamiltonian[x_1, \dots, x_\numPoints]} \productVolumeMeasure{\numPoints}[\intD \vectorize{x}] .
\end{align*}
% The main algorithmic tasks related to such Gibbs point processes are to sample point configurations from the distribution $\GibbsPointProcess{\region}{\gppFugacity}{\potential}$ and to compute $\gppPartitonFunction{\region}[\gppFugacity][\potential]$.

% In the algorithmic literature, $\region$ is usually considered to be a box-shaped region of $d$-dimensional Euclidean space (i.e., $\pointProcessSpace = \R^{d}$ and $\region = [0, \ell]^d$ for some $\ell \in \R_{>0}$ and $d \in \N_{\ge 1}$).
% Moreover, the majority of rigorous computational results that state running time guarantees focus on the hard-sphere model.
% This model results from setting $\potential[x_1][x_2] = \infty$ whenever the distance $\dist[x_1][x_2]$ is less than some constant $r \in \R_{\ge 0}$, and $\potential[x_1][x_2] = 0$ otherwise.
% In this setting, the best known parameter regime for exact sampling was given by Guo and Jerrum~\cite{guo2021perfect}.
% Using a partial rejection sampler, they achieve near linear running time in the volume $\volumeMeasure[\region]$ for all fugacities $\gppFugacity < \frac{1}{\sqrt{2}\generalizedTemperedness{\potential}}$.
% To the best of our knowledge, no results are known for the exact computation of $\gppPartitonFunction{\region}[\gppFugacity][\potential]$ for any non-trivial model.

\subsection{Randomized reduction to the hard-core model}
\label{subsec:discretization}
Our approach is to reduce the problem of sampling from a repulsive Gibbs point process and approximating its partition function to the analogous problems for a discrete hard-core model, which we briefly introduce.
For an undirected graph $\graph=(\vertices, \edges)$, let $\independentSets{\graph} \subseteq \powerset{\vertices}$ denote the set of independent sets of $\graph$.
For a vertex activity $\fugacity \in \R_{\ge 0}$, the \emph{hard-core model} on $\graph$ is a probability distribution~$\hcGibbsDistribution{\graph}{\fugacity}$ on~$\independentSets{\graph}$ that assigns each independent set $I \in \independentSets{\graph}$ a probability proportional to $\fugacity^{\size{I}}$.
The normalizing constant of this distribution, $\hcPartitionFunction{\graph}[\fugacity] = \sum_{I \in \independentSets{\graph}} \fugacity^{\size{I}}$, is called the \emph{hard-core partition function} on $\graph$.

The goal is to reduce the problem of approximate sampling from $\GibbsPointProcess{\region}{\gppFugacity}{\potential}$ to approximate sampling from $\hcGibbsDistribution{\graph}{\fugacity}$ and, similarly, to reduce the problem of approximating $\gppPartitonFunction{\region}[\gppFugacity][\potential]$ to approximating $\hcPartitionFunction{\graph}[\fugacity]$ for a suitably chosen graph $\graph$ and vertex activity $\fugacity$.
The advantage of this approach is that sampling from a hard-core model as well as approximating hard-core partition functions are well studied problems.
Specifically, a sequence of recent papers \cite{anari2021entropic,anari2021spectral,chen2022rapid,chen2020rapid,chen2021optimal} established approximate sampling from $\hcGibbsDistribution{\graph}{\fugacity}$ in $\bigOTilde{\size{\vertices}}$ running time and randomized approximation of $\hcPartitionFunction{\graph}[\fugacity]$ in $\bigOTilde{\size{\vertices}^2}[\big]$ running time for graphs~$\graph$ with maximum degree $\degree$ for all $\fugacity$ strictly below the \emph{tree threshold} $\criticalFugacity{\degree} \coloneqq \frac{\left(\degree - 1\right)^{\degree - 1}}{\left(\degree - 2\right)^{\degree}}$.

Our reduction is inspired by the discretization schemes in \cite{friedrich2021algorithms,friedrich2021spectral}.
These approaches are limited to the hard-sphere model (and similar models with hard-core interactions) in specific regions of Euclidean space.
In this setting, the utilized graph $\graph$ is essentially a unit-disk graph in $\region$.
This procedure comes with two major disadvantages.
Firstly, the analysis heavily depends geometric arguments and is therefore restricted to regions in Euclidean space that satisfy various requirements.
Secondly, it is not obvious how this technique extends to general repulsive potentials $\potential$ and especially how to account for soft-core interactions.

We circumvent the above problems by investigating hard-core models on a suitably chosen family of random graphs.
For a bounded measurable region $\region \subseteq \pointProcessSpace$, let $\uniformDistributionOn{\region}$ denote the uniform distribution on $\region$.
That is, for $x \sim \uniformDistributionOn{\region}$, we have $\Pr{x \in A} = \frac{\volumeMeasure[A]}{\volumeMeasure[\region]}$ for every measurable $A \subseteq \region$, and $\Pr{x \notin \region} = 0$.
For a repulsive potential $\potential$ and a positive integer $n \in \N_{\ge 1}$, we consider a random-graph model $\canonicalDistribution{n}{\region}{\potential}$ on the set of undirected graphs with vertex set $[n]$, where $\canonicalDistribution{n}{\region}{\potential}$ is defined by the following natural procedure to generate a graph:
\begin{enumerate}
	\item For each $i \in [n]$, draw a uniform random point $x_i \sim \uniformDistributionOn{\region}$ independently.
	\item For all $i, j \in [n]$ with $i \neq j$, connect $i$ and $j$ with an edge with probability $1 - \eulerE^{- \potential[x_i][x_j]}$ independently.
\end{enumerate}
Readers familiar with graphons might notice that this random-graph model can be expressed as a graphon-based random graph ($\edgeProbability$-random graph) for a suitably chosen graphon $\edgeProbability$.
We discuss this perspective later in the introduction.
Moreover, we would like to mentioned that a similar graph construction based on points from a Poisson point process was used in \cite{betsch2021uniqueness} to prove uniqueness of the infinite-volume Gibbs measure for $\gppFugacity < \frac{1}{\generalizedTemperedness{\potential}}$ via percolation.

The key property of the graphs from $\canonicalDistribution{n}{\region}{\potential}$ is that, for a suitably chosen vertex activity $\fugacity$, their hard-core partition functions concentrate around $\gppPartitonFunction{\region}[\gppFugacity][\potential]$.
This property is at the core of our reduction.

\begin{restatable}{theorem}{gppConcentration}
	\label{thm:gpp_concentration}
	Let $(\pointProcessSpace, \dist)$ be a complete, separable metric space, let $\Borel = \Borel[\pointProcessSpace]$ be the Borel algebra, and let $\volumeMeasure$ be a locally finite reference measure on $(\pointProcessSpace, \Borel)$.
	Let $\region \subseteq \pointProcessSpace$ be bounded and measurable, let $\gppFugacity \in \R_{\ge 0}$, and let $\potential\colon \pointProcessSpace^2 \to \R_{\ge 0} \cup \{\infty\}$ be a symmetric repulsive potential.
	For all $\error \in (0, 1]$, $\errorProb \in (0, 1]$ and $n \ge 4 \error^{-2} \errorProb^{-1} \max\left\{\eulerE^6 \gppFugacity^2 \volumeMeasure[\region]^2, \ln\left(4 \error^{-1}\right)^2\right\}$, it holds that, for $\graph \sim \canonicalDistribution{n}{\region}{\potential}$,
	\[
		\Pr{\absolute{\hcPartitionFunction{\graph}[\frac{\gppFugacity \volumeMeasure[\region]}{n}] - \gppPartitonFunction{\region}[\gppFugacity][\potential]} \ge \error  \gppPartitonFunction{\region}[\gppFugacity][\potential]} \le \errorProb .
		\qedhere
	\]
\end{restatable}

Informally, \Cref{thm:gpp_concentration} says that, for $n \in \bigTheta{\volumeMeasure[\region]^2}[\big]$ and $\graph \sim \canonicalDistribution{n}{\region}{\potential}$, the hard-core partition function $\hcPartitionFunction{\graph}[\fugacity[n]]$ with $\fugacity[n] = \frac{\gppFugacity \volumeMeasure[\region]}{n}$ is strongly concentrated around the partition function of the repulsive Gibbs point process $\gppPartitonFunction{\region}[\gppFugacity][\potential]$.
In \cite[Proposition $5.8$]{friedrich2021algorithms}, it was argued that the partition function of an unrestricted Poisson point process in a bounded measurable region $\region$ of Euclidean space cannot be approximated by the hard-core partition function $\hcPartitionFunction{\graph}[\fugacity[n]]$ for any graph $\graph$ on $n$ vertices if $n \in \smallO{\volumeMeasure[\region]^2}[\big]$.
As the unrestricted Poisson point process is a special case of a repulsive Gibbs point process with constant zero potential, this implies that our concentration result in \Cref{thm:gpp_concentration} is tight in terms of its asymptotic dependency on the volume $\volumeMeasure[\region]$.

%We note that in the setting of hard-constraint models such as the hard-sphere model in Euclidean space, the idea of discretizing based on geometric random graphs was already studied in \cite{friedrich2021algorithms} with the goal to allow for more general regions $\region$.
%However, due to the geometric arguments that were used in their proofs, their results relied on $\region$ to exhibit a certain nice partitioning, which itself is non-trivial to check.
%In contrast to that, we prove our more general concentration result in a far less \emph{ad hoc} manner.
A remarkable aspect of \Cref{thm:gpp_concentration} is the generality in which it holds with respect to the underlying space $\pointProcessSpace$.
In particular, we do not need any additional assumptions regarding its geometry.
We achieve this by deriving \Cref{thm:gpp_concentration} from a corollary of the Efron--Stein inequality \cite{efron1981jackknife}.
This corollary gives a convenient-to-use way for proving concentration of functions of independent random inputs, given that changing one input of the function only leads to small relative changes of its output.
Using this corollary, we prove \Cref{thm:gpp_concentration} by bounding changes in the hard-core partition function of a graph that are caused by small alternations of the graph structure.
We proceed by discussing this approach in detail.

\subsubsection*{Proving concentration}
We prove \Cref{thm:gpp_concentration} in two steps.
First, we show that, for $\graph \sim \canonicalDistribution{n}{\region}{\potential}$ and $\fugacity[n] = \frac{\gppFugacity \volumeMeasure[\region]}{n}$, the expected hard-core partition function $\E{\hcPartitionFunction{\graph}[\fugacity[n]]}$ converges rapidly to the partition function of the point process $\gppPartitonFunction{\region}[\gppFugacity][\potential]$ as~$n$ grows.
We prove this by fixing $n \in \bigTheta{\volumeMeasure[\region]^2}[\big]$ sufficiently large and rewriting $\E{\hcPartitionFunction{\graph}[\fugacity[n]]}$ as a sum of expectations.
We then relate each of the first $m < n$ terms in that sum with the corresponding term in $\gppPartitonFunction{\region}[\gppFugacity][\potential]$, which requires comparing the expected number of independent sets of size $k \le m$ in a graph $\graph \sim \canonicalDistribution{n}{\region}{\potential}$ with $\frac{n^k}{k! \volumeMeasure[\region]^k} \int_{\region^k} \eulerE^{- \hamiltonian[\vectorize{x}]} \intD \vectorize{x}$.
Finally, we argue that all terms of order higher than $m$ can be discarded.

Once convergence of the expectation is established, it remains to prove that the distribution of the partition functions $\hcPartitionFunction{\graph}[\fugacity[n]]$ for $\graph \sim \canonicalDistribution{n}{\region}{\potential}$ concentrates around this expectation.
To prove the latter, we derive \Cref{thm:selfbounded_concentration_simplified} from the Efron--Stein inequality \cite{efron1981jackknife}. This corollary roughly states that the output of a function $f$ on a product of probability spaces concentrates around its expectation if $f$ exhibits sufficiently small \emph{relative} changes when any component of its input is changed. Similar methods for proving concentration usually require the output of the function $f$ to exhibit small \emph{absolute} changes (i.e., $f$ to be Lipschitz, see \cite{mcdiarmid1989method,mcdiarmid1998concentration}), which does not hold in our setting.

To apply this concentration bound, we need to express the partition function of a graph drawn from~$\canonicalDistribution{n}{\region}{\potential}$ as a function of independent random inputs.
To this end, we model a random graph $\graph \sim \canonicalDistribution{n}{\region}{\potential}$ based on $n$ points $\vectorize{x} = (x_i)_{i \in [n]}$, each independently drawn from~$\uniformDistributionOn{\region}$, and $\frac{n(n-1)}{2}$ independent random variables $\vectorize{y}=(y_{i, j})_{1 \le i < j \le n}$, each uniformly distributed on the real interval $[0, 1]$.
Given the random vectors $\vectorize{x}$ and $\vectorize{y}$, we construct a graph by connecting vertices $i < j$ by an edge if and only if $y_{i, j} \le 1 - \eulerE^{-\potential[x_i][x_j]}$.
Note that the resulting graph is distributed according to $\canonicalDistribution{n}{\region}{\potential}$.
Thus, we express the hard-core partition function on the random-graph model $\canonicalDistribution{n}{\region}{\potential}$ as a function $f(\vectorize{x}, \vectorize{y})$ for $\vectorize{x}$ and $\vectorize{y}$ as described above.
The effect of changing a component of $\vectorize{y}$ is bounded by the relative change of the hard-core partition function when adding or removing an edge.
On the other hand, the effect of changing a component of $\vectorize{x}$, say $x_i$, is bounded by considering the change of the hard-core partition function when altering the neighborhood of a single vertex~$i$.
Bounding both effects and applying \Cref{thm:selfbounded_concentration_simplified} yields the desired concentration result (\Cref{thm:gpp_concentration}).

In fact, a similar argument as above applies to a broad class of antiferromagnetic spin systems on graphon-based random graphs that contains our application as a special case.
We discuss his more general setting after demonstrating the sampling and approximation results for Gibbs point processes that can be obtained from \Cref{thm:gpp_concentration}.

\subsection{Algorithmic implications}
\label{subsec:approximation_sampling}
We proceed by showcasing some algorithmic results for repulsive Gibbs point processes that follow from \Cref{thm:gpp_concentration}.
More specifically, we focus on $\error$-approximate sampling from the point process and obtaining a randomized $\error$-approximation of the partition function.
Formally, the problem of $\error$\emph{-approximate sampling} from $\GibbsPointProcess{\region}{\gppFugacity}{\potential}$ is defined as producing a random point configuration with a distribution that has a total variation distance of at most~$\error$ to~$\GibbsPointProcess{\region}{\gppFugacity}{\potential}$.
Analogously, the problem of $\error$\emph{-approximating} $\gppPartitonFunction{\region}[\gppFugacity][\potential]$ is defined as computing some value $x \in \R$ such that $(1-\error)\gppPartitonFunction{\region}[\gppFugacity][\potential] \le x \le (1+\error)\gppPartitonFunction{\region}[\gppFugacity][\potential]$.
Moreover, an algorithm is called a \emph{randomized} $\error$\emph{-approximation} if it outputs an $\error$-approximation of~$\gppPartitonFunction{\region}[\gppFugacity][\potential]$ with probability at least $\frac{2}{3}$.
The choice of the constant $\frac{2}{3}$ is rather arbitrary here, as the error probability can be made smaller than every $\errorProb \in \R_{>0}$ by taking the median of $\bigO{\log\left(\errorProb^{-1}\right)}$ independent runs, as long as the error probability of each run is some constant smaller than~$\frac{1}{2}$.
Furthermore, we consider an $\error$-approximate sampler and a (randomized) $\error$-approximation algorithm as efficient if their running time is polynomial in the volume $\volumeMeasure[\region]$ and in $\error^{-1}$.

Recent rigorous results establish bounds on the fugacity regime of different models for which these algorithmic problems can be solved efficiently.
Often, these bounds are stated in terms of the \emph{temperedness constant} $\generalizedTemperedness{\potential}$, which is defined as
\[
	\generalizedTemperedness{\potential} = \sup_{x_1 \in \pointProcessSpace} \int_{\pointProcessSpace} \absolute{1 - \eulerE^{-\potential[x_1][x_2]}} \volumeMeasure[\intD x_2].
\]
This value can be seen as measure for the strength of interactions between points.

Given our concentration result (\Cref{thm:gpp_concentration}), a straightforward idea for approximating~$\gppPartitonFunction{\region}[\gppFugacity][\potential]$ is to sample a graph $\graph \sim \canonicalDistribution{n}{\region}{\potential}$ and try to approximate its hard-core partition function.
A refined version of this procedure leads to the following theorem.

\begin{restatable}{theorem}{approximateGpp}
	\label{thm:approximate_gpp}
	Let $(\pointProcessSpace, \dist)$ be a complete, separable metric space, let $\Borel = \Borel[\pointProcessSpace]$ be the Borel algebra, and let~$\volumeMeasure$ be a locally finite reference measure on $(\pointProcessSpace, \Borel)$.
	Let $\region \subseteq \pointProcessSpace$ be bounded and measurable, let $\gppFugacity \in \R_{\ge 0}$, and let $\potential\colon \pointProcessSpace^2 \to \R_{\ge 0} \cup \{\infty\}$ be a symmetric repulsive potential.
	Assume there is a sampler for $\canonicalDistribution{n}{\region}{\potential}$ with running time $\sampleGraphTime{\region}{\potential}{n}$.

	If $\gppFugacity < \frac{\eulerE}{\generalizedTemperedness{\potential}}$,
	then, for all $\error \in (0, 1]$, there is a randomized $\error$-approximation algorithm for $\gppPartitonFunction{\region}[\gppFugacity][\potential]$ with running time in $\bigOTilde{\volumeMeasure[\region]^{4} \error^{-6}}[\big] + \sampleGraphTime{\region}{\potential}{\bigOTilde{\volumeMeasure[\region]^2 \error^{-2}}[\big]}$.
\end{restatable}

With respect to sampling from $\GibbsPointProcess{\region}{\gppFugacity}{\potential}$, it is less obvious how \Cref{thm:gpp_concentration} can be utilized.
However, under mild assumptions, we obtain an approximate sampler, based on \Cref{thm:gpp_concentration}, by the following procedure: Sample an independent set $I \in \independentSets{\graph}$ (approximately) from~$\hcGibbsDistribution{\graph}{\fugacity[n]}$ and output the point configuration $\{x_i\}_{i \in I}$.
Given that $\GibbsPointProcess{\region}{\gppFugacity}{\potential}$ is simple, which means that drawing a point configuration that contains the same point multiple times has probability zero, a refined version of the approach sketched above leads to the following result.

\begin{theorem}
	\label{thm:sampling_simplified}
	Let $(\pointProcessSpace, \dist)$ be a complete, separable metric space, let $\Borel = \Borel[\pointProcessSpace]$ be the Borel algebra, and let $\volumeMeasure$ be a locally finite reference measure on $(\pointProcessSpace, \Borel)$.
	Let $\region \subseteq \pointProcessSpace$ be bounded and measurable, let $\gppFugacity \in \R_{\ge 0}$, and let $\potential\colon \pointProcessSpace^2 \to \R_{\ge 0} \cup \{\infty\}$ be a symmetric repulsive potential.
	Assume we can sample from the uniform distribution $\uniformDistributionOn{\region}$ in time $\samplePointTime{\region}$ and, for every $x, y \in \region$, we can evaluate $\potential[x][y]$ in time $\evaluatePotentialTime{\potential}$.

	If the Gibbs point process $\GibbsPointProcess{\region}{\gppFugacity}{\potential}$ is simple and $\gppFugacity < \frac{\eulerE}{\generalizedTemperedness{\potential}}$,
	then, for every $\samplingError \in \R_{>0}$, there exists an $\samplingError$-approximate sampling algorithm for $\GibbsPointProcess{\region}{\gppFugacity}{\potential}$ with running time in $\bigOTilde{\volumeMeasure[\region]^2 \samplingError^{-4} + \volumeMeasure[\region]^2 \samplingError^{-3} \samplePointTime{\region} + \volumeMeasure[\region]^4 \samplingError^{-6} \evaluatePotentialTime{\potential}}[\big]$.
\end{theorem}

There are two main differences in the assumptions of the approximation result (\Cref{thm:approximate_gpp}) and the sampling result (\Cref{thm:sampling_simplified}).
First, the sampling result requires the point process to be simple.
The reason is that, in order to bound the total variation distance between the output of our sampler and $\GibbsPointProcess{\region}{\gppFugacity}{\potential}$, we derive a density of that output with respect to a Poisson point process.
This task is greatly simplified by assuming that $\GibbsPointProcess{\region}{\gppFugacity}{\potential}$ is simple, as it allows for an easier characterization of the output distribution of our sampling, based on a theorem by R\'{e}nyi--Mönch (see \cite[Theorem 9.2.XII]{daley2008introduction}).
However, assuming the point process to be simple is only a minor restriction, as it is satisfied for most applications of point processes.
For example, it is trivially satisfied if the reference volume measure $\volumeMeasure$ is not-atomic (i.e., assigns volume~$0$ to single points).
This includes the most frequently studied case of Gibbs point processes in Euclidean space but also a variety of other spaces, such as Gibbs point processes in hyperbolic spaces or in Riemannian manifolds.

Second, our sampling result requires efficient sampling from the uniform distribution~$\uniformDistributionOn{\region}$ and an efficient way to compute the potential $\potential$.
In contrast to that, \Cref{thm:approximate_gpp} only assumes an efficient way to sample a graph from $\canonicalDistribution{n}{\region}{\potential}$.
We state the theorems in this way to emphasize that, for approximating $\gppPartitonFunction{\region}[\gppFugacity][\potential]$, we only need to sample from the random-graph model $\canonicalDistribution{n}{\region}{\potential}$.
Our sampling procedure additionally requires the position $x_i \in \region$ for each vertex $i \in [n]$ along with the graph to output the point configuration, associated to a random independent set drawn from the hard-core model.
% Since sampling from $\uniformDistributionOn{\region}$ and evaluating $\potential$ can be used to construct a random graph from $\canonicalDistribution{n}{\region}{\potential}$, we also have the following corollary of \Cref{thm:approximate_gpp}. \todo{Remove corollary here.}
% \begin{restatable}{corollary}{approximateGppSimplified}
% 	\label{cor:approximate_gpp_simplified}
% 	Let $(\pointProcessSpace, \dist)$ be a complete, separable metric space, let $\Borel = \Borel[\pointProcessSpace]$ be the Borel algebra, and let $\volumeMeasure$ be a locally finite reference measure on $(\pointProcessSpace, \Borel)$.
% 	Let $\region \subseteq \pointProcessSpace$ be bounded and measurable, let $\gppFugacity \in \R_{\ge 0}$, and let $\potential\colon \pointProcessSpace^2 \to \R_{\ge 0} \cup \{\infty\}$ be a symmetric repulsive potential.
% 	Assume we can sample from the uniform distribution $\uniformDistributionOn{\region}$ in time $\samplePointTime{\region}$ and, for every $x, y \in \region$, we can evaluate $\potential[x][y]$ in time $\evaluatePotentialTime{\potential}$.
%
% 	If $\gppFugacity < \frac{\eulerE}{\generalizedTemperedness{\potential}}$,
% 	then, for all $\error \in (0, 1]$, there is a randomized $\error$-approximation algorithm for~$\gppPartitonFunction{\region}[\gppFugacity][\potential]$ with running time in $\bigOTilde{\volumeMeasure[\region]^{4} \error^{-6} +  \volumeMeasure[\region]^2 \samplingError^{-2} \samplePointTime{\region} + \volumeMeasure[\region]^4 \samplingError^{-4} \evaluatePotentialTime{\potential}}[\big]$.
% \end{restatable}

Last, we briefly discuss the origin of the fugacity bound $\frac{\eulerE}{\generalizedTemperedness{\potential}}$ in our algorithmic results.
Write $\fugacity[n] = \frac{\gppFugacity \volumeMeasure[\region]}{n}$ for every $n \in \N_{\ge 1}$.
Note that our algorithms rely on either an efficient approximation of the hard-core partition function $\hcPartitionFunction{\graph}[\fugacity[n]]$ or an efficient approximate sampler for an independent set from $\hcGibbsDistribution{\graph}{\fugacity[n]}$ for a random graph $\graph \sim \canonicalDistribution{n}{\region}{\potential}$.
As discussed earlier, such computational results are known for general graphs of maximum degree $\degree$ as long as the parameter $\fugacity$ is below the corresponding tree threshold $\criticalFugacity{\degree}$.
Observe that $\criticalFugacity{\degree} \approx \frac{\eulerE}{\degree}$ for large $\degree$.
Thus, roughly speaking, we can perform the necessary computational tasks as long as $\fugacity[n] = \frac{\gppFugacity \volumeMeasure[\region]}{n} < \frac{\eulerE}{\degree[\graph]}$, where $\degree[\graph]$ is the maximum degree of the graph $\graph$ that was drawn from $\canonicalDistribution{n}{\region}{\potential}$.
Equivalently, this is $\gppFugacity < \frac{\eulerE n}{\degree[\graph] \volumeMeasure[\region]}$.
The main observation is now that, for $\graph \sim \canonicalDistribution{n}{\region}{\potential}$, the expected degree of an arbitrary vertex of $\graph$ is upper-bounded by $\frac{n \generalizedTemperedness{\potential}}{\volumeMeasure[\region]}$.
By proving that, with sufficiently high probability, the maximum degree $\degree[\graph]$ is not much larger than this value, we obtain the desired bound of $\frac{\eulerE}{\generalizedTemperedness{\potential}}$.
%
%We would like to point out a surprising detail about the argument above.
%Even though the resulting hard-core model is below the tree threshold, which means in particular that it exhibits strong spatial mixing (see \cite{weitz2006counting}), the same is not necessarily true for the original Gibbs point process.
%In particular, without assuming finite range, we can easily construct tempered potentials (i.e., $\generalizedTemperedness{\potential} < \infty$) such that the direct influence between points does not decay at an exponential rate.
%Such models do not exhibit strong spatial mixing in the sense of \cite{michelen2022strong} for any non-trivial fugacity $\gppFugacity > 0$, which was crucial for previous randomized algorithms.
%However, the argument above guarantees that the discretization does in fact exhibit strong spatial mixing with high probability if $\gppFugacity < \frac{\eulerE}{\generalizedTemperedness{\potential}}$.
%The key to this is that discretizing changes the underlying metric, meaning that the graph distance between two vertices can be fundamentally different from their original distance in the space $\pointProcessSpace$.
% \todo{change this transition}
% In fact, we can show that strong spatial mixing for $\graph \sim \canonicalDistribution{n}{\region}{\potential}$ holds for a broader parameter regime, as we discuss in the following section.

\subsection{Strong spatial mixing, connective constants, and improved bounds}

A property closely tied to the existence of efficient algorithms for the hard-core model is strong spatial mixing. Strong spatial mixing describes a particular way how dependencies between distant vertices in a graph $\graph$ decay. The definition is easiest stated in terms of the occupation probability of a vertex $v \in \vertices_{\graph}$ (i.e., the probability that $v$ is in the independent set, drawn from a hard-core model on $\graph$) conditioned on certain vertices being occupied or unoccupied. Given two such conditions that differ at some vertex set $S \subset \vertices_{\graph}$, strong spatial mixing requires that the resulting difference in the occupation probability of every vertex $v$ is exponentially small in the graph distance between $v$ and $S$\footnote{Often it is more convenient to work with the occupation ratio, which is the occupation probability divided by the probability of the vertex to be unoccupied. However, the resulting strong spatial mixing definitions are equivalent.} (see \Cref{def:ssm} and remark \Cref{remark:ssm} for more details).

In a seminal paper, Weitz \cite{weitz2006counting} proved that, for a graph $\graph$ of maximum degree $\degree_\graph$ and vertex $v \in \vertices_{\graph}$, we can construct a tree with root $v$ such that the occupation probability of $v$ in the tree is the same as in $\graph$.
This continues to hold when conditioning on the state of other vertices by translating the condition to the tree appropriately. It follows that if this tree exhibits strong spatial mixing with respect to the root, then this property also holds for $\graph$. Moreover, in this case, a recursive computation on this tree can be used to approximate the occupation probability of $v$, which results in a sampling and approximation algorithm with running time $\size{\vertices_\graph}^{\bigO{\log(\degree_\graph)}}$, for the hard-core model at vertex activity up to $\criticalFugacity{\degree_\graph}$ on $\graph$.

Subsequently, this result was improved by a more refined analysis \cite{sinclair2013spatial,sinclair2017spatial}. We elaborate further. The tree used by Weitz, which we refer to as the \emph{Weitz tree}, is a truncated version of a self-avoiding walk tree\footnote{In \cite{weitz2006counting} the tree is actually not truncated but certain vertices in the tree are fixed to be always occupied or unoccupied. However, this is equivalent to truncating the tree.} (see \Cref{sec:saw_tree} for a formal definition). This truncation accounts for the effect of cycles on the hard-core distribution in the original graph. In \cite{sinclair2013spatial,sinclair2017spatial}, it was shown that strong spatial mixing results can be derived by bounding the connective constant $\connectiveConstant$, which describes the growth of the Weitz tree (see \Cref{def:connective_constant}). In particular, it was shown that strong spatial mixing applies up to a vertex activity of $\fugacity < \criticalFugacity{\connectiveConstant}$, which improves bounds derived from the maximum degree. Consequentially, this implies sampling and approximation algorithm for the hard-core model for this parameter regime with running time $\size{\vertices_\graph}^{\bigO{\log(\connectiveConstant)}}$.

Inspired by these results on the hard-core model, Michelen and Perkins \cite{michelen2021potential} recently introduced the potential-weighted connective constant $\pwcc{\potential}$ for repulsive Gibbs point processes (see \Cref{sec:pwcc} for a formal definition).
It can be seen as an alternative to $\generalizedTemperedness{\potential}$ that is more sensitive to the structure of the underlying space $\pointProcessSpace$.
In particular, for any non-trivial potential, $\pwcc{\potential}$ is strictly smaller than the temperedness constant.
Moreover, it was shown in \cite{michelen2022strong} that repulsive Gibbs point processes with bounded-range potentials exhibit a notion of strong spatial mixing up to a fugacity of $\gppFugacity < \eulerE/\pwcc{\potential}$ (see \cite[Definition 1]{michelen2022strong}).
This result was used to derive a polynomial-time approximate sampling algorithm and a randomized approximation algorithm for the partition function of bounded-range repulsive Gibbs point processes in the same fugacity regime.

Given a repulsive Gibbs point process with fugacity $\gppFugacity < \eulerE/\pwcc{\potential}$, we show in \Cref{sec:connective_constant} that the hard-core model that is obtained from our reduction exhibits strong spatial mixing with high probability. In particular our result holds without bounded-range assumptions. We prove this by using the results in \cite{sinclair2013spatial,sinclair2017spatial}. Towards this, we establish a rigorous connection between the potential-weighted connective constant of a repulsive Gibbs point process and the connective constant of a graph from $\canonicalDistribution{n}{\region}{\potential}$.
More precisely, we show that, for any $\epsilon > 0$ and $n \ge \bigTheta{\volumeMeasure[\region]}$, the connective constant of a graph from $\canonicalDistribution{n}{\region}{\potential}$ is bounded by $\eulerE^{\epsilon} \frac{n}{\volumeMeasure[\region]} \pwcc{\potential}$ with probability at least $1 - \frac{1}{n}$ (see \Cref{thm:connective_constant} for the formal statement).
To obtain this result, we make use of the fact that the construction of the Weitz tree leaves some degree of freedom.
That is, for the same graph $\graph$ and root vertex $v \in \vertices_{\graph}$, different Weitz trees can be constructed, which differ in how the self-avoiding walk tree is truncated.
In our setting, we carefully need to choose the truncation based on the underlying location of vertices in $\region$.
In particular, it is important for us to define the connective constant in terms of the Weitz tree and not the full self-avoiding walk tree.
The latter would only yield a bound of $\frac{n}{\volumeMeasure[\region]} \generalizedTemperedness{\potential}$, which would be no improvement over the maximum degree of $\graph \sim \canonicalDistribution{n}{\region}{\potential}$.

Given the above graphical interpretation of $\pwcc{\potential}$, we immediately obtain that, for all $\gppFugacity < \eulerE/\pwcc{\potential}$, a hard-core model with vertex activity $\fugacity(n) = \frac{\volumeMeasure[\region]}{n} \gppFugacity$ exhibits strong spatial mixing on $\graph \sim \canonicalDistribution{n}{\region}{\potential}$ with probability at least $1 - \frac{1}{n}$ (see \Cref{cor:ssm_discretization} for the formal statement).
This result holds for any repulsive potential, without bounded-range assumption.
However, it should be noted that the strong spatial mixing is with respect to the graph distance and not the distance metric of the underlying space $\pointProcessSpace$.

This strong spatial mixing result for $\graph \sim \canonicalDistribution{n}{\region}{\potential}$ gives our reduction further algorithmic consequences.
Using the deterministic algorithm for approximating that partition function of a hard-core model proposed by Weitz \cite{weitz2006counting} (see also \cite{sinclair2013spatial,sinclair2017spatial}) and \Cref{thm:gpp_concentration}, our strong spatial mixing result yields a randomized approximation for the partition function of repulsive Gibbs point processes with arbitrary range potentials for $\gppFugacity < \eulerE/\pwcc{\potential}$ with quasi-polynomial running time $\volumeMeasure[\region]^{\bigO{\ln(\volumeMeasure[\region])}}$.
A similar result can be derived in the setting of approximate sampling.
% While these results hold without any bounded-range assumption, they lack the fully polynomial running time that was achieved in \cite{michelen2022strong} for the bounded-range setting.
% A fully polynomial randomized algorithm could for example be obtained by proving that Glauber dynamics, a Markov chain that is frequently used for sampling from the hard-core model, is rapidly mixing up to $\fugacity < \criticalFugacity{\connectiveConstant}$ (instead of $\fugacity < \criticalFugacity{\degree}$, which is the result we used previously).
% \todo{I am not sure how to finish this}

\subsection{A more general concentration result: antiferromagnetic partition functions on graphon-based random graphs} \label{subsec:discrete_spin_systems}
So far, we discussed how concentration of hard-core partition functions $\hcPartitionFunction{\graph}[\frac{\gppFugacity \volumeMeasure[\region]}{n}]$ for random graphs $\graph \sim \canonicalDistribution{n}{\region}{\potential}$, stated in \Cref{thm:gpp_concentration} is obtained from \Cref{thm:selfbounded_concentration_simplified}.
However, as hinted earlier, a more general concentration result can be obtained for a large class of antiferromagnetic two-state spin systems on graphon-based random graph models.
As such spin systems have been studied extensively \cite{li2013correlation,sinclair2014approximation,sly2012computational}, we believe this result to be of independent interest.
In what follows, we outline this more general concentration result and show how \Cref{thm:gpp_concentration} follows as a special case of it.

We start by introducing the class of spin systems to which it applies.
For an undirected graph $\graph = (\vertices, \edges)$ with vertices $\vertices$ and edges $\edges \subseteq \binom{\vertices}{2}$, we denote by $\spinConfigurations{\graph}$ the set of all functions $\spinConfiguration\colon \vertices \to \{0, 1\}$.
To simplify notation, we assume $\vertices = [n]$ for some $n \in \N$.
A \emph{two-state spin system} with parameters $\fugacity, \edgeInteraction_0, \edgeInteraction_1 \in \R_{\ge 0}$ on $\graph$ is a probability distribution $\GibbsDistribution{\graph}{\fugacity}{\edgeInteraction_{0}}{\edgeInteraction_{1}}$ on $\spinConfigurations{\graph}$ with
\[
	\GibbsDistribution{\graph}{\fugacity}{\edgeInteraction_{0}}{\edgeInteraction_{1}}[\spinConfiguration] = \frac{\fugacity^{\countOnes{\spinConfiguration}}  \edgeInteraction_{0}^{\countEdges{\graph}{0}{\spinConfiguration}} \edgeInteraction_{1}^{\countEdges{\graph}{1}{\spinConfiguration}}}{\partitionFunction{\graph}[\fugacity][\edgeInteraction_{0}, \edgeInteraction_{1}]}	,
\]
where $\countOnes{\spinConfiguration} = \size{\spinConfiguration^{-1}(1)}$ counts the number of vertices that are assigned to $1$, $\countEdges{\graph}{a}{\spinConfiguration} = \sum_{\{i, j\} \in \edges} \ind{\spinConfiguration[i] = \spinConfiguration[j] = a}$ counts the number of edges with both endpoints assigned to $a \in \{0, 1\}$, and the normalizing constant $\partitionFunction{\graph}[\fugacity][\edgeInteraction_{0}, \edgeInteraction_{1}]$ is the partition function
\[
	\partitionFunction{\graph}[\fugacity][\edgeInteraction_{0}, \edgeInteraction_{1}] = \sum\nolimits_{\spinConfiguration \in \spinConfigurations{\graph}} \fugacity^{\countOnes{\spinConfiguration}}  \edgeInteraction_{0}^{\countEdges{\graph}{0}{\spinConfiguration}} \edgeInteraction_{1}^{\countEdges{\graph}{1}{\spinConfiguration}} .
\]
A two-state spin system is \emph{antiferromagnetic} if $\edgeInteraction_{0} \edgeInteraction_{1} \le 1$.
Our concentration result applies to antiferromagnetic two-state spin systems with $\edgeInteraction_{0} = 1$.
In this case, we omit $\edgeInteraction_{0}$ completely, write $\edgeInteraction = \edgeInteraction_{1} \in [0, 1]$, and denote the partition function by $\partitionFunction{\graph}[\fugacity][\edgeInteraction]$.

Our concentration result for partition functions $\partitionFunction{\graph}[\fugacity][\edgeInteraction]$ applies to all graphon-based random-graph models.
Here, we refer to graphons in the most general sense, as defined in \cite[Chapter $13$]{lovasz2012large}.
That is, for a probability space $\vertexProbabilitySpace = (\vertexSpace, \vertexSigmaAlgebra, \vertexDistribution)$, a \emph{graphon} is a symmetric function $\edgeProbability\colon \vertexSpace^2 \to [0, 1]$ that is measurable with respect to the product algebra $\vertexSigmaAlgebra^2 = \vertexSigmaAlgebra \tensor \vertexSigmaAlgebra$.
Note that, even though we call the function $\edgeProbability$ the graphon, we mean implicitly that a graphon is a tuple of an underlying probability space and a suitable function $\edgeProbability$.
One useful aspect of a graphon $\edgeProbability$ is that it naturally defines a family of random-graph models, sometimes called $\edgeProbability$-random graphs (see \cite[Chapter $11$]{lovasz2012large}).
For every $n \in \N_{\ge 1}$, we denote by~$\graphDistribution{n}{\vertexProbabilitySpace}{\edgeProbability}$ a distribution on undirected graphs with vertex set $[n]$ that is induced by the following procedure for generating a random graph:
\begin{enumerate}
	\item Draw a tuple $(x_1, \dots x_n) \in \vertexSpace^n$ according to the product distribution $\vertexDistribution^n$.
	\item For all $i, j \in [n], i \neq j$, add the edge $\{i, j\}$ independently with probability $\edgeProbability[x_i][x_j]$.
\end{enumerate}
Observe that $\graphDistribution{n}{\vertexProbabilitySpace}{\edgeProbability}$ encompasses classical random-graph models, such as Erdős--Rényi random graphs and geometric random graphs.

Applying \Cref{thm:selfbounded_concentration_simplified} and using essentially the same arguments as in our proof sketch for \Cref{thm:gpp_concentration} yields the following result.

\begin{restatable}{theorem}{concentrationPartitionFunctionSimplified}
	\label{thm:concentration_partition_function_simplified}
	Let $\edgeProbability$ be a graphon on the probability space $\vertexProbabilitySpace = (\vertexSpace, \vertexSigmaAlgebra, \vertexDistribution)$.
	Let $\fugacity\colon \N_{\ge 1} \to \R_{\ge 0}$ such that $\fugacity[n] \le  \initialFugacity n^{-\frac{1 + \alpha}{2}}$ for some $\initialFugacity \in \R_{\ge 0}$ and $\alpha \in \R_{>0}$.
	For all $\edgeInteraction \in [0, 1]$, $\error \in (0, 1]$, $\errorProb \in (0, 1]$, $n \ge \left(2 \initialFugacity^2 \error^{-2} \errorProb^{-1}\right)^{\frac{1}{\alpha}}$, and $\graph \sim \graphDistribution{n}{\vertexProbabilitySpace}{\edgeProbability}$, it holds that
	\[
		\Pr{\absolute{\partitionFunction{\graph}[\fugacity[n]][\edgeInteraction] - \E{\partitionFunction{\graph}[\fugacity[n]][\edgeInteraction]}} \ge \error \E{\partitionFunction{\graph}[\fugacity[n]][\edgeInteraction]}}[][\big] \le \errorProb .
	\qedhere
	\]
\end{restatable}
\begin{remark}
	In fact, \Cref{thm:concentration_partition_function_simplified} can easily be extended to an even more general setting, where we consider a sequence of probability spaces $(\vertexProbabilitySpace_{n})_{n \in \N}$ and an associated sequence of graphons $(\edgeProbability_{n})_{n \in \N}$ (i.e., each $\edgeProbability_{n}$ is a graphon on $\vertexProbabilitySpace_{n}$).
	This makes the result for example applicable to popular models studied in network theory, such as hyperbolic random graphs~\cite{krioukov2010hyperbolic} and geometric inhomogeneous random graphs~\cite{DBLP:journals/tcs/BringmannKL19}.
\end{remark}

A surprising aspect of \Cref{thm:concentration_partition_function_simplified} is that, even though $\E{\partitionFunction{\graph}[\fugacity[n]][\edgeInteraction]} \ge 1 + n \fugacity[n]$ diverges for $\fugacity[n] \in \omega(n^{-1})$ as $n$ increases, \Cref{thm:concentration_partition_function_simplified} still ensures that the distribution of the partition functions gets more and more concentrated as long as $\fugacity[n] \in \smallO{n^{-\frac{1}{2}}}[\big]$.

We derive \Cref{thm:gpp_concentration} as a special case of \Cref{thm:concentration_partition_function_simplified}.
To see how this works, first observe that, for $\edgeInteraction=0$, it holds that $\partitionFunction{\graph}[\fugacity[n]][\edgeInteraction]$ is the hard-core partition function of a graph $\graph$ with parameter $\fugacity \in \R_{\ge 0}$.
Moreover, by setting $\vertexProbabilitySpace = (\region, \Borel, \uniformDistributionOn{\region})$, considering a graphon $\canonicalEdgeProbability{\potential}[x_1][x_2] = 1 - \eulerE^{- \potential[x_1][x_2]}$, $x_1, x_2 \in \region$, on $\vertexProbabilitySpace$, we obtain $\canonicalDistribution{n}{\region}{\potential} = \graphDistribution{n}{\vertexProbabilitySpace}{\canonicalEdgeProbability{\potential}}$.
This way of expressing $\canonicalDistribution{n}{\region}{\potential}$ establishes a connection between repulsive Gibbs point processes and hard-core models on graphon-based random graphs.
Lastly, setting $\fugacity_0 = \gppFugacity \volumeMeasure[\region]$ and $\fugacity[n] = \fugacity_0 n^{-1}$ and applying \Cref{thm:concentration_partition_function_simplified} yields the desired concentration result for hard-core partition functions on $\canonicalDistribution{n}{\region}{\potential}$.

\subsection{Organization of the technical details}
The technical details in the appendix are organized as follows.
In \Cref{sec:prelim}, we formally introduce the notion of antiferromagnetic spin systems and Gibbs point processes.
Note that for the latter one, we base our definition on the notion of random counting measures.
This is common in the theory of point processes but slightly differs from the definition given in the introduction.
In \Cref{sec:concentration}, we prove our concentration result for partition functions of antiferromagnetic spin systems on graphon-based random graphs.
We continue by showing in \Cref{sec:approximation} how to efficiently approximate the partition functions of a Gibbs point process with arbitrary repulsive potentials via hard-core partition functions of random graphs based on a suitably constructed graphon. In \Cref{sec:sampling}, we show how a similar approach can be used to obtain an approximate sampler for repulsive Gibbs point processes. Finally, in \Cref{sec:connective_constant} we give a high-probability bound on the connective constant of the graphs obtained from our reduction and derive corresponding strong spatial mixing results.

%% file: content/preliminaries.tex
\section{Preliminaries}\label{sec:prelim}

We formally introduce the discrete antiferromagnetic spin systems we investigate, as well as Gibbs point processes.

\subsection{Antiferromagnetic spin systems}
For an undirected graph $\graph = (\vertices, \edges)$ with vertices $\vertices$ and edges $\edges \subseteq \binom{\vertices}{2}$, we denote by $\spinConfigurations{\graph}$ the set of all functions $\spinConfiguration: \vertices \to \{0, 1\}$.
Without loss of generality, we are going to assume the canonical vertex set $\vertices = [n]$ for some $n \in \N$.
A \emph{two-state spin system} with parameters $\fugacity, \edgeInteraction_0, \edgeInteraction_1 \in \R_{\ge 0}$ on $\graph$ is a probability distribution $\GibbsDistribution{\graph}{\fugacity}{\edgeInteraction_{0}}{\edgeInteraction_{1}}$ on $\spinConfigurations{\graph}$ with
\[
	\GibbsDistribution{\graph}{\fugacity}{\edgeInteraction_{0}}{\edgeInteraction_{1}}[\spinConfiguration] = \frac{\fugacity^{\countOnes{\spinConfiguration}}  \edgeInteraction_{0}^{\countEdges{\graph}{0}{\spinConfiguration}} \edgeInteraction_{1}^{\countEdges{\graph}{1}{\spinConfiguration}}}{\partitionFunction{\graph}[\fugacity][\edgeInteraction_{0}, \edgeInteraction_{1}]}	,
\]
where $\countOnes{\spinConfiguration} = \size{\spinConfiguration^{-1}(1)}$ counts the number of vertices assigned that are to $1$, $\countEdges{\graph}{a}{\spinConfiguration} = \sum_{\{i, j\} \in \edges} \ind{\spinConfiguration[i] = \spinConfiguration[j] = a}$  counts the number of edges with both endpoints assigned to $a \in \{0, 1\}$ and $\partitionFunction{\graph}[\fugacity][\edgeInteraction_{0}, \edgeInteraction_{1}]$ is the normalizing constant
\[
	\partitionFunction{\graph}[\fugacity][\edgeInteraction_{0}, \edgeInteraction_{1}] = \sum_{\spinConfiguration \in \spinConfigurations{\graph}} \fugacity^{\countOnes{\spinConfiguration}}  \edgeInteraction_{0}^{\countEdges{\graph}{0}{\spinConfiguration}} \edgeInteraction_{1}^{\countEdges{\graph}{1}{\spinConfiguration}} .
\]
Note that we implicitly assume $\edgeInteraction_{0} \neq 0$ or $\edgeInteraction_{1} \neq 0$, as $\GibbsDistribution{\graph}{\fugacity}{\edgeInteraction_{0}}{\edgeInteraction_{1}}$ may not be defined otherwise.

Usually, $\GibbsDistribution{\graph}{\fugacity}{\edgeInteraction_{0}}{\edgeInteraction_{1}}$ is referred to as the \emph{Gibbs distribution} of the model and $\partitionFunction{\graph}$ is called the \emph{partition function}.
Further, a two-state spin system is \emph{antiferromagnetic} if $\edgeInteraction_{0} \edgeInteraction_{1} \le 1$.
For our concentration result, we focus on the setting where $\edgeInteraction_{0} = 1$.
In this case, we omit $\edgeInteraction_{0}$ completely and write $\edgeInteraction = \edgeInteraction_{1} \in [0, 1]$ and denote the partition function by $\partitionFunction{\graph}[\fugacity][\edgeInteraction]$.
Of special interest within this class of antiferromagnetic two-state spin systems in the \emph{hard-core model}, which results from setting $\edgeInteraction = 0$.
In this case, we might just omit the edge interactions $\edgeInteraction$ completely and write $\hcGibbsDistribution{\graph}{\fugacity}$ and $\hcPartitionFunction{\graph}[\fugacity]$.
Note that this implies that only configurations $\spinConfiguration \in \spinConfigurations{\graph}$ for which $\spinConfiguration^{-1}(1)$ is an independent set in $\graph$ can have non-zero probability.
For us, this model is especially relevant, as we show that concentration of hard-core partition functions on random graphs can be used to derive randomized approximations for the partition function of repulsive Gibbs point processes, which are introduced in the next section.

\subsection{Gibbs point processes}
\label{sec:prelim:gpp}
We introduce the notion of Gibbs point processes that is used throughout this paper.
For a formal treatment, it is common to model point processes as random counting measures.
Note that this is different from the simplified definition that we gave in the introduction.
For a more detailed overview on the theory of point processes and specifically Gibbs point processes, see \cite{jansen2018gibbsian}.

Let $(\pointProcessSpace, \dist)$ be a complete, separable metric space and let $\Borel = \Borel[\pointProcessSpace]$ be the Borel algebra of that space.
Let $\volumeMeasure$ be a locally finite reference measure on $(\pointProcessSpace, \Borel)$ such that all bounded measurable sets have finite measure.
Denote by $\countingMeasures$ the set of all locally finite counting measures on $(\pointProcessSpace, \Borel)$.
Formally, this is the set of all measures $\countingMeasure$ on $(\pointProcessSpace, \Borel)$ with values in $\N \cup \{\infty\}$ such that $\volumeMeasure[A] < \infty$ implies $\countingMeasure[A] < \infty$ for all $A \in \Borel$.
For each $A \in \Borel$, define a map $\countFunction{A}: \countingMeasures \to \N \cup \{\infty\}$ with $\countingMeasure \mapsto \countingMeasure[A]$ and let $\countingSigmaAlgebra$ be the sigma algebra on $\countingMeasures$ that is generated by the set of those maps $\{\countFunction{A} \mid A \in \Borel\}$.
A \emph{point process} on $\pointProcessSpace$ is now a measurable map from some probability space to the measurable space $(\countingMeasures, \countingSigmaAlgebra)$.
With some abuse of terminology, we call any probability distribution on $(\countingMeasures, \countingSigmaAlgebra)$ a point process, as we can only use the identity as measurable mapping from $\countingMeasure$ to itself.
Moreover, a point process is call \emph{simple} if $\countFunction{x}[\countingMeasure] \le 1$ with probability $1$, where we write $\countFunction{x}$ for $\countFunction{\{x\}}$.

Note that every counting measure $\countingMeasure \in \countingMeasures$ is associated with a multiset of points in $\pointProcessSpace$.
To see this, define $\pointSet[\countingMeasure] = \{x \in \pointProcessSpace \mid \countFunction{x}[\countingMeasure] > 0\}$.
Then $\countingMeasure$ can be expressed as a weighted sum of Dirac measures
\[
	\countingMeasure = \sum_{x \in \pointSet[\countingMeasure]} \countFunction{x}[\countingMeasure] \DiracMeasure{x} .
\]
In this sense, $\countingMeasure$ is associated with a multiset of points $x \in \pointSet[\countingMeasure]$, each occurring with finite multiplicity $\countFunction{x}[\countingMeasure]$.
We may use such a \emph{point configuration} interchangeably with its corresponding counting measure.

An important example for point processes are Poisson point processes.
A \emph{Poisson point process} with intensity $\PoissonIntensity \in \R_{\ge 0}$ on $(\pointProcessSpace, \dist)$ is uniquely defined by the following properties
\begin{itemize}
	\item for all bounded measurable $A \subseteq \pointProcessSpace$ it holds that $\countFunction{A}$ is Poisson distributed with intensity $\PoissonIntensity \volumeMeasure[A]$ and
	\item for all $m \in \N_{\ge 2}$ and disjoint measurable $A_1, \dots, A_m \subseteq \pointProcessSpace$ it holds that $\countFunction{A_1}, \dots, \countFunction{A_m}$ are independent.
\end{itemize}

Generally speaking, a \emph{Gibbs point process} is a point process that is absolutely continuous with respect to a Poisson point process.
For a bounded measurable $\region \subseteq \pointProcessSpace$ let $\countingMeasures[\region]$ denote the set of locally finite counting measures $\countingMeasure \in \countingMeasures$ that satisfy $\countFunction{A}[\countingMeasure] = 0$ for all measurable $A \subseteq \pointProcessSpace \setminus \region$.
In this work we are interested in Gibbs point processes $\GibbsPointProcess{\region}{\gppFugacity}{\potential}$ on bounded measurable regions $\region \subseteq \pointProcessSpace$ that are parameterized by a \emph{fugacity} parameter $\gppFugacity \in \R_{\ge 0}$ and non-negative, symmetric, measurable \emph{potential function} $\potential: \pointProcessSpace^2 \to \R_{\ge 0} \cup \{\infty\}$.
Formally, such a process $\GibbsPointProcess{\region}{\gppFugacity}{\potential}$ is defined by having a density with respect to a Poisson point process with intensity $\gppFugacity$ of the form
\[
	\frac{\intD \GibbsPointProcess{\region}{\gppFugacity}{\potential}}{\intD \PoissonPointProcess[\gppFugacity]} (\countingMeasure) = \frac{\ind{\countingMeasure \in \countingMeasures[\region]} \eulerE^{-\hamiltonian[\countingMeasure]}  \eulerE^{\gppFugacity \volumeMeasure[\region]}}{\gppPartitonFunction{\region}[\gppFugacity][\potential]}
\]
where $\hamiltonian: \countingMeasures \to \R_{\ge 0} \cup \{\infty\}$ is the \emph{Hamiltonian} defined by
\[
	\hamiltonian[\countingMeasure] = \sum_{\{x, y\} \in \binom{\pointSet[\countingMeasure]}{2}} \countFunction{x}[\countingMeasure] \countFunction{y}[\countingMeasure] \potential[x][y] + \sum_{x \in \pointSet[\countingMeasure]} \frac{\countFunction{x}[\countingMeasure] (\countFunction{x}[\countingMeasure]-1)}{2} \potential[x][x].
\]
The normalizing constant $\gppPartitonFunction{\region}[\gppFugacity][\potential]$ is usually called the \emph{(grand-canonical) partition function} and can be written explicitly as
\begin{align*}
	\gppPartitonFunction{\region}[\gppFugacity][\potential]
	&= 1 + \sum_{\numPoints \in \N_{\ge 1}} \frac{\gppFugacity^{\numPoints}}{\numPoints!} \int_{\region^{\numPoints}} \eulerE^{- \hamiltonian[\DiracMeasure{x_1} + \dots + \DiracMeasure{x_\numPoints}]} \productVolumeMeasure{\numPoints}[\intD \vectorize{x}] \\
	&= 1 + \sum_{\numPoints \in \N_{\ge 1}} \frac{\gppFugacity^{\numPoints}}{\numPoints!} \int_{\region^{\numPoints}} \prod_{\{i, j\} \in \binom{[k]}{2}} \eulerE^{- \potential[x_i][x_j]} \productVolumeMeasure{\numPoints}[\intD \vectorize{x}] .
\end{align*}

%% file: content/general_concentration.tex
\section{Concentration of partition functions of antiferromagnetic spin systems on graphon-based random graphs}\label{sec:concentration}

The main tool we will use to derive our concentration bounds is the Efron--Stein inequality.
For $\numSpaces \in \N_{\ge 1}$ let $\{(\probSpace_i, \sigmaAlgebra_i, \probMeasureIdx[i])\}_{i \in [\numSpaces]}$ be a collection of probability spaces and let $f: \probSpace \to \R$ be a measurable function on the product space $(\probSpace, \sigmaAlgebra, \probMeasure) = \bigTensor_{i \in [\numSpaces]} (\probSpace_i, \sigmaAlgebra_i, \probMeasureIdx[i])$.
For each $i \in [\numSpaces]$ define a function $\deviation{i}{f}: \probSpace \times \probSpace_i \to \R_{\ge 0}$, where, for every $\vectorize{x} = (x_1, \dots, x_{\numSpaces}) \in \probSpace$ and $y_i \in \probSpace_i$, the value $\deviation{i}{f}[\vectorize{x}][y_i]$ is defined as the squared difference in $f$ that is caused by replacing $x_i$ in $\vectorize{x}$ with $y_i$.
Formally, this is $\deviation{i}{f}[\vectorize{x}][y_i] = (f(\vectorize{x}) - f(\vectorize{y}))^2$ where $\vectorize{y} = (x_1, \dots, x_{i-1}, y_i, x_{i+1}, \dots, x_{\numSpaces})$.
The Efron--Stein inequality bounds the variance of $f$ under $\probMeasure$ based on the local squared deviations $\deviation{i}{f}[\vectorize{x}][y_i]$.

\begin{theorem}[{Efron--Stein inequality \cite{efron1981jackknife}}]
	\label{thm:variance_local_deviation}
	Let $\{(\probSpace_i, \sigmaAlgebra_i, \probMeasureIdx[i])\}_{i \in [\numSpaces]}$ be probability spaces with product space $(\probSpace, \sigmaAlgebra, \probMeasure) = \bigTensor_{i \in [\numSpaces]} (\probSpace_i, \sigmaAlgebra_i, \probMeasureIdx[i])$.
	For every $\sigmaAlgebra$-measurable function $f: \probSpace \to \R$ it holds that
	\[
		\VarWrt{f}[\probMeasure] \le \frac{1}{2} \sum\nolimits_{i \in [\numSpaces]} \EWrt{\deviation{i}{f}}[\probMeasure \times \probMeasureIdx[i]] . \qedhere
	\]
\end{theorem}

\begin{remark}
	The Efron--Stein inequality is usually stated for functions of independent real-valued random variables. However, it extends to functions on products of arbitrary probability spaces.
\end{remark}

\Cref{thm:variance_local_deviation} immediately gives a concentration result for $f$ whenever $\frac{1}{2} \sum_{i \in [\numSpaces]} \EWrt{\deviation{i}{f}}[\probMeasure \times \probMeasureIdx[i]][][\big]$ is of order of magnitude $\EWrt{f}[\probMeasure]^2$ by using Chebyshev's inequality.
However, obtaining such a bound might turn out difficult, especially if $\EWrt{f}[\probMeasure]$ is hard to compute explicitly.
For our setting, we derive the following corollary of \Cref{thm:variance_local_deviation}.

\begin{restatable}{corollary}{selfboundedConcentrationSimplified}[{Corollary of the Efron--Stein inequality}]
	\label{thm:selfbounded_concentration_simplified}
	Let $\{(\probSpace_i, \sigmaAlgebra_i, \probMeasureIdx[i])\}_{i \in [\numSpaces]}$ be
	% 	$\sigma$-finite
	probability spaces with product space $(\probSpace, \sigmaAlgebra, \probMeasure) = \bigTensor_{i \in [\numSpaces]} (\probSpace_i, \sigmaAlgebra_i, \probMeasureIdx[i])$, and let $f\colon \probSpace \to \R$ be an $\sigmaAlgebra$-measurable function.
	Assume that there are $c_i \in \R_{\ge 0}$ for $i \in [\numSpaces]$ such that $C \coloneqq \sum_{i \in [\numSpaces]} c_i^2 < 2$ and, for all $\vectorize{x} = (x_j)_{j \in [\numSpaces]} \in \probSpace$ and $\vectorize{y} = (y_j)_{j \in [\numSpaces]} \in \probSpace$ that disagree only at position $i$, it holds that
	\[
	\absolute{f(\vectorize{x}) - f(\vectorize{y})} \le c_i \cdot \min\{\absolute{f(\vectorize{x})}, \absolute{f(\vectorize{y})}\} .
	\]
	Then, for all $\error \in \R_{> 0}$, it holds that
	\[
	\Pr{\absolute{f - \EWrt{f}[\probMeasure]} \ge \error \EWrt{f}[\probMeasure]} \le \left(\frac{2}{2 - C} - 1\right) \frac{1}{\error^2} .\qedhere
	\]
\end{restatable}
%\selfboundedConcentrationSimplified*

\begin{proof}
	First, we observe that $\absolute{f(\vectorize{x}) - f(\vectorize{y})} \le c_i \min\{\absolute{f(\vectorize{x})}, \absolute{f(\vectorize{y})}\} \le c_i \absolute{f(\vectorize{x})}$ implies $\EWrt{\deviation{i}{f}}[\probMeasure \times \probMeasureIdx[i]] \le c_i^2 \EWrt{f^2}[\probMeasure]$ for all $i \in [\numSpaces]$.
	Thus, by \Cref{thm:variance_local_deviation}, we have $\VarWrt{f}[\probMeasure] \le \frac{C}{2} \EWrt{f^2}[\probMeasure]$.
	Now, recall that by definition $\VarWrt{f}[\probMeasure] = \EWrt{f^2}[\probMeasure] - \EWrt{f}[\probMeasure]^2$, which implies $\EWrt{f^2}[\probMeasure] - \EWrt{f}[\probMeasure]^2 \le \frac{C}{2} \EWrt{f^2}[\probMeasure]$.
	Rearranging for $\EWrt{f^2}[\probMeasure]$ and using the fact that $\frac{C}{2} < 1$ yields $\EWrt{f^2}[\probMeasure] \le \frac{2}{2 - C} \EWrt{f}[\probMeasure]^2$.
	Substituting this back into the definition of the variance, we obtain
	\[
		\VarWrt{f}[\probMeasure] \le \left(\frac{2}{2 - C} - 1\right) \EWrt{f}[\probMeasureIdx]^2 .
	\]
	The claim follows immediately by applying Chebyshev's inequality.
\end{proof}

\begin{remark}
	Usually, we want to characterize concentration asymptotically in $\numSpaces$.
	In this setting, \Cref{thm:selfbounded_concentration_simplified} tells us that, if $c_i \in \bigO{\numSpaces^{-\frac{1 + \alpha}{2}}}$ for all $i \in [\numSpaces]$ and some $\alpha > 0$, then, for all $\error \in \R_{>0}$ and $\errorProb \in (0, 1]$ such that $\error^2 \errorProb < 1$, it is sufficient to choose $\numSpaces \in \bigTheta{\errorProb^{-\frac{1}{\alpha}} \error^{-\frac{2}{\alpha}}}$ to ensure
	\[
		\Pr{\absolute{f - \EWrt{f}[\probMeasure]} \ge \error \EWrt{f}[\probMeasure]} \le \errorProb. \qedhere
	\]
\end{remark}

%% file: content/spin_systems.tex
We are now ready to use \Cref{thm:selfbounded_concentration_simplified} and derive a concentration result for the partition functions of antiferromagnetic two-state spin systems for graphon-based random graph model.

Let us recall the definition of graphons and graphon-based random graphs that we are using (see \cite[Chapter $10$ \& $13$]{lovasz2012large}).
Let $\vertexProbabilitySpace = (\vertexSpace, \vertexSigmaAlgebra, \vertexDistribution)$ be a probability space.
A graphon on $\vertexProbabilitySpace$ is a symmetric function $\edgeProbability: \vertexSpace^2 \to [0, 1]$ that is measurable with respect to the product algebra $\vertexSigmaAlgebra^2 = \vertexSigmaAlgebra \tensor \vertexSigmaAlgebra$.
For $n \in \N_{\ge 1}$ we denote by $\graphs{n}$ the set of all graphs on the canonical vertex set $[n] = \{1, \dots, n\}$.
Note that each graph in $\graphs{n}$ is fully characterized by its edge set $\edges$.
For every $n \in \N_{\ge 1}$ the random graph model induced by a graphon $\edgeProbability$ on a probability space $(\vertexSpace, \vertexSigmaAlgebra, \vertexDistribution)$ is described by generating a random graph $\graph = ([n], E)$ by
\begin{itemize}
	\item drawing a tuple $(x_1, \dots x_n) \in \vertexSpace^n$ according to the product distribution $\vertexDistribution^n$ and
	\item adding the edge $\{i, j\}$ for all $i, j \in [n], i \neq j$ independently with probability $\edgeProbability[x_i][x_j]$.
\end{itemize}
Formally, this gives a probability distribution $\graphDistribution{n}{\vertexProbabilitySpace}{\edgeProbability}$ on $\graphs{n}$ with
\[
	\graphDistribution{n}{\vertexProbabilitySpace}{\edgeProbability}[\graph] = \bigintsss_{\vertexSpace^n} \left(\prod_{\{i, j\} \in E} \edgeProbability[x_i][x_j]\right) \cdot \left(\prod_{\{i, j\} \in \binom{[n]}{2} \setminus E} \left(1 - \edgeProbability[x_i][x_j]\right)\right) \vertexDistribution^n(\intD \vectorize{x}) 
\] 
for all $\graph \in \graphs{n}$, where $\vectorize{x} = (x_i)_{i \in [n]}$ inside the integral.

To apply \Cref{thm:selfbounded_concentration_simplified} to partition functions on random graphs from $\graphDistribution{n}{\vertexProbabilitySpace}{\edgeProbability}$, we will need to bound how much the partition function changes when applying small modifications to the structure of a graph.
More specifically, we want to get a bound on the relative change of the partition function, given that we
\begin{itemize}
	\item add or remove a single edge, or
	\item add or remove a set of edges that are all incident to the same vertex.
\end{itemize} 
The following two lemmas provide such bounds.
\begin{lemma}
	\label{lemma:remove_edge}
	Let $\graph = (\vertices, \edges)$ be an undirected graph and, for any $e \in \edges$ let $\graph' = (\vertices, \edges \setminus \{e\})$.
	For all $\fugacity \in \R_{\ge 0}$ and $\edgeInteraction \in [0, 1]$ it holds that 
	\[
		0 \le \partitionFunction{\graph'}[\fugacity][\edgeInteraction] - \partitionFunction{\graph}[\fugacity][\edgeInteraction] \le \fugacity^2 	\partitionFunction{\graph}[\fugacity][\edgeInteraction]
	\]
	and especially
	\[
		\absolute{\partitionFunction{\graph'}[\fugacity][\edgeInteraction] - \partitionFunction{\graph}[\fugacity][\edgeInteraction]} \le \fugacity^2 	\min\{\partitionFunction{\graph}[\fugacity][\edgeInteraction], \partitionFunction{\graph'}[\fugacity][\edgeInteraction]\}.
		\qedhere
	\]
\end{lemma}

\begin{proof}
	Without loss of generality, assume $\vertices = [n]$ for some $n \in \N_{\ge 2}$ and let $e = \{i, j\}$ for $i, j \in [n]$.
	Note that $\spinConfigurations{\graph} = \spinConfigurations{\graph'}$, as their vertex sets are identical.
	Further, observe that, for all $\spinConfiguration \in \spinConfigurations{\graph}$, it holds that $\countEdges{\graph}{1}{\spinConfiguration} \ge \countEdges{\graph'}{1}{\spinConfiguration}$.
	Thus, we have $\edgeInteraction^{\countEdges{\graph}{1}{\spinConfiguration}} \le \edgeInteraction^{\countEdges{\graph'}{1}{\spinConfiguration}}$ and $\partitionFunction{\graph}[\fugacity][\edgeInteraction] \le \partitionFunction{\graph'}[\fugacity][\edgeInteraction]$, which proves 
	\[
		0 \le \partitionFunction{\graph'}[\fugacity][\edgeInteraction] - \partitionFunction{\graph}[\fugacity][\edgeInteraction] .
	\]
	We proceed by rewriting the partition function of $\graph'$ as
	\[
		\partitionFunction{\graph'}[\fugacity][\edgeInteraction] = \sum_{\substack{\spinConfiguration \in \spinConfigurations{\graph'}:\\ \spinConfiguration[i] = 0 \text{ or } \spinConfiguration[j] = 0}} \fugacity^{\countOnes{\spinConfiguration}} \edgeInteraction^{\countEdges{\graph'}{1}{\spinConfiguration}} + \sum_{\substack{\spinConfiguration \in \spinConfigurations{\graph'}:\\ \spinConfiguration[i] = \spinConfiguration[j] = 1}} \fugacity^{\countOnes{\spinConfiguration}} \edgeInteraction^{\countEdges{\graph'}{1}{\spinConfiguration}}.
	\]
	Observe that
	\[
		\sum_{\substack{\spinConfiguration \in \spinConfigurations{\graph'}:\\ \spinConfiguration[i] = 0 \text{ or } \spinConfiguration[j] = 0}}  \fugacity^{\countOnes{\spinConfiguration}} \edgeInteraction^{\countEdges{\graph'}{1}{\spinConfiguration}} 
		= \sum_{\substack{\spinConfiguration \in \spinConfigurations{\graph}:\\ \spinConfiguration[i] = 0 \text{ or } \spinConfiguration[j] = 0}}  \fugacity^{\countOnes{\spinConfiguration}} \edgeInteraction^{\countEdges{\graph}{1}{\spinConfiguration}}
	 	\le \partitionFunction{\graph}[\fugacity][\edgeInteraction] .
	\]
	For every $k \in [n]$, let $\neighbors{\graph'}{k}$ denote the neighbors of vertex $k$ in $\graph'$.
	We have
	\begin{align*}
		\sum_{\substack{\spinConfiguration \in \spinConfigurations{\graph'}:\\ \spinConfiguration[i] = \spinConfiguration[j] = 1}} \fugacity^{\countOnes{\spinConfiguration}} \edgeInteraction^{\countEdges{\graph'}{1}{\spinConfiguration}}
		&= \sum_{\substack{\spinConfiguration \in \spinConfigurations{\graph'}:\\ \spinConfiguration[i] = \spinConfiguration[j] = 0}} \fugacity^{\countOnes{\spinConfiguration}+2}   \edgeInteraction^{\countEdges{\graph'}{1}{\spinConfiguration}} \edgeInteraction^{\sum_{k \in \neighbors{\graph'}{i}} \spinConfiguration[k]} \edgeInteraction^{\sum_{k \in \neighbors{\graph'}{j}} \spinConfiguration[k]} \\
		&\le \fugacity^2 \sum_{\substack{\spinConfiguration \in \spinConfigurations{\graph'}:\\ \spinConfiguration[i] = \spinConfiguration[j] = 0}} \fugacity^{\countOnes{\spinConfiguration}}   \edgeInteraction^{\countEdges{\graph'}{1}{\spinConfiguration}} \\
		&= \fugacity^2 \sum_{\substack{\spinConfiguration \in \spinConfigurations{\graph}:\\ \spinConfiguration[i] = \spinConfiguration[j] = 0}} \fugacity^{\countOnes{\spinConfiguration}}   \edgeInteraction^{\countEdges{\graph}{1}{\spinConfiguration}} \\
		&\le \fugacity^2 \partitionFunction{\graph}[\fugacity][\edgeInteraction] .
	\end{align*}
	We conclude that $\partitionFunction{\graph'}[\fugacity][\edgeInteraction] \le \left(1 + \fugacity^2\right) \partitionFunction{\graph}[\fugacity][\edgeInteraction]$ and thus
	\[
		\partitionFunction{\graph'}[\fugacity][\edgeInteraction] - \partitionFunction{\graph}[\fugacity][\edgeInteraction]
		\le \fugacity^2 \partitionFunction{\graph}[\fugacity][\edgeInteraction] .
	\]
	The upper bound on $\absolute{\partitionFunction{\graph'}[\fugacity][\edgeInteraction] - \partitionFunction{\graph}[\fugacity][\edgeInteraction]}$ follows immediately.
\end{proof}

\begin{lemma}
	\label{lemma:add_vertex}
	Let $\graph = (\vertices, \edges)$ be an undirected graph and without loss of generality assume $\vertices = [n]$ for $n \in \N$.
	Let $\edges_{H}, \edges_{H'} \subseteq \{\{n+1, i\} \mid i \in [n]\}$, and set $H=([n+1], \edges \cup \edges_{H})$ and $H'=([n+1], \edges \cup \edges_{H'})$.
	For all $\fugacity \in \R_{\ge 0}$ and $\edgeInteraction \in [0, 1]$ it holds that 
	\[
		\absolute{\partitionFunction{H}[\fugacity][\edgeInteraction] - \partitionFunction{H'}[\fugacity][\edgeInteraction]} \le \fugacity  \partitionFunction{\graph}[\fugacity][\edgeInteraction]
		\le \fugacity \min\{\partitionFunction{H}[\fugacity][\edgeInteraction], \partitionFunction{H'}[\fugacity][\edgeInteraction]\} .
		\qedhere
	\]
\end{lemma}

\begin{proof}
	By \Cref{lemma:remove_edge}, we know that removing an edge from a graph doesn't decrease the partition function.
	Thus, $\partitionFunction{H}[\fugacity][\edgeInteraction]$ is maximized by choosing $\edges_{H} = \emptyset$ and minimized by choosing $\edges_{H} = \{\{n+1, i\} \mid i \in [n]\}$.
	Consequently, we have
	\[
		\partitionFunction{H}[\fugacity][\edgeInteraction] \le (1 + \fugacity) \partitionFunction{\graph}[\fugacity][\edgeInteraction]
	\] 
	and 
	\[
		\partitionFunction{H}[\fugacity][\edgeInteraction] 
		\ge  \partitionFunction{\graph}[\fugacity][\edgeInteraction] + \fugacity
		\ge \partitionFunction{\graph}[\fugacity][\edgeInteraction] .
	\] 
	As the same holds for $H'$, we obtain
	\[
		\absolute{\partitionFunction{H}[\fugacity][\edgeInteraction] - \partitionFunction{H'}[\fugacity][\edgeInteraction]} \le \fugacity \partitionFunction{\graph}[\fugacity][\edgeInteraction] 
	\]
	and the claim follows by noting that $\partitionFunction{\graph}[\fugacity][\edgeInteraction] \le \min\{\partitionFunction{H}[\fugacity][\edgeInteraction], \partitionFunction{H'}[\fugacity][\edgeInteraction]\}$.
\end{proof}

Based on \Cref{lemma:add_vertex,lemma:remove_edge}, we use \Cref{thm:selfbounded_concentration_simplified} to prove the following statement. 

\begin{theorem}
	\label{thm:concentration_partition_function}
	Let $\edgeProbability$ be a graphon on the probability space $\vertexProbabilitySpace = (\vertexSpace, \vertexSigmaAlgebra, \vertexDistribution)$.
	Let $\fugacity: \N_{\ge 1} \to \R_{\ge 0}$ such that $\fugacity[n] \le \initialFugacity n^{-\frac{1 + \alpha}{2}}$ for some $\initialFugacity \in \R_{\ge 0}$ and $\alpha \in \R_{>0}$.
	For all $\edgeInteraction \in [0, 1]$, $\error \in \R_{>0}$, $n > \initialFugacity^{\frac{2}{\alpha}}$ and $\graph \sim \graphDistribution{n}{\vertexProbabilitySpace}{\edgeProbability}$ it holds that
	\[
		\Pr{\absolute{\partitionFunction{\graph}[\fugacity[n]][\edgeInteraction] - \E{\partitionFunction{\graph}[\fugacity[n]][\edgeInteraction]}} \ge \error \E{\partitionFunction{\graph}[\fugacity[n]][\edgeInteraction]}}[][\big] \le \frac{\initialFugacity^2}{\left(n^{\alpha} - \initialFugacity^2\right) \error^2} .
		\qedhere
	\]
\end{theorem}

\begin{proof}
	We aim for applying \Cref{thm:selfbounded_concentration_simplified} to prove our claim.
	To this end, for each $n \in \N_{\ge 1}$ we need to write the partition function $\partitionFunction{\graph}[\fugacity[n]][\edgeInteraction]$ for $\graph \sim \graphDistribution{n}{\vertexProbabilitySpace}{\edgeProbability}$ as a function on a product of $\sigma$-finite probability spaces.
	At first, an obvious choice seems to be $\vertexProbabilitySpace^{n}$ together with $\binom{n}{2}$ additional binary random variables, one for each potential edge $\{i, j\} \in \binom{[n]}{2}$.
	However, note that the edges might not necessarily be independent, meaning that the resulting product distribution would not resemble  $\graphDistribution{n}{\vertexProbabilitySpace}{\edgeProbability}$.
	Instead, let $\uniformProbabilitySpace = ([0, 1], \Borel[[0, 1]], \uniformDistribution)$, where $\Borel[[0, 1]]$ is the Borel algebra restricted to $[0, 1]$ and $\uniformDistribution$ is the uniform distribution on that interval. 
	We consider the probability space $\vertexProbabilitySpace^{n} \tensor \uniformProbabilitySpace^{\binom{n}{2}}$.
	
	For $\vectorize{x} \in \vertexSpace^{n}$ and $\vectorize{y} \in [0, 1]^{\binom{n}{2}}$ let $\vectorize{x} \concat \vectorize{y} \in \vertexSpace^{n} \times [0, 1]^{\binom{n}{2}}$ denote the concatenation of $\vectorize{x}$ and $\vectorize{y}$.
	We construct a measurable function $g: \vertexSpace^n \times [0, 1]^{\binom{n}{2}} \to \graphs{n}$ by mapping every $\vectorize{z} = \vectorize{x} \concat \vectorize{y} \in \vertexSpace^n \times [0, 1]^{\binom{n}{2}}$ with $\vectorize{x} = (x_i)_{i \in [n]} \in \vertexSpace^n$ and $\vectorize{y} = (y_{i, j})_{1 \le i < j \le n} \in [0, 1]^{\binom{n}{2}}$ to $g(\vectorize{z}) = ([n], \edges)$ such that, for all $i<j$, it holds that $\{i, j\} \in \edges$ if and only if $\edgeProbability[x_i][x_j] \ge y_{i, j}$.
	Simple calculations show that, for $\vectorize{z} \sim \vertexDistribution^n \times \uniformDistribution^{\binom{n}{2}}$, it holds that $g(z) \sim \graphDistribution{n}{\vertexProbabilitySpace}{\edgeProbability}$.
	Now, let $f: \vertexSpace^n \times [0, 1]^{\binom{n}{2}} \to \R$ with $\vectorize{z} \mapsto \partitionFunction{g(z)}[\fugacity[n]][\edgeInteraction]$.
	In order to apply \Cref{thm:selfbounded_concentration_simplified}, we need to bound the relative change of $f(\vectorize{z})$ if we change one component of $\vectorize{z}$.
	Let $\vectorize{x'} = (x_1, \cdots, x_{i-1}, x_{i}', x_{i+1}, \dots, x_{n}) \in \vertexSpace^n$ for any $i \in [n]$.
	Then $g(\vectorize{x'} \concat \vectorize{y})$ can only differ from $g(\vectorize{z})$ on edges that are incident to vertex $i$. 
	Thus, by \Cref{lemma:add_vertex}, it holds that
	\[
		\absolute{f(\vectorize{z}) - f(\vectorize{x'} \concat \vectorize{y})} \le \fugacity[n] \min\{f(\vectorize{z}), f(\vectorize{x'} \concat \vectorize{y})\}.
	\] 
	Now, let $\vectorize{y'} = (y_{i, j}' )_{1 \le i < j \le n} \in [0, 1]^{\binom{n}{2}}$ such that $y_{i, j}' = y_{i, j}$ except for one pair $1 \le i < j \le n$.
	Note that $g(\vectorize{z})$ and $g(\vectorize{x} \concat \vectorize{y'})$ differ by at most one edge. 
	By \Cref{lemma:remove_edge}, we have
	\[
		\absolute{f(\vectorize{z}) - f(\vectorize{x} \concat \vectorize{y'})} \le \fugacity[n]^2 \min\{f(\vectorize{z}), f(\vectorize{x} \concat \vectorize{y'})\} .
	\] 
	Furthermore, note that for $\fugacity[n] \le \initialFugacity n^{-\frac{1 + \alpha}{2}}$ and $n > \initialFugacity^{\frac{2}{\alpha}}$ it holds that
	\[
		C = n \fugacity[n]^2 + \binom{n}{2} \fugacity[n]^4 
		\le \initialFugacity^2 n^{-\alpha} + \initialFugacity^4 n^{-2 \alpha}
		\le 2 \initialFugacity^2 n^{-\alpha}
		< 2 .
	\]
	Thus, by \Cref{thm:selfbounded_concentration_simplified} we obtain
	\begin{align*}
		\Pr{\absolute{\partitionFunction{\graph}[\fugacity[n]][\edgeInteraction] - \E{\partitionFunction{\graph}[\fugacity[n]][\edgeInteraction]}} \ge \error \E{\partitionFunction{\graph}[\fugacity[n]][\edgeInteraction]}}[][\big] 
		& \le \left(\frac{2}{2 - C} - 1\right) \frac{1}{\error^2} \\
		&\le \left(\frac{1}{1 - \initialFugacity^2 n^{-\alpha}} - 1\right) \frac{1}{\error^2} \\
		&= \frac{\initialFugacity^2}{\left(n^{\alpha} - \initialFugacity^2\right) \error^2} ,
	\end{align*}  
	which concludes the proof.
\end{proof}

\Cref{thm:concentration_partition_function_simplified} follows immediately from \Cref{thm:concentration_partition_function}.

\concentrationPartitionFunctionSimplified*

\begin{proof}
	For $\error \le 1$ and $\errorProb \le 1$ it holds that $n \ge \left(2 \initialFugacity^2 \error^{-2} \errorProb^{-1}\right)^{\frac{1}{\alpha}} > \initialFugacity^{\frac{2}{\alpha}}$.
	Applying \Cref{thm:concentration_partition_function} yields
	\begin{align*}
		\Pr{\absolute{\partitionFunction{\graph}[\fugacity[n]][\edgeInteraction] - \E{\partitionFunction{\graph}[\fugacity[n]][\edgeInteraction]}} \ge \error \E{\partitionFunction{\graph}[\fugacity[n]][\edgeInteraction]}}[][\big] 
		&\le \frac{1}{(2 \error^{-2} \errorProb^{-1} - 1) \error^2} \\
		&= \frac{\error^2 \errorProb}{(2 - \error^2 \errorProb) \error^2} \\
		&= \frac{\errorProb}{2 - \error^2 \errorProb} \\
		&\le \errorProb. 
		\qedhere 
	\end{align*}
	
\end{proof}

%% file: content/point_processes.tex
\section{Application to repulsive Gibbs point processes}\label{sec:approximation}
We use our concentration results for antiferromagnetic spin systems to relate repulsive Gibbs point processes to a hard-core model on carefully constructed classes of random graphs.
To this end, let $(\pointProcessSpace, \dist)$ be a complete, separable metric space, let $\Borel = \Borel[\pointProcessSpace]$ be the Borel algebra and let $\volumeMeasure$ be a locally finite reference measure on $(\pointProcessSpace, \Borel)$.
For every bounded and measurable $\region \subseteq \pointProcessSpace$ we define a probability space $\canonicalProbabilitySpace{\region} = (\region, \BorelOn{\region}, \uniformDistributionOn{\region})$, where $\BorelOn{\region}$ denotes the restriction of $\Borel$ to $\region$ and $\uniformDistributionOn{\region}$ is the probability measure on $(\region, \BorelOn{\region})$ that is defined via the constant density $\frac{1}{\volumeMeasure[\region]}$ with respect to $\volumeMeasure$ restricted to $\region$.
For every symmetric, repulsive and measurable pair potential function $\potential: \pointProcessSpace^2 \to \R_{\ge 0} \cup \{\infty\}$ and all $n \in \N_{\ge 1}$.
Define $\canonicalEdgeProbability{\potential}: \region^2 \to [0, 1]$ with $\canonicalEdgeProbability{\potential}[x][y] = 1 - \eulerE^{- \potential[x][y]}$ and observe that $\canonicalEdgeProbability{\potential}$ is a graphon on $\canonicalProbabilitySpace{\region}$. We proceed by considering the random graph model $\canonicalDistribution{n}{\region}{\potential} = \graphDistribution{n}{\canonicalProbabilitySpace{\region}}{\canonicalEdgeProbability{\potential}}$.

The following lemma relates the expected hard-core partition function on $\canonicalDistribution{n}{\region}{\potential}$ with the partition function of the continuous Gibbs point process $\gppPartitonFunction{\region}[\gppFugacity][\potential]$.

\begin{lemma}
	\label{lemma:bound_hc_expectation}
	Let $(\pointProcessSpace, \dist)$ be a complete separable metric space, let $\Borel = \Borel[\pointProcessSpace]$ be the Borel algebra and let $\volumeMeasure$ be a locally finite reference measure on $(\pointProcessSpace, \Borel)$.
	Let $\region \subseteq \pointProcessSpace$ be bounded and measurable, let $\gppFugacity \in \R_{\ge 0}$ and let $\potential: \pointProcessSpace^2 \to \R_{\ge 0} \cup \{\infty\}$ be a symmetric repulsive potential.
	For all $\error \in \R_{>0}$ and $n \ge 2 \error^{-1} \max\left\{\eulerE^6 \gppFugacity^2 \volumeMeasure[\region]^2, \ln\left(2 \error^{-1}\right)^2\right\}$ it holds that
	\[
		\left(1 - \error\right) \gppPartitonFunction{\region}[\gppFugacity][\potential]
		\le  \EWrt{\hcPartitionFunction{\graph}[\frac{\gppFugacity \volumeMeasure[\region]}{n}]}[\graph \sim \canonicalDistribution{n}{\region}{\potential}]
		\le  \gppPartitonFunction{\region}[\gppFugacity][\potential].
		\qedhere
	\]
\end{lemma}

\begin{proof}
	We start by rewriting the hard-core partition function as
	\[
		\hcPartitionFunction{\graph}[\frac{\gppFugacity \volumeMeasure[\region]}{n}]
		= 1 + \sum_{\numPoints = 1}^{n} \gppFugacity^{\numPoints} \frac{\volumeMeasure[\region]^{\numPoints}}{n^{\numPoints}} \sum_{S \in \binom{[n]}{\numPoints}} \prod_{\{i, j\} \in \binom{S}{2}} \ind{\{i, j\} \notin \edges} .
	\]
	Thus, by linearity of expectation we have
	\begin{align*}
		\EWrt{\hcPartitionFunction{\graph}[\frac{\gppFugacity \volumeMeasure[\region]}{n}]}[\graph \sim \canonicalDistribution{n}{\region}{\potential}]
		&= 1 + \sum_{\numPoints = 1}^{n} \gppFugacity^{\numPoints} \frac{\volumeMeasure[\region]^{\numPoints}}{n^{\numPoints}} \sum_{S \in \binom{[n]}{\numPoints}} \EWrt{\prod_{\{i, j\} \in \binom{S}{2}} \ind{\{i, j\} \notin \edges}}[\graph \sim \canonicalDistribution{n}{\region}{\potential}] \\
		&= 1 + \sum_{\numPoints = 1}^{n} \gppFugacity^{\numPoints} \frac{\volumeMeasure[\region]^{\numPoints}}{n^{\numPoints}} \sum_{S \in \binom{[n]}{\numPoints}} \Pr{\bigwedge_{\{i, j\} \in \binom{S}{2}} \{i, j\} \notin \edges} .
	\end{align*}
	Next, observe that for all $S \in \binom{[n]}{\numPoints}$ with $\size{S} = \numPoints$
	\begin{align*}
		\Pr{\bigwedge_{\{i, j\} \in \binom{S}{2}} \{i, j\} \notin \edges}
		&= \int_{\region^{n}} \prod_{\{i, j\} \in \binom{S}{2}} \left(1 - \canonicalEdgeProbability{\potential}[x_i][x_j]\right) \productUniformDistributionOn{\region}{n}[\intD \vectorize{x}] \\
		&= \int_{\region^{n}} \prod_{\{i, j\} \in \binom{S}{2}} \eulerE^{-\potential[x_i][x_j]} \productUniformDistributionOn{\region}{n}[\intD \vectorize{x}] \\
		&= \frac{1}{\volumeMeasure[\region]^{n}} \int_{\region^{n}} \prod_{\{i, j\} \in \binom{S}{2}} \eulerE^{-\potential[x_i][x_j]} \productVolumeMeasure{n}[\intD \vectorize{x}] \\
		&= \frac{\volumeMeasure[\region]^{n - \numPoints}}{\volumeMeasure[\region]^{n}} \int_{\region^{\numPoints}} \prod_{\{i, j\} \in \binom{[\numPoints]}{2}} \eulerE^{-\potential[x_i][x_j]} \productVolumeMeasure{\numPoints}[\intD \vectorize{x}] \\
		&= \frac{1}{\volumeMeasure[\region]^{\numPoints}} \int_{\region^{\numPoints}} \prod_{\{i, j\} \in \binom{[\numPoints]}{2}} \eulerE^{-\potential[x_i][x_j]} \productVolumeMeasure{\numPoints}[\intD \vectorize{x}] .
	\end{align*}
	This yields
	\begin{align*}
		\EWrt{\hcPartitionFunction{\graph}[\frac{\gppFugacity \volumeMeasure[\region]}{n}]}[\graph \sim \canonicalDistribution{n}{\region}{\potential}]
		&= 1 + \sum_{\numPoints = 1}^{n} \gppFugacity^{\numPoints} \frac{1}{n^{\numPoints}} \sum_{S \in \binom{[n]}{\numPoints}} \int_{\region^{\numPoints}} \prod_{\{i, j\} \in \binom{[\numPoints]}{2}} \eulerE^{-\potential[x_i][x_j]} \productVolumeMeasure{\numPoints}[\intD \vectorize{x}] \\
		&= 1 + \sum_{\numPoints = 1}^{n} \gppFugacity^{\numPoints} \frac{\binom{n}{\numPoints}}{n^{\numPoints}} \int_{\region^{\numPoints}} \prod_{\{i, j\} \in \binom{[\numPoints]}{2}} \eulerE^{-\potential[x_i][x_j]} \productVolumeMeasure{\numPoints}[\intD \vectorize{x}] \\
		&= 1 + \sum_{\numPoints = 1}^{n} \frac{\gppFugacity^{\numPoints}}{\numPoints!} \prod_{i = 0}^{\numPoints-1}\left(1 - \frac{i}{n}\right) \int_{\region^{\numPoints}} \prod_{\{i, j\} \in \binom{[\numPoints]}{2}} \eulerE^{-\potential[x_i][x_j]} \productVolumeMeasure{\numPoints}[\intD \vectorize{x}] ,
	\end{align*}
	from which the upper bound
	\[
		\EWrt{\hcPartitionFunction{\graph}[\frac{\gppFugacity \volumeMeasure[\region]}{n}]}[\graph \sim \canonicalDistribution{n}{\region}{\potential}]
		\le \gppPartitonFunction{\region}[\gppFugacity][\potential]
	\]
	follows immediately.

	For the lower bound set
	\[
	S_{m} = 1 + \sum_{\numPoints = 1}^{m} \frac{\gppFugacity^{\numPoints}}{\numPoints!} \int_{\region^{\numPoints}} \prod_{\{i, j\} \in \binom{[\numPoints]}{2}} \eulerE^{-\potential[x_i][x_j]} \productVolumeMeasure{\numPoints}[\intD \vectorize{x}]
	\]
	for any $1 \le m \le n$.
	Observe that
	\[
		\EWrt{\hcPartitionFunction{\graph}[\frac{\gppFugacity \volumeMeasure[\region]}{n}]}[\graph \sim \canonicalDistribution{n}{\region}{\potential}]  \ge \left(1 - 	\frac{m}{n}\right)^{m} S_{m} .
	\]
	Thus, for $n \ge 2 \error^{-1} m^2$ Bernoulli's inequality yields
	\[
		\EWrt{\hcPartitionFunction{\graph}[\frac{\gppFugacity \volumeMeasure[\region]}{n}]}[\graph \sim \canonicalDistribution{n}{\region}{\potential}]
		\ge \left(1 - \frac{m^2}{n}\right) S_{m}
		\ge \left(1 - \frac{\error}{2}\right) S_{m} .
	\]
	Furthermore, note that
	\begin{align*}
		\gppPartitonFunction{\region}[\gppFugacity][\potential] - S_{m}
		&= \sum_{\numPoints = m+1}^{\infty} \frac{\gppFugacity^{\numPoints}}{\numPoints!} \int_{\region^{\numPoints}} \prod_{\{i, j\} \in \binom{[\numPoints]}{2}} \eulerE^{-\potential[x_i][x_j]} \productVolumeMeasure{\numPoints}[\intD \vectorize{x}] \\
		&\le \sum_{\numPoints = m+1}^{\infty} \frac{\gppFugacity^{\numPoints} \volumeMeasure[\region]^{\numPoints}}{\numPoints!} ,
	\end{align*}
	where the last inequality comes from the fact that $\potential$ is non-negative.
	Next, observe that this is equal to the error of the Taylor expansion of $\eulerE^{\gppFugacity \volumeMeasure[\region]}$ around $0$, truncated after $m$ terms.
	Thus, by Lagrange's remainder formula, we obtain
	\[
		\gppPartitonFunction{\region}[\gppFugacity][\potential] - S_{m}
		\le \frac{\eulerE^{\gppFugacity \volumeMeasure[\region]}}{(m+1)!} \left(\gppFugacity \volumeMeasure[\region]\right)^{m+1} .
	\]
	Choosing $m \ge \max\left\{\eulerE^3 \gppFugacity \volumeMeasure[\region], \ln\left(2 \error^{-1}\right)\right\}$ and using the fact that $(m+1)! > \left(\frac{m+1}{\eulerE}\right)^{m+1}$ yields
	\[
		\gppPartitonFunction{\region}[\gppFugacity][\potential] - S_{m} \le \left(\frac{\eulerE^2 \gppFugacity \volumeMeasure[\region]}{m+1}\right)^{m+1} \le \eulerE^{- 	(m+1)} \le \frac{\error}{2} .
	\]
	As $\gppPartitonFunction{\region}[\gppFugacity][\potential] \ge 1$, we get
	\[
		S_{m}
		\ge \gppPartitonFunction{\region}[\gppFugacity][\potential] - \frac{\error}{2}
		\ge \left( 1 - \frac{\error}{2}\right) \gppPartitonFunction{\region}[\gppFugacity][\potential].
	\]
	For $n \ge 2 \error^{-1} m^2 = 2 \error^{-1} \max\left\{\eulerE^6 \gppFugacity^2 \volumeMeasure[\region]^2, \ln\left(2 \error^{-1}\right)^2\right\}$ we obtain
	\[
		\EWrt{\hcPartitionFunction{\graph}[\frac{\gppFugacity \volumeMeasure[\region]}{n}]}[\graph \sim \canonicalDistribution{n}{\region}{\potential}]
		\ge \left(1 - \frac{m}{n}\right)^{m} S_{m}
		\ge \left(1 - \frac{\error}{2}\right)^2  \gppPartitonFunction{\region}[\gppFugacity][\potential] \ge \left(1 - \error\right)  	\gppPartitonFunction{\region}[\gppFugacity][\potential] ,
	\]
	which proves the claim.
\end{proof}

\gppConcentration*

\begin{proof}
	By setting $\alpha = 1$ and $\initialFugacity = \gppFugacity \volumeMeasure[\region]$ and using the fact that
	\[
		n \ge 4 \error^{-2} \errorProb^{-1} \max\left\{\eulerE^6 \gppFugacity^2 \volumeMeasure[\region]^2, \ln\left(4 \error^{-1}\right)^2\right\}
		\ge \left(2 \initialFugacity^2 \left(\frac{\error}{2}\right)^{-2} \errorProb^{-1}\right)^{\frac{1}{\alpha}}
	\]
	\Cref{thm:concentration_partition_function_simplified} yields
	\[
		\Pr{\absolute{\hcPartitionFunction{\graph}[\frac{\gppFugacity \volumeMeasure[\region]}{n}] - \E{\hcPartitionFunction{\graph}[\frac{\gppFugacity  \volumeMeasure[\region]}{n}]}} \ge \frac{\error}{2} \E{\hcPartitionFunction{\graph}[\frac{\gppFugacity \volumeMeasure[\region]}{n}]}} \le \errorProb .
	\]
	Furthermore, by \Cref{lemma:bound_hc_expectation} we know that for
	\[
		n \ge 4 \error^{-2} \errorProb^{-1} \max\left\{\eulerE^6 \gppFugacity^2 \volumeMeasure[\region]^2, \ln\left(4 \error^{-1}\right)^2\right\}
		\ge 2 \left(\frac{\error}{2}\right)^{-1} \max\left\{\eulerE^6 \gppFugacity^2 \volumeMeasure[\region]^2, \ln\left( 2  \left(\frac{\error}{2}\right)^{-1}\right)^2\right\}
	\]
	it holds that
	\[
		\left(1 - \frac{\error}{2}\right) \gppPartitonFunction{\region}[\gppFugacity][\potential]
 		\le  \EWrt{\hcPartitionFunction{\graph}[\frac{\gppFugacity \volumeMeasure[\region]}{n}]}[\graph \sim \canonicalDistribution{n}{\region}{\potential}]
		\le  \gppPartitonFunction{\region}[\gppFugacity][\potential].
	\]
	Thus, we have
	\[
		\left(1 + \frac{\error}{2}\right) \E{\hcPartitionFunction{\graph}[\frac{\gppFugacity \volumeMeasure[\region]}{n}]}
		\le \left(1 + \frac{\error}{2}\right) \gppPartitonFunction{\region}[\gppFugacity][\potential]
		\le \left(1 + \error\right) \gppPartitonFunction{\region}[\gppFugacity][\potential]
	\]
	and similarly
	\[
		\left(1 - \frac{\error}{2}\right) \E{\hcPartitionFunction{\graph}[\frac{\gppFugacity \volumeMeasure[\region]}{n}]}
		\ge \left(1 - \frac{\error}{2}\right)^{2} \gppPartitonFunction{\region}[\gppFugacity][\potential]
		\ge \left(1 - \error\right) \gppPartitonFunction{\region}[\gppFugacity][\potential].
	\]
	We obtain
	\[
		\Pr{\absolute{\hcPartitionFunction{\graph}[\frac{\gppFugacity \volumeMeasure[\region]}{n}] - \gppPartitonFunction{\region}[\gppFugacity][\potential]} \ge \error  \gppPartitonFunction{\region}[\gppFugacity][\potential]} \le \errorProb ,
	\]
	which proves the claim.
\end{proof}

\subsection{Approximating the partition function}
One of the main applications of \Cref{thm:gpp_concentration} is that it yields a rather simple randomized procedure for approximating $\gppPartitonFunction{\region}[\gppFugacity][\potential]$.
The rough idea is as follows:
\begin{enumerate}
	\item For $n \in \N$ sufficiently large, sample a graph $\graph$ from $\canonicalDistribution{n}{\region}{\potential}$.
	\item Approximate $\hcPartitionFunction{\graph}[\frac{\gppFugacity \volumeMeasure[\region]}{n}]$ and use the result as an approximation for $\gppPartitonFunction{\region}[\gppFugacity][\potential]$.
\end{enumerate}
We are especially interested in obtaining an algorithm that is asymptotically efficient in the volume $\volumeMeasure[\region]$, as this gives a natural way to parameterize the algorithmic problem.
More specifically, we want to characterize the regime of the fugacity $\gppFugacity$ in terms of the potential $\potential$ for which we can get a randomized $\error$-approximation of $\gppPartitonFunction{\region}[\gppFugacity][\potential]$ in time polynomial in $\volumeMeasure[\region]$ and $\frac{1}{\error}$.
We characterize this fugacity regime in terms of the temperedness constant
\[
	\generalizedTemperedness{\potential} = \esssup_{x_1 \in \pointProcessSpace} \int_{\pointProcessSpace} \absolute{1 - \eulerE^{-\potential[x_1][x_2]}} \volumeMeasure[\intD x_2],
\]
where $\esssup$ denotes the essential supremum (i.e., an upper bound that holds almost everywhere).

In order to ensure that the approximation algorithm runs efficiently in $\volumeMeasure[\region]$, two ingredients are important.
First, we need to bound how large $n$ needs to be chosen to ensure that $\hcPartitionFunction{\graph}[\frac{\gppFugacity \volumeMeasure[\region]}{n}]$ is close to $\gppPartitonFunction{\region}[\gppFugacity][\potential]$ with high probability.
Second, we need to ensure that $\hcPartitionFunction{\graph}[\frac{\gppFugacity \volumeMeasure[\region]}{n}]$ can be approximated in time polynomial in $\volumeMeasure[\region]$.
Obviously, both requirements are satisfied if $n \in \poly{\volumeMeasure[\region]}$ is sufficient and if $\hcPartitionFunction{\graph}[\frac{\gppFugacity \volumeMeasure[\region]}{n}]$ can be approximated in time $\poly{n}$.
To tackle the first part, \Cref{thm:gpp_concentration} gives a useful tool.
For the second part, we will use some well known results on approximating the hard-core partition function.
\begin{theorem}[{\cite[Corollary $8.4$]{vstefankovivc2009adaptive} and \cite[Theorem $1$]{anari2021entropic}}]
	\label{thm:hc_fpras_univariate}
	Let $\graph = (\vertices, \edges)$ be an undirected graph with maximum vertex degree bounded by $\degree[\graph] \in \N_{\ge 2}$ and let $\fugacity \in \R_{\ge 0}$ with
	\[
		\fugacity < \criticalFugacity{\degree[\graph]} = \frac{\left(\degree[\graph] - 1\right)^{\degree[\graph] - 1}}{\left(\degree[\graph] - 2\right)^{\degree[\graph]}}.
	\]
	Then, for all $\error \in (0, 1]$, there is a randomized $\error$-approximation algorithm for the hard-core partition function $\hcPartitionFunction{\graph}[\fugacity]$ with running time $\bigOTilde{\size{\vertices}^2 \error^{-2}}$.
\end{theorem}

\begin{remark}
	In \cite{vstefankovivc2009adaptive} the result above is only stated for $\fugacity < \frac{2}{\degree[\graph]}$ as an older mixing time result for Glauber dynamics from \cite{vigoda2001note} is used. Combining their approach with the more recent mixing time bound in \cite{anari2021entropic} gives the desired bound of $\fugacity < \criticalFugacity{\degree[\graph]}$.
\end{remark}

Thus, arguing that $\hcPartitionFunction{\graph}$ for $\graph \sim \canonicalDistribution{n}{\region}{\potential}$ can be approximated in time $\poly{n}$ boils down to obtaining a probabilistic upper bound on $\degree[\graph]$.
We use the following simple lemma.

\begin{lemma}
	\label{lemma:degree_bound}
	Let $(\pointProcessSpace, \dist)$ be a complete, separable metric space, let $\Borel = \Borel[\pointProcessSpace]$ be the Borel algebra and let $\volumeMeasure$ be a locally finite reference measure on $(\pointProcessSpace, \Borel)$.
	Let $\region \subseteq \pointProcessSpace$ be bounded and measurable, let $\gppFugacity \in \R_{\ge 0}$ and let $\potential: \pointProcessSpace^2 \to \R_{\ge 0} \cup \{\infty\}$ be a symmetric repulsive potential.
	Assume $\generalizedTemperedness{\potential} > 0$.
	For $\degreeError \in \R_{>0}$, $\degreeErrorProb \in (0, 1]$, $n \ge 3 \max\left\{\degreeError^{-1}, \degreeError^{-2}\right\} \ln \left(\degreeErrorProb^{-1}\right) \generalizedTemperedness{\potential}^{-1} \volumeMeasure[\region] + 1$ and $\graph \sim \canonicalDistribution{n}{\region}{\potential}$ it holds that
	\[
		\Pr{\degree[\graph] \ge (1 + \degreeError)  \frac{n-1}{\volumeMeasure[\region]}\generalizedTemperedness{\potential}} \le \degreeErrorProb n.
		\qedhere
	\]
\end{lemma}

\begin{proof}
	By union bound, it is sufficient to argue that, for each $i \in [n]$ it holds that
	\[
		\Pr{\degree[\graph][i] \ge (1 + \degreeError)  \frac{n-1}{\volumeMeasure[\region]}\generalizedTemperedness{\potential}} \le \degreeErrorProb ,
	\]
	where $\degree[\graph][i]$ denotes the degree of vertex $i \in [n]$ in $\graph$.
	Now, observe that the random variables $\degree[\graph][i]$ for $i \in [n]$ are identically distributed.
	Thus, we can focus on $\degree[\graph][n]$ for ease of notation.
	By definition, it holds for $k \in [n-1] \cup \{0\}$ that
	\begin{align*}
		&\Pr{\degree[\graph][n] = k} \\
		&\hspace{3em}= \sum_{S \in \binom{[n-1]}{k}} \bigintsss_{\region^{n}} \left(\prod_{i \in S} \canonicalEdgeProbability{\potential}[x_n][x_i] \right) \cdot \left(\prod_{i \in [n-1] \setminus S} (1 - \canonicalEdgeProbability{\potential}[x_n][x_i]) \right) \productUniformDistributionOn{\region}{n}[\intD \vectorize{x}] \\
		&\hspace{3em}= \bigintsss_{\region} \sum_{S \in \binom{[n-1]}{k}} \left(\prod_{i \in S} \int_{\region} \canonicalEdgeProbability{\potential}[x_n][x_i] \,\uniformDistributionOn{\region}[\intD x_i]\right)
		\cdot \left(\prod_{i \in [n-1] \setminus S} \int_{\region} 1- \canonicalEdgeProbability{\potential}[x_n][x_i] \,\uniformDistributionOn{\region}[\intD x_i]\right) \,\uniformDistributionOn{\region}[\intD x_n] \\
		&\hspace{3em}= \int_{\region} \binom{n-1}{k} \left(\int_{\region} \canonicalEdgeProbability{\potential}[x_1][x_2] \,\uniformDistributionOn{\region}[\intD x_2]\right)^{k}
		\left(1 - \int_{\region} \canonicalEdgeProbability{\potential}[x_1][x_2] \,\uniformDistributionOn{\region}[\intD x_2]\right)^{n-1-k} \,\uniformDistributionOn{\region}[\intD x_1] .
	\end{align*}
	For every $x_1 \in \region$, let $B_{x_1}$ be a binomial random variable with $n-1$ trials and with success probability $\int_{\region} \canonicalEdgeProbability{\potential}[x_1][x_2]\,\uniformDistributionOn{\region}[\intD x_2]$.
	We obtain
	\[
		\Pr{\degree[\graph][n] = k} = \int_{\region} \Pr{B_{x_1} = k} \,\uniformDistributionOn{\region}[\intD x_1] ,
	\]
	which implies for all $a \in [0, n-1]$
	\begin{align*}
		\Pr{\degree[\graph][n] \ge a}
		&= \sum_{k=\lceil a \rceil}^{n-1} \int_{\region} \Pr{B_{x_1} = k} \,\uniformDistributionOn{\region}[\intD x_1] \\
		&= \int_{\region} \sum_{k=\lceil a \rceil}^{n-1} \Pr{B_{x_1} = k} \,\uniformDistributionOn{\region}[\intD x_1] \\
		&= \int_{\region} \Pr{B_{x_1} \ge a} \,\uniformDistributionOn{\region}[\intD x_1] .
	\end{align*}
	Next, let $B$ be a binomial random variable with $n-1$ trials and success probability $\frac{\generalizedTemperedness{\potential}}{\volumeMeasure[\region]}$.
	Observe that, by the definition of $\generalizedTemperedness{\potential}$, it holds for $\volumeMeasure$-almost all $x_1 \in \region$ that $\int_{\region} \canonicalEdgeProbability{\potential}[x_1][x_2]\,\uniformDistributionOn{\region}[\intD x_2] \le \frac{\generalizedTemperedness{\potential}}{\volumeMeasure[\region]}$.
	Thus, we have that $B$ stochastically dominates $B_{x_1}$ for $\uniformDistributionOn{\region}$-almost all $x_1 \in \region$.
	Consequently, we obtain
	\[
		\Pr{\degree[\graph][n] \ge a}
		\le \int_{\region} \Pr{B \ge a} \,\uniformDistributionOn{\region}[\intD x_1]
		= \Pr{B \ge a} .
	\]
	Observing that $\E{B} = \frac{n-1}{\volumeMeasure[\region]}\generalizedTemperedness{\potential}$ and applying Chernoff bound yields
	\begin{align*}
		\Pr{\degree[\graph][n] \ge (1 + \degreeError)  \frac{n-1}{\volumeMeasure[\region]}\generalizedTemperedness{\potential}}
		\le \eulerE^{- \frac{\min\left\{\degreeError, \degreeError^{2}\right\} \generalizedTemperedness{\potential} (n-1)}{3 \volumeMeasure[\region]}} .
	\end{align*}
	Setting $n \ge 3 \max\left\{\degreeError^{-1}, \degreeError^{-2}\right\} \ln \left(\degreeErrorProb^{-1}\right) \generalizedTemperedness{\potential}^{-1} \volumeMeasure[\region] + 1$ we have $\Pr{\degree[\graph][n] \ge (1 + \degreeError)  \frac{n-1}{\volumeMeasure[\region]}\generalizedTemperedness{\potential}} \le \degreeErrorProb$, which proves the claim.
\end{proof}

Combining \Cref{thm:gpp_concentration}, \Cref{lemma:degree_bound}, and \Cref{thm:hc_fpras_univariate}, we obtain the following algorithmic result.

\approximateGpp*

\begin{proof}
	We start by giving a more precise outline of the algorithmic idea.
	To this end, we define
	\[
	N = \max \left\{\substack{
		%4 \eulerE^6 \error^{-2} \errorProb^{-1} \gppFugacity^2 \volumeMeasure[\region]^2, \\
		324 \error^{-2} \max\left\{ \eulerE^6 \gppFugacity^2 \volumeMeasure[\region]^2,\, \ln\left(4 \error^{-1}\right)^2 \right\}, \\
		%4 \error^{-2} \errorProb^{-1} \ln\left(4 \error^{-1}\right)^2, \\
		%6 \ln\left(\eulerE q^{-1}\right) \volumeMeasure[\region] \left(\eulerE - \gppFugacity \generalizedTemperedness{\potential}\right)^{-1} \ln\left(6 \ln\left(\eulerE q^{-1}\right) \volumeMeasure[\region] \left(\eulerE - \gppFugacity \generalizedTemperedness{\potential}\right)^{-1}\right)^2, \\
		%6 \ln\left(\eulerE q^{-1}\right) \volumeMeasure[\region] \gppFugacity^2 \generalizedTemperedness{\potential} \left(\eulerE - \gppFugacity \generalizedTemperedness{\potential}\right)^{-2}\ln\left(6 \ln\left(\eulerE q^{-1}\right) \volumeMeasure[\region] \gppFugacity^2 \generalizedTemperedness{\potential} \left(\eulerE - \gppFugacity \generalizedTemperedness{\potential}\right)^{-2}\right)^2 \\
		24 \max\left\{\frac{1}{\eulerE - \gppFugacity \generalizedTemperedness{\potential}}, \frac{\gppFugacity \generalizedTemperedness{\potential} }{\left(\eulerE - \gppFugacity \generalizedTemperedness{\potential}\right)^{2}}\right\} \gppFugacity \volumeMeasure[\region] \ln\left(24 \max\left\{\frac{1}{\eulerE - \gppFugacity \generalizedTemperedness{\potential}}, \frac{\gppFugacity \generalizedTemperedness{\potential} }{\left(\eulerE - \gppFugacity \generalizedTemperedness{\potential}\right)^{2}}\right\} \gppFugacity \volumeMeasure[\region]\right)^2
	}\right\}.
	\]
	We now use the following procedure to approximate $\gppPartitonFunction{\region}[\gppFugacity][\potential]$:
	\begin{enumerate}
		\item Choose some integer $n \ge N$.
		\item Draw a graph $\graph$ from $\canonicalDistribution{n}{\region}{\potential}$. \label{step:draw_graph}
		\item If $\degree[\graph] \ge \frac{\eulerE n}{\gppFugacity \volumeMeasure[\region]}$, return an arbitrary value. \label{step:check_degree}
		\item Else, use the algorithm from \Cref{thm:hc_fpras_univariate} to $\frac{\error}{3}$-approximate $\hcPartitionFunction{\graph}[\frac{\gppFugacity \volumeMeasure[\region]}{n}]$ with an error probability of at most $\frac{1}{9}$ and return the result. \label{step:approximation}
	\end{enumerate}

	We proceed by arguing that this procedure yields an $\error$-approximation of $\gppPartitonFunction{\region}[\gppFugacity][\potential]$ in time $\poly{\volumeMeasure[\region] \error^{-1}}$.
	We start by bounding the probability that the computed value is not an $\error$-approximation.

	First, we assume that, whenever $\degree[\graph] \ge \frac{\eulerE n}{\gppFugacity \volumeMeasure[\region]}$, the algorithm returns no $\error$-approximation in step \ref{step:check_degree}.
	Let $A$ be the event that this happens.
	Second, let $B$ denote the event that the hard-core partition function $\hcPartitionFunction{\graph}[\frac{\gppFugacity \volumeMeasure[\region]}{n}]$ the graph $\graph$ that we drew in step \ref{step:draw_graph} is not an $\frac{\error}{3}$-approximation of $\gppPartitonFunction{\region}[\gppFugacity][\potential]$.
	Finally, let $C$ denote the event we do not manage to compute an $\frac{\error}{3}$-approximation of $\hcPartitionFunction{\graph}[\frac{\gppFugacity \volumeMeasure[\region]}{n}]$ in step \ref{step:approximation}.
	Note that the probability that the above procedure does not output an $\error$-approximation for $\gppPartitonFunction{\region}[\gppFugacity][\potential]$ is upper bounded by
	\[
		\Pr{A \cup (B \cap \compEvent{A}) \cup (C \cap \compEvent{B} \cap \compEvent{A})}
		\le \Pr{A} + \Pr{B} + \Pr{C} .
	\]
	We proceed with bounding each of these probabilities separately.

	To bound $\Pr{A}$, let $z = 24 \max\left\{\frac{1}{\eulerE - \gppFugacity \generalizedTemperedness{\potential}}, \frac{\gppFugacity \generalizedTemperedness{\potential} }{\left(\eulerE - \gppFugacity \generalizedTemperedness{\potential}\right)^{2}}\right\} \gppFugacity \volumeMeasure[\region]$.
	As we are interested in asymptotic behavior in terms of $\volumeMeasure[\region]$, we may assume that $\volumeMeasure[\region]$ is sufficiently large to ensure $z \ge 5$.
	Note that for this, we have to exclude the case $\gppFugacity = 0$, which trivially yields $\gppPartitonFunction{\region}[\gppFugacity][\potential] = 1$.
	Now, observe that for $z \ge 5$ it holds that $z \ln(z)^2 \ge z \ln\left(z \ln(z)^2\right)$.
	Next, observe that $n \ge z \ln(z)^2$.
	Thus, we have $n \ge z \ln(n)$.
	Furthermore, by $n \ge 5 \ln(5)^2 \ge \eulerE \ge 2$, we have
	\begin{align*}
		n - 1
		&\ge \frac{n}{2} \\
		&\ge 12 \max\left\{\frac{1}{\eulerE - \gppFugacity \generalizedTemperedness{\potential}}, \frac{\gppFugacity \generalizedTemperedness{\potential} }{\left(\eulerE - \gppFugacity \generalizedTemperedness{\potential}\right)^{2}}\right\} \gppFugacity \volumeMeasure[\region] \ln(n)	\\
		&\ge 3 \left(\ln\left(9\right) + 1\right) \max\left\{\frac{1}{\eulerE - \gppFugacity \generalizedTemperedness{\potential}}, \frac{\gppFugacity \generalizedTemperedness{\potential} }{\left(\eulerE - \gppFugacity \generalizedTemperedness{\potential}\right)^{2}}\right\} \gppFugacity \volumeMeasure[\region] \ln(n) \\
		&= 3 \left(\ln\left(9\right) \ln(n) + \ln(n)\right) \max\left\{\frac{1}{\eulerE - \gppFugacity \generalizedTemperedness{\potential}}, \frac{\gppFugacity \generalizedTemperedness{\potential} }{\left(\eulerE - \gppFugacity \generalizedTemperedness{\potential}\right)^{2}}\right\} \gppFugacity \volumeMeasure[\region] \\
		&\ge 3 \ln\left(9n\right) \max\left\{\frac{1}{\eulerE - \gppFugacity \generalizedTemperedness{\potential}}, \frac{\gppFugacity \generalizedTemperedness{\potential} }{\left(\eulerE - \gppFugacity \generalizedTemperedness{\potential}\right)^{2}}\right\} \gppFugacity \volumeMeasure[\region] .
	\end{align*}
	Thus, we obtain
	\[
		n \ge 3 \max\left\{\frac{\gppFugacity \generalizedTemperedness{\potential}}{\eulerE - \gppFugacity \generalizedTemperedness{\potential}},  \left(\frac{\gppFugacity \generalizedTemperedness{\potential} }{\eulerE - \gppFugacity \generalizedTemperedness{\potential}}\right)^{2}\right\} \ln\left(9n\right) \generalizedTemperedness{\potential}^{-1}   \volumeMeasure[\region] + 1
	\]
	and by \Cref{lemma:degree_bound}
	\[
		\Pr{\degree[\graph] \ge \frac{\eulerE n}{\gppFugacity \volumeMeasure[\region]}}
		\le \Pr{\degree[\graph] \ge \left(1 + \frac{\eulerE - \gppFugacity \generalizedTemperedness{\potential}}{\gppFugacity  \generalizedTemperedness{\potential}}\right)  \frac{n-1}{\volumeMeasure[\region]}\generalizedTemperedness{\potential}}
		\le \frac{1}{9} .
	\]

	To bound $\Pr{B}$, note that for $n \ge 324 \error^{-2} \max\left\{ \eulerE^6 \gppFugacity^2 \volumeMeasure[\region]^2,\, \ln\left(4 \error^{-1}\right)^2 \right\}$ \Cref{thm:gpp_concentration} yields
	\[
		\Pr{B} = \Pr{\absolute{\hcPartitionFunction{\graph}[\frac{\gppFugacity \volumeMeasure[\region]}{n}] - \gppPartitonFunction{\region}[\gppFugacity][\potential]} \ge  \frac{\error}{3} \gppPartitonFunction{\region}[\gppFugacity][\potential]} \le \frac{1}{9} .
	\]

	Finally, note that, by \Cref{thm:hc_fpras_univariate}, we can obtain an $\frac{\error}{3}$-approximation of $\hcPartitionFunction{\graph}[\frac{\gppFugacity \volumeMeasure[\region]}{n}]$ with error probability at most $\Pr{C} \le \frac{1}{9}$ in time $\bigOTilde{n^2 \error^{-2}}$ as long as $\frac{\gppFugacity \volumeMeasure[\region]}{n} < \criticalFugacity{\degree[\graph]}$.
	As we only run the approximation for graphs $\graph$ with $\degree[\graph] < \frac{\eulerE n}{\gppFugacity \volumeMeasure[\region]}$ it holds that
	\[
		\frac{\gppFugacity \volumeMeasure[\region]}{n}
		< \frac{\eulerE}{\degree[\graph]} < \criticalFugacity{\degree[\graph]} ,
	\]
	proving that the requirement is satisfied.

	We obtain that the error probability is bounded by $\frac{1}{3}$.
	To finish the proof, we need to argue that our algorithm has the desired running time.
	To this end, note that $N \in \bigO{\volumeMeasure[\region]^{2} \error^{-2}}$.
	Thus, we can also choose $n \in \bigO{\volumeMeasure[\region]^{2} \error^{-2}}$.
	By assumption, step \ref{step:draw_graph} can be computed in time $\sampleGraphTime{\region}{\potential}{n} = \sampleGraphTime{\region}{\potential}{\bigO{\volumeMeasure[\region]^2 \error^{-2}}}$.
	Furthermore, step \ref{step:check_degree} can be computed in time $\bigOTilde{n^2 \volumeMeasure[\region]^{-1}} = \bigOTilde{\volumeMeasure[\region]^3 \samplingError^{-4}}$ and, by \Cref{thm:hc_fpras_univariate}, step \ref{step:approximation} runs in time $\bigOTilde{n^2 \error^{-2}} = \bigOTilde{\volumeMeasure[\region]^{4} \error^{-6}}$ for $\frac{\gppFugacity \volumeMeasure[\region]}{n} < \criticalFugacity{\degree[\graph]}$.
	Consequently, the overall running time is in $\bigOTilde{\volumeMeasure[\region]^{4} \error^{-6}} + \sampleGraphTime{\region}{\potential}{\bigO{\volumeMeasure[\region]^2 \error^{-2}}}$.
\end{proof}

% \Cref{cor:approximate_gpp_simplified} follows immediately from \Cref{thm:gpp_concentration} by noting that a graph from $\canonicalDistribution{n}{\region}{\potential}$ can be sampled by drawing $n$ points from $\uniformDistributionOn{\region}$ and evaluating $\potential$ for each pair of points.

%% file: content/sampling.tex
\section{Sampling from repulsive Gibbs point processes}\label{sec:sampling}
In this section, we propose an approximate sampling algorithm for the Gibbs measure of a repulsive Gibbs point process, based in random hard-core models.
More precisely, we investigate the sampling procedure given by \Cref{algo:sampling}

\begin{algorithm}[h]
	\SetAlgoLined
	\KwData{Instance of a repulsive Gibbs point process $(\region, \gppFugacity, \potential)$, error bound $\samplingError \in (0, 1]$}
	\KwResult{multiset of points in $\region$}
	set $n = \left\lceil\max\left\{
		\substack{
			8 \frac{18^2 \cdot 12}{\samplingError^3} \max\left\{\eulerE^6 \gppFugacity^2 \volumeMeasure[\region]^2, \ln\left(\frac{4 \cdot 18}{\samplingError}\right)\right\}, \\
			6 \ln\left(\frac{4 \eulerE}{\samplingError}\right) \max\left\{\frac{1}{\eulerE - \gppFugacity \generalizedTemperedness{\potential}}, \frac{\gppFugacity \generalizedTemperedness{\potential} }{\left(\eulerE - \gppFugacity \generalizedTemperedness{\potential}\right)^{2}}\right\} \gppFugacity \volumeMeasure[\region] \ln\left(3 \ln\left(\frac{4 \eulerE}{\samplingError}\right) \max\left\{\frac{1}{\eulerE - \gppFugacity \generalizedTemperedness{\potential}}, \frac{\gppFugacity \generalizedTemperedness{\potential} }{\left(\eulerE - \gppFugacity \generalizedTemperedness{\potential}\right)^{2}}\right\} \gppFugacity \volumeMeasure[\region]\right)^2
		}
		 \right\} \right\rceil$\;
	for each $i \in [n]$ draw $\randomPoint_i \sim \uniformDistributionOn{\region}$ independently\;
	draw $\edges \subseteq \binom{[n]}{2}$ s.t. $\{i, j\} \in \edges$ with probability $\canonicalEdgeProbability{\potential}[\randomPoint_i][\randomPoint_j] = 1 - \eulerE^{-\potential[\randomPoint_i][\randomPoint_j]}$ independently\;
	set $\samplingGraph = ([n], \edges)$\;
	\eIf{maximum degree $\degree[\samplingGraph] \ge \frac{\eulerE n}{\gppFugacity \volumeMeasure[\region]}$}{
		set $\sampledPointSet = \emptyset$\;
	}
	{
		sample $\spinConfiguration \in \spinConfigurations{\samplingGraph}$ $\frac{\samplingError}{4}$-approximately from the hard-core distribution $\hcGibbsDistribution{\samplingGraph}{\fugacity[n]}$ where $\fugacity[n] = \frac{\gppFugacity \volumeMeasure[\region]}{n}$\;
		set $\sampledPointSet = \{\randomPoint_i \mid i \in [n] \text{ s.t. } \spinConfiguration[i] = 1\}$ (possibly multiset)\;
	}
	\Return $\sampledPointSet$\;
	\caption{\label{algo:sampling}Approximate sampling algorithm for a repulsive point process $(\region, \gppFugacity, \potential)$.}
\end{algorithm}

Our main theorem in this section is as follows.
\begin{theorem}
	\label{thm:sampling}
	Let $(\pointProcessSpace, \dist)$ be a complete, separable metric space, let $\Borel = \Borel[\pointProcessSpace]$ be the Borel algebra and let $\volumeMeasure$ be a locally finite reference measure on $(\pointProcessSpace, \Borel)$.
	Let $\region \subseteq \pointProcessSpace$ be bounded and measurable, let $\gppFugacity \in \R_{\ge 0}$ and let $\potential: \pointProcessSpace^2 \to \R_{\ge 0} \cup \{\infty\}$ be a symmetric repulsive potential.
	Assume we can sample from the uniform distribution $\uniformDistributionOn{\region}$ in time $\samplePointTime{\region}$ and, for every $x, y \in \region$, evaluate $\potential[x][y]$ in time $\evaluatePotentialTime{\potential}$.
	If the Gibbs point process $\GibbsPointProcess{\region}{\gppFugacity}{\potential}$ is simple and
	$	\gppFugacity < \frac{\eulerE}{\generalizedTemperedness{\potential}}$
	then, for every $\samplingError \in \R_{>0}$, \Cref{algo:sampling} samples $\samplingError$-approximately from $\GibbsPointProcess{\region}{\gppFugacity}{\potential}$ and has running time in
	$\bigOTilde{\volumeMeasure[\region]^2 \samplingError^{-4} + \volumeMeasure[\region]^2 \samplingError^{-3} \samplePointTime{\region} + \volumeMeasure[\region]^4 \samplingError^{-6} \evaluatePotentialTime{\potential}}$.
\end{theorem}

\Cref{thm:sampling_simplified} follows immediately from the theorem above.
To prove \Cref{thm:sampling}, we start by analyzing a simplified algorithm, given in \Cref{algo:modifed_sampling}.

\begin{algorithm}[h]
	\SetAlgoLined
	\KwData{Instance of a repulsive Gibbs point process $(\region, \gppFugacity, \potential)$, error bound $\samplingError \in (0, 1]$}
	\KwResult{multiset of points in $\region$}
	set $n = \left\lceil\max\left\{
	\substack{
		8 \frac{18^2 \cdot 12}{\samplingError^3} \max\left\{\eulerE^6 \gppFugacity^2 \volumeMeasure[\region]^2, \ln\left(\frac{4 \cdot 18}{\samplingError}\right)\right\}, \\
		6 \ln\left(\frac{4 \eulerE}{\samplingError}\right) \max\left\{\frac{1}{\eulerE - \gppFugacity \generalizedTemperedness{\potential}}, \frac{\gppFugacity \generalizedTemperedness{\potential} }{\left(\eulerE - \gppFugacity \generalizedTemperedness{\potential}\right)^{2}}\right\} \gppFugacity \volumeMeasure[\region] \ln\left(3 \ln\left(\frac{4 \eulerE}{\samplingError}\right) \max\left\{\frac{1}{\eulerE - \gppFugacity \generalizedTemperedness{\potential}}, \frac{\gppFugacity \generalizedTemperedness{\potential} }{\left(\eulerE - \gppFugacity \generalizedTemperedness{\potential}\right)^{2}}\right\} \gppFugacity \volumeMeasure[\region]\right)^2
	}
	\right\} \right\rceil$\;
	for each $i \in [n]$ draw $\randomPoint_i \sim \uniformDistributionOn{\region}$ independently\;
	draw $\edges \subseteq \binom{[n]}{2}$ s.t. $\{i, j\} \in \edges$ with probability $\canonicalEdgeProbability{\potential}[\randomPoint_i][\randomPoint_j] = 1 - \eulerE^{-\potential[\randomPoint_i][\randomPoint_j]}$ independently\;
	set $\samplingGraph = ([n], \edges)$\;
	sample $\spinConfigurationModified \in \spinConfigurations{\samplingGraph}$ exactly from the hard-core distribution $\hcGibbsDistribution{\samplingGraph}{\fugacity[n]}$ where $\fugacity[n] = \frac{\gppFugacity \volumeMeasure[\region]}{n}$\;
	set $\modifiedSampledPointSet = \{\randomPoint_i \mid i \in [n] \text{ s.t. } \spinConfigurationModified[i] = 1\}$ (possibly multiset)\;
	\Return $\modifiedSampledPointSet$\;
	\caption{\label{algo:modifed_sampling}Modified sampling process}
\end{algorithm}

The main difference between \Cref{algo:sampling} and \Cref{algo:modifed_sampling} is that the latter one does not check if the maximum degree of the sampled graph $\samplingGraph$ is bounded and that is assumes access to a perfect sampler for $\hcGibbsDistribution{\samplingGraph}{\fugacity[n]}$.
It is not clear if such a perfect sampler for the hard-core Gibbs distribution can be realized in polynomial time, especially for arbitrary vertex degrees.
Therefore, \Cref{algo:modifed_sampling} is not suitable for algorithmic applications.
However, the main purpose of \Cref{algo:modifed_sampling} is that the distribution of point multisets that it outputs are much easier to analyze.
We use this, together with a coupling argument, to bound the total variation distance between the output of \Cref{algo:sampling} and $\GibbsPointProcess{\region}{\gppFugacity}{\potential}$.
Once this is done, it remains to show that \Cref{algo:sampling} satisfies the running time requirements, given in \Cref{thm:sampling}.

To analyze the output distribution of \Cref{algo:modifed_sampling}, we start by considering the resulting distribution of multisets of points (or counting measures respectively) when conditioning on the event that the hard-core partition function $\hcPartitionFunction{\samplingGraph}[\fugacity[n]]$ of the drawn graph $\samplingGraph$ is close to the partition function of the continuous process $\gppPartitonFunction{\region}[\gppFugacity][\potential]$.
More specifically, for any given $n$ and $\alpha \in \R_{\ge 0}$, let $\approximationGraphs{n}{\alpha} = \{H \in \graphs{n} \mid \absolute{\hcPartitionFunction{H}[\fugacity[n]] - \gppPartitonFunction{\region}[\gppFugacity][\potential]} \le \alpha \gppPartitonFunction{\region}[\gppFugacity][\potential]\}$.
We derive an explicit density for the output of \Cref{algo:modifed_sampling} with respect to a Poisson point process under the condition that $\samplingGraph \in \approximationGraphs{n}{\alpha}$ for some sufficiently small $\alpha$.
To this end, we use the following characterization of simple point processes via so called \emph{void probabilities}.

\begin{theorem}[{R\'{e}nyi--Mönch, see \cite[Theorem 9.2.XII]{daley2008introduction}}]
	\label{thm:pp_void_probabilities}
	Let $(\pointProcessSpace, \dist)$ be a complete, separable metric space, let $\Borel = \Borel[\pointProcessSpace]$ be the associated Borel algebra.
	Let $P$ and $Q$ be simple point process on $(\pointProcessSpace, \dist)$.
	If, for $\countingMeasure_P \sim P$ and $\countingMeasure_Q \sim Q$ and for all bounded $B \in \Borel$, it holds that
	\[
		\Pr{\countingMeasure_{P}(B) = 0} = \Pr{\countingMeasure_{Q}(B)= 0},
	\]
	then $P = Q$.
\end{theorem}

\Cref{thm:pp_void_probabilities} greatly simplifies proving that a given candidate function actually is a valid density for the point process in question, as it implies that it is sufficient to check if it yields the correct void probabilities.

Before we proceed, we introduce some additional notation that is useful for stating and proving our next lemmas.
For a given graph $H = (\vertices, \edges)$, we denote by $\independentSets{H} \subseteq \powerset{\vertices}$ the set of all independent sets in $H$.
Moreover, for every spin configuration $\spinConfiguration \in \spinConfigurations{H}$, we denote by $\configurationToSet{\spinConfiguration}$ the set of all vertices $v \in \vertices$ with $\spinConfiguration[v] = 1$.
Note that, for a hard-core model on $H$ with $\fugacity > 0$, this construction gives a one-to-one correspondence between $\independentSets{H}$ and the set of spin configurations $\spinConfiguration \in \spinConfigurations{H}$ with $\hcGibbsDistribution{H}{\fugacity}[\spinConfiguration] > 0$.
Therefore, it is often convenient to argue about elements in $\independentSets{H}$ instead of using spin configurations.

\begin{lemma}
	\label{lemma:conditional_density}
	Let $(\pointProcessSpace, \dist)$ be a complete, separable metric space, let $\Borel = \Borel[\pointProcessSpace]$ be the Borel algebra and let $\volumeMeasure$ be a locally finite reference measure on $(\pointProcessSpace, \Borel)$.
	Let $\region \subseteq \pointProcessSpace$ be bounded and measurable, let $\gppFugacity \in \R_{\ge 0}$ and let $\potential: \pointProcessSpace^2 \to \R_{\ge 0} \cup \{\infty\}$ be a symmetric repulsive potential.
	Furthermore, for any given $\samplingError \in  (0, 1]$, let $\modfiedSamplerOutput{\samplingError}$ be the point process produced by \Cref{algo:modifed_sampling} conditioned on $\samplingGraph \in \approximationGraphs{n}{\frac{\samplingError}{12}}$, and let $\PoissonPointProcess[\gppFugacity]$ denote a Poisson point process with intensity $\gppFugacity$.
	If the Gibbs point process $\GibbsPointProcess{\region}{\gppFugacity}{\potential}$ is simple, then $\modfiedSamplerOutput{\samplingError}$ has a density with respect to $\PoissonPointProcess$ of the form
	\begin{align*}
		\modifiedSamplingDensity{\samplingError}[\countingMeasure]
		=  \ind{\countingMeasure \in \countingMeasures[\region]} &\Pr{\samplingGraph \in \approximationGraphs{n}{\frac{\samplingError}{12}}}^{-1} \left(\prod_{i=0}^{\countingMeasure[\region] - 1} 1 - \frac{i}{n} \right) \ind{\countingMeasure[\region] \le n} \\
		&\cdot
		\left(\prod_{\{x, y\} \in \binom{\pointSet[\countingMeasure]}{2}} \eulerE^{- \countFunction{x}[\countingMeasure] \countFunction{y}[\countingMeasure] \potential[x][y]}\right)
		\left(\prod_{x \in \pointSet[\countingMeasure]} \eulerE^{- \frac{\countFunction{x}[\countingMeasure] (\countFunction{x}[\countingMeasure]-1)}{2} \potential[x][x]}\right) \densityNormalizingFunction{n}{\countingMeasure[\region]}[\countingMeasureToTuple[\countingMeasure]] \eulerE^{\gppFugacity\volumeMeasure[\region]} ,
	\end{align*}
	where $\countingMeasureToTuple$ maps every finite counting measure $\countingMeasure$ to an arbitrary but fixed tuple $(x_1, \dots, x_{\countingMeasure[\pointProcessSpace]})$ such that $\countingMeasure = \sum_{i = 1}^{\countingMeasure[\pointProcessSpace]} \DiracMeasure{x_i}$ and
	\begin{align*}
		\densityNormalizingFunction{n}{k}[\vectorize{x}] = \sum_{\substack{H \in \approximationGraphs{n}{\frac{\samplingError}{12}}:\\ [k] \in \independentSets{H}}} \frac{1}{\hcPartitionFunction{H}[\fugacity[n]]}
		\bigintss_{\region^{n - k}}
		&\left( \prod_{\substack{(i, j) \in [k] \times [n - k]: \\\{i, j + k\} \in \edges_{H}}} 1 - \eulerE^{-\potential[x_i][y_j]}\right)
		\left( \prod_{\substack{(i, j) \in [k] \times [n - k]: \\\{i, j + k\} \notin \edges_{H}}} \eulerE^{-\potential[x_i][y_j]}\right) \\
		&\left( \prod_{\substack{\{i, j\} \in \binom{[n - k]}{2}: \\\{i + k, j + k\} \in \edges_{H}}} 1 - \eulerE^{-\potential[y_i][y_j]}\right)
		\left( \prod_{\substack{\{i, j\} \in \binom{[n - k]}{2}: \\\{i + k, j + k\} \notin \edges_{H}}} \eulerE^{-\potential[y_i][y_j]}\right)
		\productUniformDistributionOn{\region}{n - k}[\intD \vectorize{y}]
	\end{align*}
	for all $\vectorize{x} = (x_1, \dots, x_k) \in \region^{k}$.
\end{lemma}

\begin{proof}
	First, observe that $\samplingGraph \sim \canonicalDistribution{n}{\region}{\potential}$.
	As $n \ge 4 \frac{12^3}{\samplingError^{3}} \max\left\{\eulerE^6 \gppFugacity^2 \volumeMeasure[\region]^2, \ln\left(4 \frac{12}{\samplingError}\right)^2\right\}$, \Cref{thm:gpp_concentration} implies that $\Pr{\samplingGraph \in \approximationGraphs{n}{\frac{\samplingError}{12}}} \ge 1 - \frac{\samplingError}{12} > 0$.
	Therefore, conditioning on the event $\samplingGraph \in \approximationGraphs{n}{\frac{\samplingError}{12}}$ is well defined.

	Next, note that, for all $x \in \region$ it holds that
	\[
		\Pr{\countingMeasure[\{x\}] \ge 2}
		\ge \frac{\eulerE^{-\potential[x][x]} \gppFugacity^2 \volumeMeasure[\{x\}]^2}{\gppPartitonFunction{\region}[\gppFugacity][\potential]}
		\ge \frac{\eulerE^{-\potential[x][x]} \gppFugacity^2 \volumeMeasure[\{x\}]^2}{\eulerE^{\gppFugacity \volumeMeasure[\region]}}
	\]
	for $\countingMeasure \sim \GibbsPointProcess{\region}{\gppFugacity}{\potential}$.
	Thus, if $\GibbsPointProcess{\region}{\gppFugacity}{\potential}$ is simple (i.e., $\Pr{\countingMeasure[\{x\}] \ge 2}$ for all $x \in \region$), it holds that $\gppFugacity = 0$ or, for all $x \in \region$, $\volumeMeasure[\{x\}] = 0$ or $\potential[x][x] = \infty$.
	This implies that the output of \Cref{algo:modifed_sampling} is simple as well, and consequently $\modfiedSamplerOutput{\samplingError}$ is a simple point process.

	Knowing that $\modfiedSamplerOutput{\samplingError}$ is simple, \Cref{thm:pp_void_probabilities} implies that, in order to verify that $\modifiedSamplingDensity{\samplingError}$ is indeed a density for $\modfiedSamplerOutput{\samplingError}$, it suffices to prove that it yields the correct void probabilities.
	Formally, this means showing that for all bounded $B \in \Borel$ it holds that
	\[
		\Pr{\modifiedSampledPointSet \cap B = \emptyset}[\samplingGraph \in \approximationGraphs{n}{\frac{\samplingError}{12}}]
		= \int_{\countingMeasures} \ind{\countingMeasure[B] = 0} \modifiedSamplingDensity{\samplingError}[\countingMeasure] \PoissonPointProcess[\gppFugacity][\intD \countingMeasure]
	\]
	for $\modifiedSampledPointSet$ and $\samplingGraph$ as in \Cref{algo:modifed_sampling}.

	To prove this, we first write
	\begin{align*}
		\int_{\countingMeasures} \ind{\countingMeasure[B] = 0} \modifiedSamplingDensity{\samplingError}[\countingMeasure] \PoissonPointProcess[\gppFugacity][\intD \countingMeasure]
		&= \int_{\countingMeasures[\region]} \ind{\countingMeasure[B] = 0} \modifiedSamplingDensity{\samplingError}[\countingMeasure] \PoissonPointProcess[\gppFugacity][\intD \countingMeasure] \\
		&= \eulerE^{-\gppFugacity\volumeMeasure[\region]} \cdot \left( \modifiedSamplingDensity{\samplingError}[\zeroFunction] + \sum_{k \in \N_{\ge 1}} \frac{\gppFugacity^k}{k!} \int_{\region^k} \ind{\forall i \in [k]: x_i \notin B} \modifiedSamplingDensity{\samplingError}[\sum_{i \in [k]} \DiracMeasure{x_i}] \productVolumeMeasure{k}[\intD \vectorize{x}]\right) ,
	\end{align*}
	where $\zeroFunction$ denotes the constant $0$ measure on $\pointProcessSpace$.
	Note that
	\begin{align*}
		&\eulerE^{-\gppFugacity\volumeMeasure[\region]} \modifiedSamplingDensity{\samplingError}[\zeroFunction]\\
		&= \Pr{\samplingGraph \in \approximationGraphs{n}{\frac{\samplingError}{12}}}^{-1} \cdot \sum_{H \in \approximationGraphs{n}{\frac{\samplingError}{12}}} \frac{1}{\hcPartitionFunction{H}[\fugacity[n]]} 
		\bigintss_{\region^{n}}
		\left( \prod_{\substack{\{i, j\} \in \binom{[n]}{2}: \\\{i, j\} \in \edges_{H}}} 1 - \eulerE^{-\potential[y_i][y_j]}\right)
		\left( \prod_{\substack{\{i, j\} \in \binom{[n]}{2}: \\\{i, j\} \notin \edges_{H}}} \eulerE^{-\potential[y_i][y_j]}\right)
		\productUniformDistributionOn{\region}{n}[\intD \vectorize{y}] \\
		&= \Pr{\samplingGraph \in \approximationGraphs{n}{\frac{\samplingError}{12}}}^{-1} \cdot \sum_{H \in \approximationGraphs{n}{\frac{\samplingError}{12}}} \Pr{\configurationToSet{\spinConfigurationModified} = \emptyset}[\samplingGraph = H] \Pr{\samplingGraph = H} \\
		&= \frac{\Pr{\configurationToSet{\spinConfigurationModified} = \emptyset \wedge \samplingGraph \in \approximationGraphs{n}{\frac{\samplingError}{12}}}}{\Pr{\samplingGraph \in \approximationGraphs{n}{\frac{\samplingError}{12}}}} \\
		&= \Pr{\configurationToSet{\spinConfigurationModified} = \emptyset}[\samplingGraph \in \approximationGraphs{n}{\frac{\samplingError}{12}}]
	\end{align*}
	for $\spinConfigurationModified$ as in \Cref{algo:modifed_sampling}.
	We proceed by a case distinction based on $k$.
	For every $k > n$ and $(x_1, \dots, x_k) \in \region^k$ we have $\modifiedSamplingDensity{\samplingError}[\sum_{i \in [k]} \DiracMeasure{x_k}] = 0$.
	Therefore, we get
	\[
		\int_{\region^k} \ind{\forall i \in [k]: x_i \notin B} \modifiedSamplingDensity{\samplingError}[\sum_{i \in [k]} \DiracMeasure{x_i}] \productVolumeMeasure{k}[\intD \vectorize{x}] = 0
	\]
	for all $k > n$.
	Now, consider $k \in [n]$ and observe that for all $\vectorize{x} = (x_1, \dots, x_k) \in \region^k$ we have
	\[
		\densityNormalizingFunction{n}{k}[\countingMeasureToTuple[\sum_{i \in [k]} \DiracMeasure{x_i}]] = \densityNormalizingFunction{n}{k}[\vectorize{x}]
	\]
	by symmetry.
	Moreover, it holds that
	\[
		\frac{\gppFugacity^{k}}{k!} \left(\prod_{i=0}^{k - 1} 1 - \frac{i}{n} \right)
		= \binom{n}{k} \frac{\fugacity[n]^k}{\volumeMeasure[\region]^k} .
	\]
	Therefore, we have
	\begin{align*}
		\eulerE^{-\gppFugacity\volumeMeasure[\region]} \frac{\gppFugacity^{k}}{k!}
		&\int_{\region^k} \ind{\forall i \in [k]: x_i \notin B} \modifiedSamplingDensity{\samplingError}[\sum_{i \in [k]} \DiracMeasure{x_i}] \productVolumeMeasure{k}[\intD \vectorize{x}] \\
		&= \Pr{\samplingGraph \in \approximationGraphs{n}{\frac{\samplingError}{12}}}^{-1} \binom{n}{k} \fugacity[n]^k \bigintsss_{\region^k} \ind{\forall i \in [k]: x_i \notin B} \left(\prod_{{i, j} \in \binom{[k]}{2}} \eulerE^{- \potential[x_i][x_j]} \right)  \densityNormalizingFunction{n}{k}[\vectorize{x}] \productUniformDistributionOn{\region}{k}[\intD \vectorize{x}] .
	\end{align*}
	Next, note that
	\begin{align*}
		&\fugacity[n]^k
		\bigintsss_{\region^k} \ind{\forall i \in [k]: x_i \notin B} \left(\prod_{{i, j} \in \binom{[k]}{2}} \eulerE^{- \potential[x_i][x_j]} \right)  \densityNormalizingFunction{n}{k}[\vectorize{x}] \productUniformDistributionOn{\region}{k}[\intD \vectorize{x}] \\
		&\hspace{1em}= \sum_{H \in \approximationGraphs{n}{\frac{\samplingError}{12}}} \ind{[k] \in \independentSets{H}} \frac{\fugacity[n]^k}{\hcPartitionFunction{H}[\fugacity[n]]} \bigintsss_{\region^n} \ind{\forall i \in [k]: x_i \notin B}
			\left( \prod_{\substack{\{i, j\} \in \binom{[n]}{2}: \\\{i, j\} \in \edges_{H}}} 1 - \eulerE^{-\potential[x_i][x_j]}\right)
			\left( \prod_{\substack{\{i, j\} \in \binom{[n]}{2}: \\\{i, j\} \notin \edges_{H}}} \eulerE^{-\potential[x_i][x_j]}\right)
		\productUniformDistributionOn{\region}{n}[\intD \vectorize{x}] \\
		&\hspace{1em}= \sum_{H \in \approximationGraphs{n}{\frac{\samplingError}{12}}} \Pr{\configurationToSet{\spinConfigurationModified} = [k]}[\samplingGraph = H] \Pr{\samplingGraph = H \wedge \forall i \in [k]: \randomPoint_i \notin B}
	\end{align*}
	for $\randomPoint_1, \dots, \randomPoint_n$ as in \Cref{algo:modifed_sampling}.
	Furthermore, because the event $\configurationToSet{\spinConfigurationModified} = [k]$ is independent of $\randomPoint_1, \dots, \randomPoint_n$ given $\samplingGraph$, it holds that
	\begin{align*}
		&\sum_{H \in \approximationGraphs{n}{\frac{\samplingError}{12}}} \Pr{\configurationToSet{\spinConfigurationModified} = [k]}[\samplingGraph = H] \Pr{\samplingGraph = H \wedge \forall i \in [k]: \randomPoint_i \notin B} \\
		&\hspace{3em}= \sum_{H \in \approximationGraphs{n}{\frac{\samplingError}{12}}} \Pr{\configurationToSet{\spinConfigurationModified} = [k] \wedge \samplingGraph = H \wedge \forall i \in [k]: \randomPoint_i \notin B} \\
		&\hspace{3em}= \Pr{\configurationToSet{\spinConfigurationModified} = [k] \wedge \samplingGraph \in \approximationGraphs{n}{\frac{\samplingError}{12}} \wedge \forall i \in [k]: \randomPoint_i \notin B}
	\end{align*}
	and
	\begin{align*}
		\eulerE^{-\gppFugacity\volumeMeasure[\region]} \frac{\gppFugacity^{k}}{k!} \int_{\region^k} \ind{\forall i \in [k]: x_i \notin B} \modifiedSamplingDensity{\samplingError}[\sum_{i \in [k]} \DiracMeasure{x_i}] \productVolumeMeasure{k}[\intD \vectorize{x}]
		&= \binom{n}{k} \frac{\Pr{\configurationToSet{\spinConfigurationModified} = [k] \wedge \samplingGraph \in \approximationGraphs{n}{\frac{\samplingError}{12}} \wedge \forall i \in [k]: \randomPoint_i \notin B}}{\Pr{\samplingGraph \in \approximationGraphs{n}{\frac{\samplingError}{12}}}} \\
		&= \binom{n}{k} \Pr{\configurationToSet{\spinConfigurationModified} = [k] \wedge \forall i \in [k]: \randomPoint_i \notin B}[\samplingGraph \in \approximationGraphs{n}{\frac{\samplingError}{12}}] \\
		&= \sum_{\vertices' \in \binom{[n]}{k}} \Pr{\configurationToSet{\spinConfigurationModified} = \vertices' \wedge \forall i \in \vertices': \randomPoint_i \notin B}[\samplingGraph \in \approximationGraphs{n}{\frac{\samplingError}{12}}] ,
	\end{align*}
	where the last equality is due to symmetry.
	Combining everything yields
	\begin{align*}
		\int_{\countingMeasures} \ind{\countingMeasure[B] = 0} \modifiedSamplingDensity{\samplingError}[\countingMeasure] \PoissonPointProcess[\gppFugacity][\intD \countingMeasure]
		&= \Pr{\configurationToSet{\spinConfigurationModified} = \emptyset}[\samplingGraph \in \approximationGraphs{n}{\frac{\samplingError}{12}}] + \sum_{k=1}^{n} \sum_{\vertices' \in \binom{[n]}{k}} \Pr{\configurationToSet{\spinConfigurationModified} = \vertices' \wedge \forall i \in \vertices': \randomPoint_i \notin B}[\samplingGraph \in \approximationGraphs{n}{\frac{\samplingError}{12}}] \\
		&= \sum_{\vertices' \in \powerset{[n]}} \Pr{\configurationToSet{\spinConfigurationModified} = \vertices' \wedge \forall i \in \vertices': \randomPoint_i \notin B}[\samplingGraph \in \approximationGraphs{n}{\frac{\samplingError}{12}}] \\
		&= \Pr{\forall i \in \configurationToSet{\spinConfigurationModified}: \randomPoint_i \notin B}[\samplingGraph \in \approximationGraphs{n}{\frac{\samplingError}{12}}] \\
		&= \Pr{\modifiedSampledPointSet \cap B = \emptyset}[\samplingGraph \in \approximationGraphs{n}{\frac{\samplingError}{12}}],
	\end{align*}
	which concludes the proof.
\end{proof}

We proceed by upper and lower bounding the density $\modifiedSamplingDensity{\samplingError}[\countingMeasure]$ in terms of the density of $\GibbsPointProcess{\region}{\gppFugacity}{\potential}$.
To this end, we use the following basic facts about the partition function of the hard-core model.

\begin{observation}[{see \cite{friedrich2021algorithms}}]
	\label{obs:hardcore_bounds}
	For every undirected graph $\graph = (\vertices, \edges)$ the following holds:
	\begin{enumerate}[1.]
		\item For all $\fugacity_1, \fugacity_2 \in \R_{\ge 0}$
		\[
			\hcPartitionFunction{\graph}[\fugacity_1] \le \hcPartitionFunction{\graph}[\fugacity_1 + \fugacity_2] \le \eulerE^{\fugacity_2 \size{\vertices}} \hcPartitionFunction{\graph}[\fugacity_1] .
		\]
		\item For all $\fugacity \in \R_{\ge 0}$ and $S \subseteq \vertices$
		\[
			\hcPartitionFunction{\graph - S}[\fugacity] \le \hcPartitionFunction{\graph}[\fugacity] \le \eulerE^{\fugacity \size{S}}\hcPartitionFunction{\graph - S}[\fugacity],
		\]
		where $\graph - S$ denotes the subgraph of $\graph$ that is induced by $\vertices \setminus S$.\qedhere
	\end{enumerate}
\end{observation}

Using \Cref{obs:hardcore_bounds} we derive the following bounds.
\begin{lemma}
	\label{lemma:conditional_density_bound}
	Consider the setting of \Cref{lemma:conditional_density} and let $\gppDensity$ denote the density of $\GibbsPointProcess{\region}{\gppFugacity}{\potential}$ with respect to $\PoissonPointProcess[\gppFugacity]$.
	For $n$ as in \Cref{algo:modifed_sampling} and all $\countingMeasure \in \countingMeasures$ with $\countingMeasure[\region] \le \min\left\{\sqrt{\frac{\samplingError n}{12}}, \frac{\samplingError}{40 \gppFugacity \volumeMeasure[\region] + \samplingError} n\right\}$ it holds that
	\[
		\left(1 - \frac{\samplingError}{4}\right) \gppDensity[\countingMeasure] \le \modifiedSamplingDensity{\samplingError}[\countingMeasure] \le \left(1 + \frac{\samplingError}{4}\right) \gppDensity[\countingMeasure] .
		\qedhere
	\]
\end{lemma}

\begin{proof}
	First, recall that, when $\GibbsPointProcess{\region}{\gppFugacity}{\potential}$ is simple, its density with respect to $\PoissonPointProcess[\gppFugacity]$ can be expressed as
	\[
		\gppDensity[\countingMeasure] = \frac{1}{\gppPartitonFunction{\region}[\gppFugacity][\potential]} \ind{\countingMeasure \in \countingMeasures[\region]}  \left(\prod_{\{x, y\} \in \binom{\pointSet[\countingMeasure]}{2}} \eulerE^{- \countFunction{x}[\countingMeasure] \countFunction{y}[\countingMeasure] \potential[x][y]}\right) \left(\prod_{x \in \pointSet[\countingMeasure]} \eulerE^{- \frac{\countFunction{x}[\countingMeasure] (\countFunction{x}[\countingMeasure]-1)}{2} \potential[x][x]}\right) \eulerE^{\gppFugacity \volumeMeasure[\region]}
	\]
	for every $\countingMeasure \in \countingMeasures$.
	Therefore, we have
	\[
		\modifiedSamplingDensity{\samplingError}[\countingMeasure] = \Pr{\samplingGraph \in \approximationGraphs{n}{\frac{\samplingError}{12}}}^{-1} \left(\prod_{i=0}^{\countingMeasure[\region] - 1} 1 - \frac{i}{n} \right) \ind{\countingMeasure[\region] \le n} \densityNormalizingFunction{n}{\countingMeasure[\region]}[\countingMeasureToTuple[\countingMeasure]] \gppPartitonFunction{\region}[\gppFugacity][\potential] \gppDensity[\countingMeasure] .
	\]
	As we focus on $\countingMeasure$ with $\countingMeasure[\region] \le \sqrt{\frac{\samplingError n}{12}} \le n$, we omit the indicator $\ind{\countingMeasure[\region] \le n}$ from now on.

	We proceed by deriving an upper bound on $\modifiedSamplingDensity{\samplingError}[\countingMeasure]$ for $\countingMeasure \in \countingMeasures$ with $\countingMeasure[\region] \le \min\left\{\sqrt{\frac{\samplingError n}{12}}, \frac{\samplingError}{40 \gppFugacity \volumeMeasure[\region] + \samplingError} n\right\}$.
	To this end, note that
	\[
		\left(\prod_{i=0}^{\countingMeasure[\region] - 1} 1 - \frac{i}{n} \right) \le 1.
	\]
	Moreover, for $\samplingGraph \sim \canonicalDistribution{n}{\region}{\potential}$ and $n \ge 4 \frac{12^3}{\samplingError^{3}} \max\left\{\eulerE^6 \gppFugacity^2 \volumeMeasure[\region]^2, \ln\left(4 \frac{12}{\samplingError}\right)^2\right\}$ \Cref{thm:gpp_concentration} yields $\Pr{\samplingGraph \in \approximationGraphs{n}{\frac{\samplingError}{12}}} \ge 1 - \frac{\samplingError}{12}$.
	Finally, observe that for all $\vectorize{x} \in \region^{k}$ for $k \le n$ we have
	\begin{align*}
		\densityNormalizingFunction{n}{k}[\vectorize{x}]
		&\le \frac{1}{\left(1 - \frac{\samplingError}{12}\right) \gppPartitonFunction{\region}[\gppFugacity][\potential]} \sum_{\substack{H \in \approximationGraphs{n}{\frac{\samplingError}{12}}:\\ [k] \in \independentSets{H}}} \bigintss_{\region^{n - k}}
		\left( \prod_{\substack{(i, j) \in [k] \times [n - k]: \\\{i, j + k\} \in \edges_{H}}} 1 - \eulerE^{-\potential[x_i][y_j]}\right)
		\left( \prod_{\substack{(i, j) \in [k] \times [n - k]: \\\{i, j + k\} \notin \edges_{H}}} \eulerE^{-\potential[x_i][y_j]}\right) \\
		&\hspace*{12em}\cdot \left( \prod_{\substack{\{i, j\} \in \binom{[n - k]}{2}: \\\{i + k, j + k\} \in \edges_{H}}} 1 - \eulerE^{-\potential[y_i][y_j]}\right)
		\left( \prod_{\substack{\{i, j\} \in \binom{[n - k]}{2}: \\\{i + k, j + k\} \notin \edges_{H}}} \eulerE^{-\potential[y_i][y_j]}\right) \productUniformDistributionOn{\region}{n - k}[\intD \vectorize{y}] \\
		&= \frac{1}{\left(1 - \frac{\samplingError}{12}\right) \gppPartitonFunction{\region}[\gppFugacity][\potential]} \bigintss_{\region^{n - k}} \sum_{\substack{H \in \approximationGraphs{n}{\frac{\samplingError}{12}}:\\ [k] \in \independentSets{H}}}
		\left( \prod_{\substack{(i, j) \in [k] \times [n - k]: \\\{i, j + k\} \in \edges_{H}}} 1 - \eulerE^{-\potential[x_i][y_j]}\right)
		\left( \prod_{\substack{(i, j) \in [k] \times [n - k]: \\\{i, j + k\} \notin \edges_{H}}} \eulerE^{-\potential[x_i][y_j]}\right) \\
		&\hspace*{12em}\cdot\left( \prod_{\substack{\{i, j\} \in \binom{[n - k]}{2}: \\\{i + k, j + k\} \in \edges_{H}}} 1 - \eulerE^{-\potential[y_i][y_j]}\right)
		\left( \prod_{\substack{\{i, j\} \in \binom{[n - k]}{2}: \\\{i + k, j + k\} \notin \edges_{H}}} \eulerE^{-\potential[y_i][y_j]}\right) \productUniformDistributionOn{\region}{n - k}[\intD \vectorize{y}] \\
		&\le \frac{1}{\left(1 - \frac{\samplingError}{12}\right) \gppPartitonFunction{\region}[\gppFugacity][\potential]} \int_{\region^{n - k}} 1\, \productUniformDistributionOn{\region}{n - k}[\intD \vectorize{y}] \\
		&\le \frac{1}{\left(1 - \frac{\samplingError}{12}\right) \gppPartitonFunction{\region}[\gppFugacity][\potential]} .
	\end{align*}
	Given that $\samplingError \le 1$ we get
	\[
		\modifiedSamplingDensity{\samplingError}[\countingMeasure]
		\le \left(1 - \frac{\samplingError}{12}\right)^{-2} \gppDensity[\countingMeasure]
		\le \left(1 + \frac{\samplingError}{11}\right)^{2} \gppDensity[\countingMeasure]
		\le \left(1 + \frac{\samplingError}{4}\right) \gppDensity[\countingMeasure] ,
	\]
	which proves the upper bound.

	For the lower bound, note that
	\[
		\Pr{\samplingGraph \in \approximationGraphs{n}{\frac{\samplingError}{12}}}^{-1} \ge 1
	\]
	and for $\countingMeasure[\region] \le \sqrt{\frac{\samplingError n}{12}}$
	\[
		\left(\prod_{i=0}^{\countingMeasure[\region] - 1} 1 - \frac{i}{n} \right)
		\ge \left(1 - \frac{\countingMeasure[\region]}{n}\right)^{\countingMeasure[\region]}
		\ge 1 - \frac{\countingMeasure[\region]^2}{n}
		\ge 1 - \frac{\samplingError}{12}.
	\]
	We proceed by lower bounding $\densityNormalizingFunction{n}{k}[\vectorize{x}]$.
	First, observe that
	\begin{align*}
		\densityNormalizingFunction{n}{k}[\vectorize{x}]
		&\ge \frac{1}{\left(1 + \frac{\samplingError}{12}\right) \gppPartitonFunction{\region}[\gppFugacity][\potential]} \sum_{\substack{H \in \approximationGraphs{n}{\frac{\samplingError}{12}}:\\ [k] \in \independentSets{H}}} \bigintss_{\region^{n - k}}
		\left( \prod_{\substack{(i, j) \in [k] \times [n - k]: \\\{i, j + k\} \in \edges_{H}}} 1 - \eulerE^{-\potential[x_i][y_j]}\right)
		\left( \prod_{\substack{(i, j) \in [k] \times [n - k]: \\\{i, j + k\} \notin \edges_{H}}} \eulerE^{-\potential[x_i][y_j]}\right) \\
		&\hspace*{14em}\cdot\left( \prod_{\substack{\{i, j\} \in \binom{[n - k]}{2}: \\\{i + k, j + k\} \in \edges_{H}}} 1 - \eulerE^{-\potential[y_i][y_j]}\right)
		\left( \prod_{\substack{\{i, j\} \in \binom{[n - k]}{2}: \\\{i + k, j + k\} \notin \edges_{H}}} \eulerE^{-\potential[y_i][y_j]}\right) \productUniformDistributionOn{\region}{n - k}[\intD \vectorize{y}] \\
		&= \frac{1}{\left(1 + \frac{\samplingError}{12}\right) \gppPartitonFunction{\region}[\gppFugacity][\potential]} \sum_{\substack{H \in \graphs{n}:\\ [k] \in \independentSets{H}}} \ind{H \in \approximationGraphs{n}{\frac{\samplingError}{12}}} \bigintss_{\region^{n - k}}
		\left( \prod_{\substack{(i, j) \in [k] \times [n - k]: \\\{i, j + k\} \in \edges_{H}}} 1 - \eulerE^{-\potential[x_i][y_j]}\right)
		\left( \prod_{\substack{(i, j) \in [k] \times [n - k]: \\\{i, j + k\} \notin \edges_{H}}} \eulerE^{-\potential[x_i][y_j]}\right) \\
		&\hspace*{14em}\cdot\left( \prod_{\substack{\{i, j\} \in \binom{[n - k]}{2}: \\\{i + k, j + k\} \in \edges_{H}}} 1 - \eulerE^{-\potential[y_i][y_j]}\right)
		\left( \prod_{\substack{\{i, j\} \in \binom{[n - k]}{2}: \\\{i + k, j + k\} \notin \edges_{H}}} \eulerE^{-\potential[y_i][y_j]}\right) \productUniformDistributionOn{\region}{n - k}[\intD \vectorize{y}].
	\end{align*}
	Next, for each graph $H \in \graphs{n}$, let $H' = ([n-k], E')$ denote the subgraph that results from $H - [k]$ after relabeling each vertex in $i \in [n] \setminus [k]$ to $i - k \in [n-k]$ (note that this relabeling is formally required for $H' \in \graphs{n-k}$).
	By \Cref{obs:hardcore_bounds} and the fact that $\fugacity[n] \le \fugacity[n-k]$ and $\hcPartitionFunction{H'}[\fugacity] = \hcPartitionFunction{H - [k]}[\fugacity]$ for all $\fugacity \in \R_{\ge 0}$ we have
	\[
		\hcPartitionFunction{H}[\fugacity[n]]
		\le \eulerE^{\fugacity[n] k}\hcPartitionFunction{H'}[\fugacity[n]] \le \eulerE^{\frac{k}{n} \gppFugacity \volumeMeasure[\region]} \hcPartitionFunction{H'}[\fugacity[n-k]] .
	\]
	On the other hand, note that
	\begin{align*}
		\fugacity[n - k] = \frac{\gppFugacity \volumeMeasure[\region]}{n - k}
		&= \frac{n}{n(n - k)} \gppFugacity \volumeMeasure[\region]
		= \left(\frac{n - k}{n(n - k)} + \frac{k}{n(n-k)}\right) \gppFugacity \volumeMeasure[\region] \\
		&= \left(\frac{1}{n} + \frac{k}{n(n-k)}\right) \gppFugacity \volumeMeasure[\region]
		= \fugacity[n] + \frac{k}{n(n-k)} \gppFugacity \volumeMeasure[\region].
	\end{align*}
	Therefore, \Cref{obs:hardcore_bounds} yields
	\[
		\hcPartitionFunction{H}[\fugacity[n-k]] \le \eulerE^{\frac{k}{n-k} \gppFugacity \volumeMeasure[\region]} \hcPartitionFunction{H}[\fugacity[n]]
	\]
	and
	\[
		\hcPartitionFunction{H}[\fugacity[n]] \ge \eulerE^{-\frac{k}{n-k} \gppFugacity \volumeMeasure[\region]} \hcPartitionFunction{H}[\fugacity[n-k]]
		\ge \eulerE^{-\frac{k}{n-k} \gppFugacity \volumeMeasure[\region]} \hcPartitionFunction{H'}[\fugacity[n-k]].
	\]
	Thus, for $k \le \frac{\samplingError}{40 \gppFugacity \volumeMeasure[\region] + \samplingError} n$ we have
	\[
		\eulerE^{-\frac{\samplingError}{40}} \hcPartitionFunction{H'}[\fugacity[n-k]] \le \hcPartitionFunction{H}[\fugacity[n]] \le  \eulerE^{\frac{\samplingError}{40}}\hcPartitionFunction{H'}[\fugacity[n-k]] .
	\]
	As $\eulerE^{-\frac{\samplingError}{40}}\left(1 - \frac{\samplingError}{18}\right) \ge \left(1 - \frac{\samplingError}{12}\right)$ and $\eulerE^{\frac{\samplingError}{40}}\left(1 + \frac{\samplingError}{18}\right) \le \left(1 + \frac{\samplingError}{12}\right)$ for all $\samplingError \in [0, 1]$, this means that $H' \in \approximationGraphs{n-k}{\frac{\samplingError}{18}}$ is a sufficient condition for $H \in \approximationGraphs{n}{\frac{\samplingError}{12}}$ and
	\begin{align*}
		&\sum_{\substack{H \in \graphs{n}:\\ [k] \in \independentSets{H}}}
		\ind{H \in \approximationGraphs{n}{\frac{\samplingError}{12}}} \bigintss_{\region^{n - k}}
		\left( \prod_{\substack{(i, j) \in [k] \times [n - k]: \\\{i, j + k\} \in \edges_{H}}} 1 - \eulerE^{-\potential[x_i][y_j]}\right)
		\left( \prod_{\substack{(i, j) \in [k] \times [n - k]: \\\{i, j + k\} \notin \edges_{H}}} \eulerE^{-\potential[x_i][y_j]}\right) \\
		&\hspace*{10em}\left( \prod_{\substack{\{i, j\} \in \binom{[n - k]}{2}: \\\{i + k, j + k\} \in \edges_{H}}} 1 - \eulerE^{-\potential[y_i][y_j]}\right)
		\left( \prod_{\substack{\{i, j\} \in \binom{[n - k]}{2}: \\\{i + k, j + k\} \notin \edges_{H}}} \eulerE^{-\potential[y_i][y_j]}\right) \productUniformDistributionOn{\region}{n - k}[\intD \vectorize{y}] \\
		& \hspace*{2em} \ge  \sum_{H' \in \graphs{n-k}}
		\ind{H' \in \approximationGraphs{n-k}{\frac{\samplingError}{18}}} \bigintss_{\region^{n - k}}
		\left( \prod_{\substack{\{i, j\} \in \binom{[n - k]}{2}: \\\{i,j\} \in \edges_{H'}}} 1 - \eulerE^{-\potential[y_i][y_j]}\right)
		\left( \prod_{\substack{\{i, j\} \in \binom{[n - k]}{2}: \\\{i, j\} \notin \edges_{H'}}} \eulerE^{-\potential[y_i][y_j]}\right) \\
		&\hspace*{10em} \sum_{F \subseteq [k]\times[n-k]}
		\left( \prod_{\substack{(i, j) \in [k] \times [n - k]: \\(i, j) \in F}} 1 - \eulerE^{-\potential[x_i][y_j]}\right)
		\left( \prod_{\substack{(i, j) \in [k] \times [n - k]: \\(i, j) \notin F}} \eulerE^{-\potential[x_i][y_j]}\right)
		\productUniformDistributionOn{\region}{n - k}[\intD \vectorize{y}]\\
		& \hspace{2em} =  \sum_{H' \in \graphs{n-k}}
		\ind{H' \in \approximationGraphs{n-k}{\frac{\samplingError}{18}}} \bigintss_{\region^{n - k}}
		\left( \prod_{\substack{\{i, j\} \in \binom{[n - k]}{2}: \\\{i,j\} \in \edges_{H'}}} 1 - \eulerE^{-\potential[y_i][y_j]}\right)
		\left( \prod_{\substack{\{i, j\} \in \binom{[n - k]}{2}: \\\{i, j\} \notin \edges_{H'}}} \eulerE^{-\potential[y_i][y_j]}\right)
		\productUniformDistributionOn{\region}{n - k}[\intD \vectorize{y}] \\
		& \hspace{2em} =  \Pr{\samplingGraph' \in \approximationGraphs{n-k}{\frac{\samplingError}{18}}}
	\end{align*}
	for $\samplingGraph' \sim \canonicalDistribution{n-k}{\region}{\potential}$.
	Next, observe that $n  \ge 1$ we have $k \le \min\left\{\sqrt{\frac{\samplingError n}{12}}, \frac{\samplingError}{40 \gppFugacity \volumeMeasure[\region] + \samplingError} n\right\} \le \frac{n}{2}$ and $n - k \ge \frac{n}{2}$.
	Therefore, for $n \ge 8 \frac{18^2 \cdot 12}{\samplingError^3} \max\left\{\eulerE^6 \gppFugacity^2 \volumeMeasure[\region]^2, \ln\left(\frac{4 \cdot 18}{\samplingError}\right)^2\right\}$ \Cref{thm:gpp_concentration} yields $\Pr{\samplingGraph' \in \approximationGraphs{n-k}{\frac{\samplingError}{18}}} \ge 1 - \frac{\samplingError}{12}$ for $\samplingGraph' \sim \canonicalDistribution{n-k}{\region}{\potential}$.
	Consequently, we have
	\[
		\densityNormalizingFunction{n}{k}[\vectorize{x}]
		\ge \frac{1 - \frac{\samplingError}{12}}{1 + \frac{\samplingError}{12}} \cdot \frac{1}{\gppPartitonFunction{\region}[\gppFugacity][\potential]}
	\]
	and
	\[
		\modifiedSamplingDensity{\samplingError}[\countingMeasure]
		\ge \left(1 - \frac{\samplingError}{12}\right)^{2} \left(1 + \frac{\samplingError}{12}\right)^{-1} \gppDensity[\countingMeasure]
		\ge \left(1 - \frac{\samplingError}{12}\right)^{3} \gppDensity[\countingMeasure]
		\ge \left(1 - \frac{\samplingError}{4}\right) \gppDensity[\countingMeasure],
	\]
	which concludes the proof.
\end{proof}

We proceed by using \Cref{lemma:conditional_density,lemma:conditional_density_bound} to bound the total variation distance between $\GibbsPointProcess{\region}{\gppFugacity}{\potential}$ and the output distribution of \Cref{algo:modifed_sampling}.
However, as \Cref{lemma:conditional_density_bound} only provides information for point sets that are sufficiently small compared to $n$, we need a different way to deal with large point configurations.
To this end, the following two results are useful.
The first lemma is a domination result for the size of independent sets, drawn from a hard-core model.

\begin{lemma}
	\label{lemma:domination_hc}
	Let $\graph \in \graphs{n}$ for some $n \in \N$ and let $\fugacity \in \R_{\ge 0}$.
	For $\spinConfiguration \sim \hcGibbsDistribution{\graph}{\fugacity}$ it holds that $\size{\configurationToSet{\spinConfiguration}}$ is stochastically dominated by a binomial random variable with $n$ trials and success probability $\frac{\fugacity}{1 + \fugacity}$.
\end{lemma}

\begin{proof}
	We use a coupling argument to prove this statement.
	Consider the following procedure for sampling a set $S_n \subseteq [n]$:
	\begin{enumerate}
		\item Start with $S_0 = \emptyset$.
		\item For each $i \in [n]$, set $S_i = S_{i-1} \cup \{i\}$ with probability $\Pr{\spinConfiguration[i] = 1}[\bigwedge_{j \in [i-1]} \spinConfiguration[j] = \ind{j \in S_{i-1}}]$ for $\spinConfiguration \sim \hcGibbsDistribution{\graph}{\fugacity}$.
	\end{enumerate}
	Note that the resulting set $S_n$ follows the same distribution as $S_{\spinConfiguration}$ for $\spinConfiguration \sim \hcGibbsDistribution{\graph}{\fugacity}$.
	Due to the definition of this process, it suffices to consider sequences $(S_i)_{i \in [n] \cup \{0\}}$ such that the event $\{\bigwedge_{j \in [i-1]} \spinConfiguration[j] = \ind{j \in S_{i-1}}\}$ has non-zero probability.
	Further, note that
	\[
		\Pr{\spinConfiguration[i] = 1}[\bigwedge_{j \in [i-1]} \spinConfiguration[j] = \ind{j \in S_{i-1}}] \le \frac{\fugacity}{1 + \fugacity}
	\]
	for all $i \in [n]$.
	Now, we consider a modified process $(S'_i)_{i \in [n] \cup \{0\}}$ with $S'_0 = \emptyset$ and $S'_i = S'_{i-1} \cup \{i\}$ with probability $\frac{\fugacity}{1 + \fugacity}$.
	Observe that $(S_i)_{i \in [n] \cup \{0\}}$ and $(S'_i)_{i \in [n] \cup \{0\}}$ can be coupled in such a way that $S_i \subseteq S'_i$ whenever $S_{i-1} \subseteq S'_{i-1}$ for all $i \in [n]$.
	As initially $S_0 = S'_0$, the same coupling yields $S_n \subseteq S'_n$.
	Finally, observing that $\size{S'_n}$ follows a binomial distribution with $n$ trials and success probability $\frac{\fugacity}{1 + \fugacity}$ concludes the proof.
\end{proof}

The second lemma is the analog of \Cref{lemma:domination_hc} for repulsive point processes.
However, proving it is slightly more technically involved.
We start by introducing some additional notation and terminology.
For two counting measures $\countingMeasure_1, \countingMeasure_2 \in \countingMeasures$, we write $\countingMeasure_1 \le \countingMeasure_2$ if $\countingMeasure_1(B) \le \countingMeasure_2(B)$ for every $B \in \Borel$.
A measurable function $h: \countingMeasures \to \R$ is called \emph{increasing} if $h(\countingMeasure_1) \le h(\countingMeasure_2)$ for all $\countingMeasure_1 \le \countingMeasure_2$.
Moreover, for some $\PoissonIntensity \in \R_{\ge 0}$, let $\PoissonPointProcess[\PoissonIntensity]$ denote the Poisson point process with intensity $\PoissonIntensity$ and let $\pointProcess$ be a point process that has a density $f_{\pointProcess}$ with respect to $\PoissonPointProcess[\PoissonIntensity]$.
A function $\PapangelouIntensity: \countingMeasures \times \pointProcessSpace \to \R_{\ge 0}$ is called a \emph{Papangelou intensity} for $\pointProcess$ (w.r.t. $\PoissonPointProcess[\PoissonIntensity]$) if, for all $\countingMeasure \in \countingMeasures$ and $x \in \pointProcessSpace$, it holds that
\[
	f_{\pointProcess}(\countingMeasure + \DiracMeasure{x}) = \PapangelouIntensity[][\countingMeasure][x] f_{\pointProcess}(\countingMeasure) .
\]

The domination lemma we are aiming for is implied by the following result.

\begin{theorem}[{\cite[Theorem $1.1$]{georgii1997stochastic}}]
	\label{thm:PapangelouIntensity}
	Let $\PoissonPointProcess[\PoissonIntensity]$ be a Poisson point process of intensity $\PoissonIntensity \in \R_{\ge 0}$ and let $\pointProcess_1, \pointProcess_2$ be point processes that are absolutely continuous with respect to $\PoissonPointProcess[\PoissonIntensity]$.
	Assume $\pointProcess_1$ and $\pointProcess_2$ have Papangelou intensities $\PapangelouIntensity[1]$ and $\PapangelouIntensity[2]$.
	If, for all $x \in \pointProcessSpace$ and $\countingMeasure_1, \countingMeasure_2 \in \countingMeasures$ with $\countingMeasure_1 \le \countingMeasure_2$, $\PapangelouIntensity[1][\countingMeasure_1][x] \le \PapangelouIntensity[2][\countingMeasure_2][x]$, then, for all increasing $h: \countingMeasures \to \R$, it holds that
	\[
		\int_{\countingMeasures} h(\countingMeasure) \pointProcess_1(\intD \countingMeasure) \le \int_{\countingMeasures} h(\countingMeasure) \pointProcess_2(\intD \countingMeasure). \qedhere
	\]
\end{theorem}

With that, we show the following simple domination result.

\begin{lemma}
	\label{lemma:domination_gpp}
	Let $(\pointProcessSpace, \dist)$ be a complete, separable metric space, let $\Borel = \Borel[\pointProcessSpace]$ be the Borel algebra and let $\volumeMeasure$ be a locally finite reference measure on $(\pointProcessSpace, \Borel)$.
	Let $\region \subseteq \pointProcessSpace$ be bounded and measurable, let $\gppFugacity \in \R_{\ge 0}$ and let $\potential: \pointProcessSpace^2 \to \R_{\ge 0} \cup \{\infty\}$ be a symmetric repulsive potential.
	For $\countingMeasure \sim \GibbsPointProcess{\region}{\gppFugacity}{\potential}$ it holds that $\countingMeasure[\region]$ is dominated by a Poisson random variable with parameter $\gppFugacity \volumeMeasure[\region]$.
\end{lemma}

\begin{proof}
	Let $\PoissonPointProcess[\gppFugacity]$ denote a Poisson point process with intensity $\gppFugacity$.
	Note that
	a density of
	\[
		f_1(\countingMeasure) = \frac{1}{\gppPartitonFunction{\region}[\gppFugacity][\potential]} \ind{\countingMeasure \in \countingMeasures[\region]}\left(\prod_{\{x, y\} \in \binom{\pointSet[\countingMeasure]}{2}} \eulerE^{- \countFunction{x}[\countingMeasure] \countFunction{y}[\countingMeasure] \potential[x][y]}\right) \left(\prod_{x \in \pointSet[\countingMeasure]} \eulerE^{- \frac{\countFunction{x}[\countingMeasure] (\countFunction{x}[\countingMeasure]-1)}{2} \potential[x][x]}\right) \eulerE^{\gppFugacity \volumeMeasure[\region]}
	\]
	is a density for $\GibbsPointProcess{\region}{\gppFugacity}{\potential}$ with respect to $\PoissonPointProcess[\gppFugacity]$.
	Therefore,
	\[
		\PapangelouIntensity[1][\countingMeasure][x] = \ind{x \in \region} \prod_{y \in \pointSet[\countingMeasure]} \eulerE^{- \countFunction{y}[\countingMeasure] \potential[x][y]}
	\]
	is a Papangelou intensity for
	$\GibbsPointProcess{\region}{\gppFugacity}{\potential}$.
	Moreover, let $\pointProcess$ denote the point process defined by the density $f_2(\countingMeasure) = \ind{\countingMeasure \in \countingMeasures[\region]}$ and observe that $\PapangelouIntensity[2][\countingMeasure][x] = \ind{x \in \region}$ is a Papangelou intensity for $\pointProcess$.

	For all $k \in \N$, let $h_k(\countingMeasure) = \ind{\countingMeasure[\region] \ge k}$ and observe that $h_k$ is increasing.
	Further, note that, for all $x \in \pointProcessSpace$ and $\countingMeasure_1, \countingMeasure_2 \in \countingMeasures$, it holds that
	\[
		\PapangelouIntensity[1][\countingMeasure_1][x]
		= \ind{x \in \region} \prod_{y \in \pointSet[\countingMeasure_1]} \eulerE^{- \countFunction{y}[\countingMeasure_1] \potential[x][y]}
		\le \ind{x \in \region}
		= \PapangelouIntensity[2][\countingMeasure_2][x] .
	\]
	By \Cref{thm:PapangelouIntensity}, this implies that for all $k \in \N$
	\[
		\int_{\countingMeasures} h_k(\countingMeasure) \GibbsPointProcess{\region}{\gppFugacity}{\potential}(\intD \countingMeasure) \le \int_{\countingMeasures} h_k(\countingMeasure) \pointProcess(\intD \countingMeasure).
	\]
	Consequently, for $\countingMeasure \sim \GibbsPointProcess{\region}{\gppFugacity}{\potential}$ and $\xi \sim \pointProcess$ and for all $k \in \N$, it holds that
	\[
		\Pr{\countingMeasure[\region] \ge k} \le \Pr{\xi(\region) \ge k}
	\]
	and observing that $\xi(\region)$ follows a Poisson distribution with parameter $\gppFugacity \volumeMeasure[\region]$ concludes the proof.
\end{proof}

We now bound the total variation distance between the output of \Cref{algo:modifed_sampling} and $\GibbsPointProcess{\region}{\gppFugacity}{\potential}$.

\begin{lemma}
	\label{lemma:modified_sampler_dtv}
	Let $(\pointProcessSpace, \dist)$ be a complete, separable metric space, let $\Borel = \Borel[\pointProcessSpace]$ be the Borel algebra and let $\volumeMeasure$ be a locally finite reference measure on $(\pointProcessSpace, \Borel)$.
	Let $\region \subseteq \pointProcessSpace$ be bounded and measurable, let $\gppFugacity \in \R_{\ge 0}$ and let $\potential: \pointProcessSpace^2 \to \R_{\ge 0} \cup \{\infty\}$ be a symmetric repulsive potential.
	For every given $\samplingError \in  (0, 1]$, \Cref{algo:modifed_sampling} is an $\frac{\samplingError}{2}$-approximate sampler from $\GibbsPointProcess{\region}{\gppFugacity}{\potential}$.
\end{lemma}

\begin{proof}
	We start by bounding the total variation distance between $\dtv{\GibbsPointProcess{\region}{\gppFugacity}{\potential}}{\modfiedSamplerOutput{\samplingError}}$ for $\modfiedSamplerOutput{\samplingError}$ as in \Cref{lemma:conditional_density}.
	The statement then follows from a coupling argument.
	Let $\PoissonPointProcess[\gppFugacity]$ denote a Poisson point process of intensity $\gppFugacity$.
	Let $\modifiedSamplingDensity{\samplingError}$ be the density of $\modfiedSamplerOutput{\samplingError}$ with respect to $\PoissonPointProcess[\gppFugacity]$ as given in \Cref{lemma:conditional_density} and let $\gppDensity$ denote the density of $\GibbsPointProcess{\region}{\gppFugacity}{\potential}$ with respect to $\PoissonPointProcess$.
	Moreover, set $m = \min\left\{\sqrt{\frac{\samplingError n}{12}}, \frac{\samplingError}{40 \gppFugacity \volumeMeasure[\region] + \samplingError} n\right\}$.
	Note that the total variation distance can be expressed as
	\begin{align*}
		\dtv{\GibbsPointProcess{\region}{\gppFugacity}{\potential}}{\modfiedSamplerOutput{\samplingError}}
		&= \int_{\countingMeasures} \absolute{\gppDensity[\countingMeasure] - \modifiedSamplingDensity{\samplingError}[\countingMeasure]} \PoissonPointProcess[\gppFugacity][\intD \countingMeasure] \\
		&= \int_{\countingMeasures} \ind{\countingMeasure[\region] \le m} \absolute{\gppDensity[\countingMeasure] - \modifiedSamplingDensity{\samplingError}[\countingMeasure]} \PoissonPointProcess[\gppFugacity][\intD \countingMeasure] + \int_{\countingMeasures} \ind{\countingMeasure[\region] > m} \absolute{\gppDensity[\countingMeasure] - \modifiedSamplingDensity{\samplingError}[\countingMeasure]} \PoissonPointProcess[\gppFugacity][\intD \countingMeasure] .
	\end{align*}
	By \Cref{lemma:conditional_density_bound}, we get
	\[
		\int_{\countingMeasures} \ind{\countingMeasure[\region] \le m} \absolute{\gppDensity[\countingMeasure] - \modifiedSamplingDensity{\samplingError}[\countingMeasure]} \PoissonPointProcess[\gppFugacity][\intD \countingMeasure]
		\le \frac{\samplingError}{4} \int_{\countingMeasures} \ind{\countingMeasure[\region] \le m} \gppDensity[\countingMeasure] \PoissonPointProcess[\gppFugacity][\intD \countingMeasure]
		\le \frac{\samplingError}{4} .
	\]
	Further, it holds that
	\begin{align*}
		\int_{\countingMeasures} \ind{\countingMeasure[\region] > m} \absolute{\gppDensity[\countingMeasure] - \modifiedSamplingDensity{\samplingError}[\countingMeasure]} \PoissonPointProcess[\gppFugacity][\intD \countingMeasure]
		&\le \int_{\countingMeasures} \ind{\countingMeasure[\region] > m} \gppDensity[\countingMeasure] \PoissonPointProcess[\gppFugacity][\intD \countingMeasure] + \int_{\countingMeasures} \ind{\countingMeasure[\region] > m} \modifiedSamplingDensity{\samplingError}[\countingMeasure] \PoissonPointProcess[\gppFugacity][\intD \countingMeasure] \\
		&= \Pr{\xi(\region) > m} + \Pr{\size{\modifiedSampledPointSet} > m}[\samplingGraph \in \approximationGraphs{n}{\frac{\samplingError}{12}}]
	\end{align*}
	for $\xi \sim \GibbsPointProcess{\region}{\gppFugacity}{\potential}$, and $\samplingGraph$ and $\modifiedSampledPointSet$ as in \Cref{algo:modifed_sampling}.

	We proceed by bounding each of these probability separately.
	Note that, by our choice of $n$ it holds that $m \ge \frac{12}{\samplingError} \gppFugacity \volumeMeasure[\region]$.
	By \Cref{lemma:domination_gpp}, we have $\EWrt{\xi(\region)}[\xi \sim \GibbsPointProcess{\region}{\gppFugacity}{\potential}] \le \gppFugacity \volumeMeasure[\region]$.
	Thus, Markov's inequality yields $\Pr{\xi(\region) > m} \le \frac{\samplingError}{12}$.
	Moreover, note that $\size{\modifiedSampledPointSet} = \size{\configurationToSet{\spinConfigurationModified}}$ for $\spinConfigurationModified$ as in \Cref{algo:modifed_sampling}.
	As $\spinConfigurationModified \sim \hcGibbsDistribution{H}{\fugacity[n]}$ for some $H \in \approximationGraphs{n}{\frac{\samplingError}{12}} \subseteq \graphs{n}$ and \Cref{lemma:domination_hc} applies to all such graphs, we get
	\[
		\E{\size{\modifiedSampledPointSet}}[\samplingGraph \in \approximationGraphs{n}{\frac{\samplingError}{12}}] \le \frac{\fugacity[n]}{1+\fugacity[n]} n \le \fugacity[n] n = \gppFugacity \volumeMeasure[\region] .
	\]
	Again, applying Markov's inequality gives $\Pr{\size{\modifiedSampledPointSet} > m}[\samplingGraph \in \approximationGraphs{n}{\frac{\samplingError}{12}}] \le \frac{\samplingError}{12}$.
	Consequently, we have
	\[
		\dtv{\GibbsPointProcess{\region}{\gppFugacity}{\potential}}{\modfiedSamplerOutput{\samplingError}} \le \frac{\samplingError}{4} + \frac{\samplingError}{6} = \frac{5}{12} \samplingError .
	\]

	To finish the proof, we now relate the output of \Cref{algo:modifed_sampling} with $\modfiedSamplerOutput{\samplingError}$ by using a coupling argument.
	To this end, note that \Cref{algo:modifed_sampling} can be used to sample from $\modfiedSamplerOutput{\samplingError}$ by simply restarting the sampler whenever $\samplingGraph \notin \approximationGraphs{n}{\frac{\samplingError}{12}}$.
	For our choice of $n$ we know that with a probability of $\Pr{\samplingGraph \in \approximationGraphs{n}{\frac{\samplingError}{12}}} \ge 1 - \frac{\samplingError}{12}$ only a single run of \Cref{algo:modifed_sampling} is required.
	By this coupling, the total variation distance between the output of \Cref{algo:modifed_sampling} and $\modfiedSamplerOutput{\samplingError}$ is at most $\frac{\samplingError}{12}$.
	Finally, applying triangle inequality shows that the total variation distance between the output of \Cref{algo:modifed_sampling} and $\GibbsPointProcess{\region}{\gppFugacity}{\potential}$ is bounded by $\frac{5}{12} \samplingError + \frac{\samplingError}{12} = \frac{\samplingError}{2}$, which concludes the proof.
\end{proof}

Using \Cref{lemma:modified_sampler_dtv}, we are able to prove that \Cref{algo:sampling} is an $\samplingError$-approximate sampler for $\GibbsPointProcess{\region}{\gppFugacity}{\potential}$.
In order to argue that \Cref{algo:sampling} also satisfies the running time requirements, given in \Cref{thm:sampling}, we require an efficient approximate sampler from the hard-core distribution $\hcGibbsDistribution{\samplingGraph}{\fugacity[n]}$.
To this end, we use the following known result.

\begin{theorem}[{\cite[Theorem $5$]{anari2021entropic}}]
	\label{thm:hc_sampler}
	Let $\graph = (\vertices, \edges)$ be an undirected graph with maximum vertex degree bounded by $\degree[\graph] \in \N_{\ge 2}$ and let $\fugacity \in \R_{\ge 0}$ with
	\[
		\fugacity < \criticalFugacity{\degree[\graph]} = \frac{\left(\degree[\graph] - 1\right)^{\degree[\graph] - 1}}{\left(\degree[\graph] - 2\right)^{\degree[\graph]}}.
	\]
	Then, for all $\error \in (0, 1]$, there is an $\error$-approximate sampler for the hard-core Gibbs distribution $\hcGibbsDistribution{\graph}{\fugacity}$ with an expected running time of $\bigO{\size{\vertices} \ln\left(\frac{\size{\vertices}}{\error}\right)}$.
\end{theorem}

\begin{proof}[Proof of \Cref{thm:sampling}]
	We start by arguing that \Cref{algo:sampling} is an $\samplingError$-approximate sampler for $\GibbsPointProcess{\region}{\gppFugacity}{\potential}$.
	To this end, we show that the total variation distance between the output distributions of \Cref{algo:sampling} and \Cref{algo:modifed_sampling} is bounded by $\frac{\samplingError}{2}$.
	Using the triangle inequality and \Cref{lemma:modified_sampler_dtv} then yields the desired result.
	To bound the total variation distance between the \Cref{algo:sampling} and \Cref{algo:modifed_sampling} by $\frac{\samplingError}{2}$, it suffices to construct a coupling of both algorithms such that their output coincides with a probability of at least $1 - \frac{\samplingError}{2}$.
	This is, we want to find a coupling of both algorithms such that $\sampledPointSet \neq \modifiedSampledPointSet$ with probability at most $\frac{\samplingError}{2}$, where $\sampledPointSet$ and $\modifiedSampledPointSet$ are as in \Cref{algo:sampling} and \Cref{algo:modifed_sampling}.

	To construct such a coupling, we start by letting both algorithms draw the same points $\randomPoint_1, \dots, \randomPoint_n$ and construct the same graph $\samplingGraph$.
	If $\degree[\samplingGraph] \ge \frac{\eulerE n}{\gppFugacity \volumeMeasure[\region]}$, then we may just assume $\sampledPointSet \neq \modifiedSampledPointSet$.
	Otherwise, if $\degree[\samplingGraph] < \frac{\eulerE n}{\gppFugacity \volumeMeasure[\region]}$, then $\spinConfiguration = \spinConfigurationModified$ is a sufficient condition for $\sampledPointSet = \modifiedSampledPointSet$.
	As $\spinConfigurationModified$ is drawn from $\hcGibbsDistribution{\samplingGraph}{\fugacity[n]}$ and $\spinConfiguration$ is drawn from an $\frac{\samplingError}{4}$ approximation of that distribution, they can be coupled in such a way that $\Pr{\spinConfigurationModified \neq \spinConfiguration} \le \frac{\samplingError}{4}$.
	Using this coupling of \Cref{algo:sampling} and \Cref{algo:modifed_sampling}, we have
	\[
		\Pr{\sampledPointSet \neq \modifiedSampledPointSet}
		\le \Pr{\degree[\samplingGraph] \ge \frac{\eulerE n}{\gppFugacity \volumeMeasure[\region]}} + \frac{\samplingError}{4} \cdot \Pr{\degree[\samplingGraph] < \frac{\eulerE n}{\gppFugacity \volumeMeasure[\region]}}
		\le  \Pr{\degree[\samplingGraph] \ge \frac{\eulerE n}{\gppFugacity \volumeMeasure[\region]}} + \frac{\samplingError}{4} .
	\]
	Therefore, it remains to prove that $\degree[\samplingGraph] \ge \frac{\eulerE n}{\gppFugacity \volumeMeasure[\region]}$ with probability at most $\frac{\samplingError}{4}$, where $\samplingGraph \sim \canonicalDistribution{n}{\region}{\potential}$.
	We follow a similar arguments as in the proof of \Cref{thm:approximate_gpp}.
	Note that, for our choice of $n$, there exists $z \ge 3 \ln\left(\frac{4 \eulerE}{\samplingError}\right) \max\left\{\frac{1}{\eulerE - \gppFugacity \generalizedTemperedness{\potential}}, \frac{\gppFugacity \generalizedTemperedness{\potential} }{\left(\eulerE - \gppFugacity \generalizedTemperedness{\potential}\right)^{2}}\right\} \gppFugacity \volumeMeasure[\region]$ such that $n = 2 z \ln(z)^2$.
	Moreover, we have $n \ge \eulerE \ge 2$ and, for $\volumeMeasure[\region]$ (consequently $z$) sufficiently large, it holds that $2 z \ln(z)^2 \ge 2 z \ln\left(2 z \ln(z)^2\right) = 2 z \ln(n)$.
	Therefore, we have
	\begin{align*}
		n - 1
		&\ge \frac{n}{2} \\
		&\ge 3 \ln\left(\frac{4 \eulerE}{\samplingError}\right) \max\left\{\frac{1}{\eulerE - \gppFugacity \generalizedTemperedness{\potential}}, \frac{\gppFugacity \generalizedTemperedness{\potential} }{\left(\eulerE - \gppFugacity \generalizedTemperedness{\potential}\right)^{2}}\right\} \gppFugacity \volumeMeasure[\region] \ln(n) \\
		&\ge 3 \ln\left(\frac{4n}{\samplingError}\right) \max\left\{\frac{1}{\eulerE - \gppFugacity \generalizedTemperedness{\potential}}, \frac{\gppFugacity \generalizedTemperedness{\potential} }{\left(\eulerE - \gppFugacity \generalizedTemperedness{\potential}\right)^{2}}\right\} \gppFugacity \volumeMeasure[\region]
	\end{align*}
	and by \Cref{lemma:degree_bound}
	\[
		\Pr{\degree[\samplingGraph] \ge \frac{\eulerE n}{\gppFugacity \volumeMeasure[\region]}}
		\le \Pr{\degree[\graph] \ge \left(1 + \frac{\eulerE - \gppFugacity \generalizedTemperedness{\potential}}{\gppFugacity  \generalizedTemperedness{\potential}}\right)  \frac{n-1}{\volumeMeasure[\region]}\generalizedTemperedness{\potential}}
		\le \frac{\samplingError}{4} .
	\]

	To prove \Cref{thm:sampling}, it remains to show that \Cref{algo:sampling} satisfies the given running time requirements.
	To this end, note that, for all $\gppFugacity < \frac{\eulerE}{\generalizedTemperedness{\potential}}$, it holds that $n \in \bigOTilde{\volumeMeasure[\region]^2 \samplingError^{-3}}$.
	Therefore, sampling $\randomPoint_1, \dots, \randomPoint_n$ requires a running time of $\bigOTilde{n \samplePointTime{\region}} = \bigOTilde{\volumeMeasure[\region]^2 \samplingError^{-3} \samplePointTime{\region}}$.
	Moreover, the graph can be constructed in time $\bigOTilde{n^2 \evaluatePotentialTime{\potential}} = \bigOTilde{\volumeMeasure[\region]^4 \samplingError^{-6} \evaluatePotentialTime{\potential}}$ and $\degree[\samplingGraph] \ge \frac{\eulerE n}{\gppFugacity \volumeMeasure[\region]}$ can be checked in $\bigO{1}$ if we keep track of $\degree[\samplingGraph]$ while constructing the graph.
	Finally, for $\degree[\graph] < \frac{\eulerE n}{\gppFugacity \volumeMeasure[\region]}$ it holds that
	\[
		\fugacity[n]
		\le \frac{\gppFugacity \volumeMeasure[\region]}{n}
		< \frac{\eulerE}{\degree[\graph]} < 	\criticalFugacity{\degree[\graph]}.
	\]
	Thus, \Cref{thm:hc_sampler} guarantees the existence of an $\frac{\samplingError}{8}$-approximate sampler from $\hcGibbsDistribution{\samplingGraph}{\fugacity[n]}$ with an expected running time in $\bigO{n \ln\left(\frac{n}{\samplingError}\right)} = \bigO{\volumeMeasure[\region]^2 \samplingError^{-3} \ln\left(\frac{\volumeMeasure[\region]}{\samplingError}\right)}$.
	Note that, by Markov's inequality, the probability that this sampler takes more than $\frac{8}{\samplingError}$ times its expected running time is bounded by $\frac{\samplingError}{8}$.
	Therefore, if we run the sampler from \Cref{thm:hc_sampler} with an error bound of $\frac{\samplingError}{8}$ and, whenever the algorithm takes more than $\frac{8}{\samplingError}$ times its expected running time, stop it and return an arbitrary spin configuration, this results in an $\frac{\samplingError}{4}$-approximate sampler with a guaranteed running time in $\bigOTilde{\volumeMeasure[\region]^2 \samplingError^{-4}}$.
	Consequently, \Cref{algo:sampling} runs in time $\bigOTilde{\volumeMeasure[\region]^2 \samplingError^{-4} + \volumeMeasure[\region]^2 \samplingError^{-3} \samplePointTime{\region} + \volumeMeasure[\region]^4 \samplingError^{-6} \evaluatePotentialTime{\potential}}$, which concludes the proof.
\end{proof}

%% file: content/connective_constant.tex
\section{Potential-weighted connective constant and strong spatial mixing} \label{sec:connective_constant}
Throughout this section, we consider the setting introduced in \Cref{sec:approximation}.
In this section, we relate the potential-weighted connective constant $\pwcc{\potential}$ of a repulsive potential with a high-probability bound on a modified version of the connective constant of a graph $\graph \sim \canonicalDistribution{n}{\region}{\potential}$.
This modified connective constant represents the growth rate of a truncated version of the self-avoiding walk tree as used by Weitz~\cite{weitz2006counting}.
Besides giving a graphical interpretation for the potential-weighted connective constant, an immediate consequence of the result is that the hard-core models studied in \Cref{sec:approximation} with high probability exhibit strong spatial mixing for $\fugacity < \frac{\eulerE}{\pwcc{\potential}}$.

\subsection{Potential-weighted connective constant} \label{sec:pwcc}
The potential-weighted connective constant was introduced in \cite{michelen2021potential} to measure the strength of interaction induced by a potential $\potential$ in a way that is, compared to the temperedness constant $\generalizedTemperedness{\potential}$, more sensitive to the particular geometry of the underlying space $\pointProcessSpace$.
To this end, we set $\kPWCC{\potential}{0} = 1$ and, for $k \in \N_{\ge 1}$,
\[
	\kPWCC{\potential}{k} = \sup_{x_0 \in \pointProcessSpace} \int_{\pointProcessSpace^{k}} \prod_{j=1}^{k} \left(\exp\left(- \sum_{i=0}^{j-2} \ind{\dist[x_i][x_j] < \dist[x_i][x_{i+1}]} \potential[x_i][x_j]\right) \cdot \left(1 - \eulerE^{- \potential[x_{j-1}][x_j]}\right)\right) \productVolumeMeasure{k}[\intD \vectorize{x}] ,
\]
where $\vectorize{x} = (x_1, \dots, x_{k})$.
The potential-weighted connective constant is now defined as
\[
	\pwcc{\potential} = \lim_{k \to \infty} \kPWCC{\potential}{k}^{1/k} = \inf_{k \to \infty} \kPWCC{\potential}{k}^{1/k},
\]
where existence if the limit and the second equality are implied by the fact that we assume $\potential$ to be repulsive, which implies that $\kPWCC{\potential}{k}$ is sub-multiplicative.
Note that for all $k \in \N$ it holds that $\kPWCC{\potential}{k} \le \generalizedTemperedness{\potential}^k$ and therefore $\pwcc{\potential} \le \generalizedTemperedness{\potential}$.

\subsection{Strong spatial mixing}
Strong spatial mixing is a frequently used notion of correlation decay in discrete spin systems.
In this section, we focus on strong spatial mixing for the hard-core model.
Recall that for a graph $\graph = (\vertices_{\graph}, \edges_{\graph})$ and a parameter $\fugacity \in \R_{\ge 0}$ we write $\hcGibbsDistribution{\graph}{\fugacity}$ for the hard-core distribution on $\graph$ at weight $\fugacity$, which is a distribution on $\spinConfigurations{\graph}$ the set of all functions $\spinConfigurations{\graph} = \{\spinConfiguration: \vertices \to \{0, 1\}\}$.
We extend this notation to conditional distributions.
To this end, let $S \subseteq \vertices_{\graph}$ and let $\tau: S \to \{0, 1\}$.
Write $\spinConfiguration_{S} = \tau$ for the event that $\spinConfiguration \sim \hcGibbsDistribution{\graph}{\fugacity}$ coincides with $\tau$ on $S$.
We call $\tau$ \emph{feasible} if $\hcGibbsDistribution{\graph}{\fugacity}[\spinConfiguration_{S} = \tau] > 0$.
In that case, we write
\[
	\hcGibbsDistribution{\graph}{\fugacity}[\,\cdot\,\mid\, \spinConfiguration_{S} = \tau]
	= \frac{\hcGibbsDistribution{\graph}{\fugacity}[\cdot]}{\hcGibbsDistribution{\graph}{\fugacity}[\spinConfiguration_{S} = \tau]}
\]
for the distribution of $\spinConfiguration \sim \hcGibbsDistribution{\graph}{\fugacity}$ conditioned on $\spinConfiguration_{S} = \tau$.
Often, strong spatial mixing is phrased in terms of the so-called \emph{occupation ratios}.
For $\graph$ and $\fugacity$ as above and $v \in \vertices_{\graph}$ we write
\[
	\occupationRatio{\graph}{\fugacity}[v] = \frac{\hcGibbsDistribution{\graph}{\fugacity}[\spinConfiguration(v)=1]}{\hcGibbsDistribution{\graph}{\fugacity}[\spinConfiguration(v)=0]} .
\]
Further, for $S \subset \vertices_{\graph}$ with $v \notin S$ and feasible $\tau: S \to \{0, 1\}$ we define $\occupationRatio{\graph}{\fugacity}[v][\tau]$ analogously using the distribution $\hcGibbsDistribution{\graph}{\fugacity}[\,\cdot\,\mid\, \spinConfiguration_{S} = \tau]$ instead.
Based on that, we now state the definition of strong spatial mixing as given in \cite{sinclair2017spatial}.
\begin{definition}[\cite{sinclair2017spatial}] \label{def:ssm}
	The hard-core model with vertex activity $\fugacity \in \R_{\ge 0}$ is set to satisfy strong spatial mixing on a family of graphs $\graphFamily$ if there exists a constant $\mixingConstant \in [0, 1)$ such that for all $\graph \in \graphFamily$, vertices $v \in \vertices_{\graph}$, $S \subseteq \vertices_{\graph} \setminus \{v\}$, and feasible $\tau, \tau': S \to \{0, 1\}$ it holds that
	\[
		\absolute{\occupationRatio{\graph}{\fugacity}[v][\tau] - \occupationRatio{\graph}{\fugacity}[v][\tau']} \le \bigO{\mixingConstant^{s}} ,
	\]
	where $s$ is the graph distance between $v$ and the vertices on which $\tau$ and $\tau'$ differ (i.e., $\{ u \in S \mid \tau(u) \neq \tau'(u)\}$).
	We call $\mixingConstant$ the decay rate.
\end{definition}

\begin{remark} \label{remark:ssm}
	Note that this definition is actually weaker than the definition of (exponential) strong spatial mixing that is usually used in the literature (cf. \cite{weitz2006counting} Definition $2.2$ and the remark following it).
	Usually, this definition requires the existence of constants $\alpha \ge 0, \beta > 0$ (independent of $\graph, v, \tau_1$ and $\tau_2$) such that
	\[
		\absolute{\occupationRatio{\graph}{\fugacity}[v][\tau] - \occupationRatio{\graph}{\fugacity}[v][\tau']}
		\le \alpha \eulerE^{- \beta s} .
	\]
	In contrast, \Cref{def:ssm} only requires such a decay if the distance $s \ge s_0$ is sufficiently large, where the required lower bound $s_0$ might depend on $\size{\vertices_{\graph}}$.
	In fact, this is the case for the strong spatial mixing result for families of finite graphs in \cite{sinclair2017spatial}, which requires $s \ge s_0 \in \bigTheta{\ln(\size{\vertices_{\graph}})}$.
	Such a bound on $s_0$ still allows for applying Weitz's algorithm to obtain an efficient deterministic approximation algorithm for the partition function of the hard-core model.
	However, it is not sufficient for other applications that require the exponential decay to also hold at constant distances, such as recent exact sampling algorithms \cite{anand2022perfect,feng2022perfect} or the rapid mixing result for Glauber dynamics in \cite{weitz2006counting}.
\end{remark}

\subsection{Self-avoiding walk tree and Weitz tree} \label{sec:saw_tree}
Due to \cite{weitz2006counting}, it is well known that strong spatial mixing on a given graph can be studied in terms of a tree construction that is closely related to the \emph{self-avoiding walk tree}.
We refer to this construction as the \emph{Weitz tree}.
Moreover, it was shown by \cite{sinclair2017spatial} that bounding the growth-rate of Weitz trees for a family of graphs can be used to find a vertex activity regime in which the graph family exhibits strong spatial mixing.
In the following paragraphs, we briefly introduce both trees.

\paragraph*{Self-avoiding walk tree}
Given an undirected graph $\graph=(\vertices_{\graph}, \edges_{\graph})$ and a vertex $\rootVertex \in \vertices_{\graph}$, we denote by $\sawTree[\rootVertex][\graph] = (\vertices_{\sawTree}, \edges_{\sawTree})$ the self-avoiding walk tree with root $\rootVertex$.
It is constructed as follow:
\begin{enumerate}
	\item Let $\paths[\rootVertex][\graph]$ denote the set of all simple paths $p = \rootVertex, v_1, \dots, v_k$ (identified by their vertex sequence) in $\graph$ starting at $\rootVertex$, where we call $k$ the length of $p$.
	$\vertices_{\sawTree}$ contains exactly one vertex $w_{p}$ for every such path $p \in \paths[\rootVertex][\graph]$.
	In particular, there is a unique vertex that corresponds to the path $p = \rootVertex$ of length $0$, which is denoted by $w_{\rootVertex}$ and considered the root of $\sawTree[\rootVertex][\graph]$.
	\item For two paths $p, p' \in \paths[\rootVertex][\graph]$ we say that $w_p \in \vertices_{\sawTree}$ is a child of $w_{p'} \in \vertices_{\sawTree}$ (and connect them by an edge in $\edges_{\sawTree}$) if and only if $p$ is obtained from $p'$ by adding one vertex.
	That is, if $p' = \rootVertex, v_1, \dots v_k$, then $p = \rootVertex, v_1, \dots v_k, v_{k+1}$ for some vertex $v_{k+1} \notin \vertices_{\graph} \setminus \{\rootVertex, v_1, \dots v_k\}$ adjacent to $v_k$.
\end{enumerate}

\paragraph*{Weitz tree}
For our purpose, it is crucial to differentiate between the self-avoiding walk tree and a truncated version that we call the Weitz tree.
Its construction is analogous to the self-avoiding walk tree above but with an additional restriction on the set of paths involved.
This restriction on the paths allows for some degrees of freedom, as it depends on assigning an ordering to the neighbors $\neighbors{\graph}{v}$  of each vertex $v \in \vertices_{\graph}$.
Formally, we call a family of functions $\neighborOrderings = (\neighborOrder{v})_{v \in \vertices_{\graph}}$ a \emph{neighborhood ordering} for $\graph$ if, for every $v \in \vertices_{\graph}$, it holds that $\neighborOrder{v}$ is a bijection $\neighbors{\graph}{v} \to [\size{\neighbors{\graph}{v}}]$.
For a vertex $\rootVertex \in \vertices_{\graph}$ we now write $\orderedPaths[\rootVertex][\neighborOrderings][\graph]$ for the set of simple paths $p = v_0, v_1, \dots, v_k$ in $\graph$ with the following properties
\begin{enumerate}
	\item It holds that $v_0 = \rootVertex$.
	\item For every $2 \le i \le k$ and all $0 \le j \le i-2$ it holds that, if $v_i \in \neighbors{\graph}{v_j}$, then $\neighborOrder{v_j}[v_i] > \neighborOrder{v_j}[v_{j+1}]$.
\end{enumerate}
Obviously it holds that $\orderedPaths[\rootVertex][\neighborOrderings][\graph] \subseteq \paths[\rootVertex][\graph]$.
The Weitz tree $\WeitzTree[\rootVertex][\neighborOrderings][\graph]$ with root $\rootVertex$ is now defined analogously to the self-avoiding walk tree $\sawTree[\rootVertex][\graph]$ but restricted to paths in $\orderedPaths[\rootVertex][\neighborOrderings][\graph]$.
It is not hard to see that $\WeitzTree[\rootVertex][\neighborOrderings][\graph]$ is in fact a subtree of $\sawTree[\rootVertex][\graph]$.
Since we are going to study a notion of growth rate of $\WeitzTree[\rootVertex][\neighborOrderings][\graph]$, it is useful to denote by $\layer{k}[\rootVertex][\neighborOrderings][\graph]$ the number of vertices at layer $k \in \N$ of $\WeitzTree[\rootVertex][\neighborOrderings][\graph]$.
This is equal to the number of paths $p = \rootVertex, v_1, \dots, v_k \in \orderedPaths[\rootVertex][\neighborOrderings][\graph]$ of length~$k$.

\begin{remark}
	Note that this construction of the Weitz tree is not exactly as described in \cite{weitz2006counting}.
	The goal of constructing those trees is to study hard-core models on them.
	In the original construction, paths where allowed to close cycles.
	Whenever this happens, the spin of the vertex that closes the cycle was fixed to either~$0$ or~$1$, depending on the chosen neighborhood ordering.
	For the hard-core model, this corresponds to either removing the vertex (for spin~$0$), or the vertex and all its neighbors (for spin~$1$).
	This procedure leads to our notion of the Weitz tree.
\end{remark}

\subsection{Connective constant and strong spatial mixing}
With the definition of strong spatial mixing and the construction of the Weitz tree given, we now get to the definition of the connective constant and its implications for strong spatial mixing.
Our definition of connective constant is inspired by \cite{sinclair2017spatial}.
However, there are two things that should be noted.
Firstly, we emphasize that our definition of connective constant refers to the Weitz tree instead of the full self-avoiding walk.
Secondly, for our application to random graphs from $\canonicalDistribution{n}{\region}{\potential}$, it is important which neighborhood ordering is used for constructing the tree.
To reflect this, we use the following definition.

\begin{definition}\label{def:connective_constant}
	We say a family of graphs $\graphFamily$ has a connective constant bounded by $\connectiveConstant$ if there are constants $a, c \ge 0$ such that for all $\graph \in \graphFamily$ the following holds: there is a neighborhood ordering $\neighborOrderings$ for $\graph$ such that for all $m \ge a \ln(\size{\vertices_{\graph}})$ and all $\rootVertex \in \vertices_{\graph}$ we have
	\[
		\sum_{k = 0}^{m} \layer{k}[\rootVertex][\neighborOrderings][\graph] \le c \connectiveConstant^{m}. \qedhere
	\]
\end{definition}

In \cite{sinclair2017spatial}, it was proven that a bound on the connective constant of a graph family immediately translates to a regime of strong spatial mixing.
For our purposes, we will actually need some more detailed information about the rate of decay.
The following statement can be extracted from the proof of the main theorem in \cite{sinclair2017spatial}.
\begin{theorem}[\cite{sinclair2017spatial}]\label{thm:ssm_from_cc}
	Suppose $\graphFamily$ is a family of graphs with connective constant bounded by $\connectiveConstant$ as in \Cref{def:connective_constant} for constants $a$ and $c$.
	For all $\varepsilon > 0$, $\fugacity \le \eulerE^{-\varepsilon}\criticalFugacity{\connectiveConstant}$, and graphs $\graph \in \graphFamily$ the following holds:
	For all $v \in \vertices_{\graph}$, all $S \subseteq \vertices_{\graph} \setminus \{v\}$, and all feasible $\tau, \tau': S \to \{0, 1\}$ that only differ at distance $s \ge a \ln(\size{\vertices_{\graph}})$ from $v$, it holds that
	\[
		\absolute{\occupationRatio{\graph}{\fugacity}[v][\tau] - \occupationRatio{\graph}{\fugacity}[v][\tau']} \le c^{1/q} \frac{M}{L} \eulerE^{- \varepsilon s / q} ,
	\]
	where $1 < q \le 2$, $M = \sinh^{-1}\left(\sqrt{\fugacity}\right)$ and $L = \frac{1}{2 \sqrt{\fugacity \cdot (\fugacity + 1)}}$.
\end{theorem}

\begin{remark}
	\Cref{thm:ssm_from_cc} follows immediately from tracking the constants in the proof of the main theorem of \cite{sinclair2014approximation}.
	Note further that in the proof in \cite{sinclair2014approximation} the bound from \Cref{thm:ssm_from_cc} is stated with $\frac{L}{M}$ instead of $\frac{M}{L}$.
	However, closely inspecting the proof and the lemmas used therein reveals that this is attributed to a typo.
\end{remark}

\subsection{Connective constant and spatial mixing for random discretizations of Gibbs point processes}
Our first main result of this section is the following bound on the connective constant for random graphs from $\canonicalDistribution{n}{\region}{\potential}$.
\begin{theorem}\label{thm:connective_constant}
	Let $\region \subseteq \pointProcessSpace$ be a bounded and measurable region with volume $\volumeMeasure[\region]>0$.
	Let $\potential: \pointProcessSpace^2 \to \R_{\ge 0} \cup \{\infty\}$ be a symmetric repulsive potential with $\generalizedTemperedness{\potential} < \infty$ and $\pwcc{\potential} > 0$.
	For every $\varepsilon > 0$ there exists some $n_{0} \in \bigTheta{\volumeMeasure[\region]}$ such that for all $n \ge n_0$ the following holds:
	There exists a family of graphs $\graphFamily \subseteq \graphs{n}$ with connective constant bounded by $\eulerE^{\varepsilon} \frac{n}{\volumeMeasure[\region]} \pwcc{\potential}$ such that the constants $a, c$ from \Cref{def:connective_constant} are independent of $\region$ and $n$, and for $\graph \sim \canonicalDistribution{n}{\region}{\potential}$ it holds that
	\[
		\Pr{\graph \in \graphFamily} \ge 1 - \frac{1}{n} . \qedhere
	\]
\end{theorem}

\begin{proof}
	Fix some $\varepsilon > 0$.
	By the definition of $\pwcc{\potential}$ and the fact that $\pwcc{\potential} > 0$, we know that there is some $k_0$ such that $\kPWCC{\potential}{k} \le \eulerE^{k \varepsilon/2} \pwcc{\potential}^{k}$ for all $k \ge k_0$.
	We set $a = 4 \varepsilon^{-1}$, $c = (1 - \eulerE^{-\varepsilon/2}) \cdot \left((\generalizedTemperedness{\potential} / \pwcc{\potential})^{k_0} + 2\right)$, and $n_0 = \max\left\{\frac{2}{\pwcc{\potential}} \volumeMeasure[\region], \eulerE^{ a^{-1} k_0} \right\}$.

	Next, fix any $n \ge n_0$.
	Our goal is to show that, for $\graph \sim \canonicalDistribution{n}{\region}{\potential}$, we can find some neighborhood ordering for $\graph$ such that, with probability at least $1 - 1/n$, the requirements of \Cref{def:connective_constant} are satisfied for $a, c$ as above and $\connectiveConstant = \eulerE^{\varepsilon} \frac{n}{\volumeMeasure[\region]} \pwcc{\potential}$.

	We start by constructing the class of neighborhood orderings we consider.
	Let $\graph \in \graphs{n}$ be a graph on $[n]$.
	For every $\vectorize{x} = (x_1, \dots, x_n) \in \region^{n}$ we construct a neighborhood ordering $\neighborOrderings_{x, \graph} = (\neighborOrder{i})_{i \in [n]}$ on $\graph$ such that, for every vertex $i \in [n]$ and neighbors $j_1, j_2 \in \neighbors{\graph}{i}$, it holds that $\dist[x_i][x_{j_1}] < \dist[x_i][x_{j_2}]$ implies $\neighborOrder{i}[j_1] <  \neighborOrder{i}[j_2]$.
	That is, we order the neighbors of each vertex increasingly by distance, breaking ties arbitrarily (e.g., by vertex IDs).
	It now suffices to show that, for $\graph \sim \canonicalDistribution{n}{\region}{\potential}$, the following event has probability at least $1-1/n$: there is a sequence of points $\vectorize{x} \in \region^{n}$ such that for all $m \ge a \ln(n)$ and all $\rootVertex \in [n]$
	\[
		\sum_{k=0}^{m} \layer{k}[\rootVertex][\neighborOrderings_{\vectorize{x}, \graph}][\graph] \le c \left(\eulerE^{\varepsilon} \frac{n}{\volumeMeasure[\region]} \pwcc{\potential}\right)^{m}.
	\]
	We denote this event by $A$, so that the desired statement is simply expressed as $\canonicalDistribution{n}{\region}{\potential}(A) \ge 1-1/n$.

	To prove this, we study a distribution $\kappa$ on the space $\region^{n} \times \graphs{n}$, equipped with the product sigma field $\BorelOn{\region}^{n} \tensor \powerset{\graphs{n}}$.
	Consider the following procedure:
	\begin{enumerate}
		\item For each $i \in [n]$, draw a uniform random point $X_i \sim \uniformDistributionOn{\region}$ independently.
		We call the resulting random vector $X=(X_i)_{i \in [n]}$
		\item Construct a graph $\graph$ on vertex set $[n]$ as follows.
		For all $i, j \in [n]$ with $i \neq j$, connect vertices $i$ and $j$ with an edge with probability $1 - \eulerE^{- \potential[X_i][X_j]}$ independently.
	\end{enumerate}
	We take $\kappa$ to be the distribution of $(X, \graph)$ generated as above.
	We further write $\kappa_{X}$ and $\kappa_{\graph}$ for the respective marginals and note that $\kappa_{X} = \productUniformDistributionOn{\region}{n}$ and $\kappa_{\graph} = \canonicalDistribution{n}{\region}{\potential}$.
	Further, let $B$ denote the following event: for all $m \ge a \ln(n)$ and all $\rootVertex \in [n]$ it holds that
	\[
		\sum_{k=0}^{m} \layer{k}[\rootVertex][\neighborOrderings_{X, \graph}][\graph] \le c \left(\eulerE^{\varepsilon} \frac{n}{\volumeMeasure[\region]} \pwcc{\potential}\right)^{m}.
	\]
	Note that the main difference between the events $A$ and $B$ is that $A$ asks for the existence of a point sequence~$\vectorize{x}$ for a given random graph $\graph$, whereas $B$ is a statement about a random pair $(X, \graph)$.
	In particular, it holds that $B \subseteq \region^{n} \times A$ and consequently
	\[
		\canonicalDistribution{n}{\region}{\potential}(A) = \kappa_{\graph}(A) = \kappa(\region^{n} \times A) \ge \kappa(B).
	\]
	Thus, the theorem is proven by showing that $\kappa(B) \ge 1 - 1/n$.

	We prove this by first bounding the expectation of $\layer{k}[\rootVertex][\neighborOrderings_{X, \graph}][\graph]$ with respect to $\kappa$ for every fixed $\rootVertex \in [n]$.
	To this end, define $\psi: \pointProcessSpace^{k + 1} \to [0, 1]$ by
	\[
		\psi(y_0, y_1, \dots, y_k) = \prod_{j=1}^{k} \left(\exp\left(- \sum_{i=0}^{j-2} \ind{\dist[y_i][y_j] < \dist[y_i][y_{i+1}]} \potential[y_i][y_j]\right) \cdot \left(1 - \eulerE^{- \potential[y_{j-1}][y_j]}\right)\right) .
	\]
	For every sequence of distinct vertices $i_1, \dots, i_k \in [n] \setminus \{\rootVertex\}$ it now holds that
	\[
		\E{\ind{\rootVertex, i_1, \dots, i_k \in \orderedPaths[\rootVertex][\neighborOrderings_{X, \graph}][\graph]}}[X] \le \psi(X_{\rootVertex}, X_{i_1}, \dots, X_{i_k}).
	\]
	Applying linearity of expectation to sum over all such sequences of distinct vertices $i_1, \dots, i_k \in [n] \setminus \{\rootVertex\}$ and applying law of total expectation then yields that
	\begin{align*}
		\E{\layer{k}[\rootVertex][\neighborOrderings_{X, \graph}][\graph]}
		&\le \volumeMeasure[\region]^{-n} \int_{\region^n} \sum_{i_1, \dots, i_k} \psi(x_{\rootVertex}, x_{i_1}, \dots, x_{i_k}) \productVolumeMeasure{n}[\intD \vectorize{x}] \\
		&\le \volumeMeasure[\region]^{-k} \frac{(n-1)!}{(n-1-k)!} \sup_{z_0 \in \region} \int_{\region^{k}} \psi(z_{0}, z_{1}, \dots, z_{k}) \productVolumeMeasure{k}[\intD \vectorize{z}] \\
		&\le \left(\frac{n}{\volumeMeasure[\region]}\right)^{k} \kPWCC{\potential}{k} ,
	\end{align*}
	where $\vectorize{x} = (x_1, \dots, x_n)$ and $\vectorize{z} = (z_1, \dots, z_k)$.

	Next, we aim to obtain a tail bound for $\sum_{k=0}^{m} \layer{k}[\rootVertex][\neighborOrderings_{X, \graph}][\graph]$ for every $m \ge a \ln(n)$.
	To this end, we first bound the expectation, starting with splitting up the sum as
	\[
		\E{\sum_{k=0}^{m} \layer{k}[\rootVertex][\neighborOrderings_{X, \graph}][\graph]}
		\le \sum_{k=0}^{m} \left(\frac{n}{\volumeMeasure[\region]}\right)^{k} \kPWCC{\potential}{k}
		= \sum_{k=0}^{k_{0}-1} \left(\frac{n}{\volumeMeasure[\region]}\right)^{k} \kPWCC{\potential}{k} + \sum_{k=k_0}^{m} \left(\frac{n}{\volumeMeasure[\region]}\right)^{k} \kPWCC{\potential}{k} .
	\]
	We proceed by bounding each of the sums separately.
	For the first sum, note that it trivially holds that $\kPWCC{\potential}{k} \le \generalizedTemperedness{\potential}^{k}$ and, in particular, $\pwcc{\potential} \le \generalizedTemperedness{\potential}$.
	Moreover, for our choice of $n_0$, we $\frac{n}{\volumeMeasure[\region]} \generalizedTemperedness{\potential} \ge 2$.
	Combining both observations yields
	\[
		\sum_{k=0}^{k_{0}-1} \left(\frac{n}{\volumeMeasure[\region]}\right)^{k} \kPWCC{\potential}{k}
		\le \sum_{k=0}^{k_{0}-1} \left(\frac{n}{\volumeMeasure[\region]} \generalizedTemperedness{\potential}\right)^{k}
		= \frac{\left(\frac{n}{\volumeMeasure[\region]} \generalizedTemperedness{\potential}\right)^{k_0} - 1}{\frac{n}{\volumeMeasure[\region]} \generalizedTemperedness{\potential} - 1}
		\le \left(\frac{n}{\volumeMeasure[\region]} \generalizedTemperedness{\potential}\right)^{k_0} .
	\]
	For the second sum, recall that $\kPWCC{\potential}{k} \le \eulerE^{k \varepsilon/2} \pwcc{\potential}^k$ for all $k \ge k_0$.
	Therefore, we have
	\[
		\sum_{k=k_0}^{m} \left(\frac{n}{\volumeMeasure[\region]}\right)^{k} \kPWCC{\potential}{k}
		\le \sum_{k=k_0}^{m} \left( \eulerE^{\varepsilon/2} \frac{n}{\volumeMeasure[\region]} \pwcc{\potential} \right)^{k}
		\le  \frac{\left(\eulerE^{\varepsilon/2} \frac{n}{\volumeMeasure[\region]} \pwcc{\potential}\right)^{m + 1}}{\eulerE^{\varepsilon/2} \frac{n}{\volumeMeasure[\region]} \pwcc{\potential} - 1} .
	\]
	Combining both bounds yields
	\begin{align*}
		\E{\sum_{k=0}^{m} \layer{k}[\rootVertex][\neighborOrderings_{X, \graph}][\graph]}
		\le \left(\frac{\left(\frac{n}{\volumeMeasure[\region]} \generalizedTemperedness{\potential}\right)^{k_0}}{\left(\eulerE^{\varepsilon/2} \frac{n}{\volumeMeasure[\region]} \pwcc{\potential}\right)^{m}} + \frac{\eulerE^{\varepsilon/2} \frac{n}{\volumeMeasure[\region]} \pwcc{\potential}}{\eulerE^{\varepsilon/2} \frac{n}{\volumeMeasure[\region]} \pwcc{\potential} - 1}\right) \cdot \left(\eulerE^{\varepsilon/2} \frac{n}{\volumeMeasure[\region]} \pwcc{\potential}\right)^{m} .
	\end{align*}
	Since $\eulerE^{\varepsilon/2} \frac{n}{\volumeMeasure[\region]} \pwcc{\potential} \ge 2$ for $n \ge n_0$, we further bound
	\begin{align*}
		\E{\sum_{k=0}^{m} \layer{k}[\rootVertex][\neighborOrderings_{X, \graph}][\graph]}
		&\le \left(\left(\frac{\frac{n}{\volumeMeasure[\region]} \generalizedTemperedness{\potential}}{\eulerE^{\varepsilon/2} \frac{n}{\volumeMeasure[\region]} \pwcc{\potential}}\right)^{k_0} + 2\right) \cdot \left(\eulerE^{\varepsilon/2} \frac{n}{\volumeMeasure[\region]} \pwcc{\potential}\right)^{m} \\
		&\le \left(\left(\frac{\generalizedTemperedness{\potential}}{\pwcc{\potential}}\right)^{k_0} + 2\right) \cdot \left(\eulerE^{\varepsilon/2} \frac{n}{\volumeMeasure[\region]} \pwcc{\potential}\right)^{m}.
	\end{align*}

	Using this upper bound for the expectation and Markov's inequality, we have for $(X, \graph) \sim \kappa$ that
	\[
		\Pr{\sum_{k=0}^{m} \layer{k}[\rootVertex][\neighborOrderings_{X, \graph}][\graph] \ge c \left(\eulerE^{\varepsilon} \frac{n}{\volumeMeasure[\region]} \pwcc{\potential}\right)^{m}} \le (1 - \eulerE^{-\varepsilon/2}) \cdot \eulerE^{ - m \varepsilon / 2} .
	\]
	Applying the union bound over $m \ge m_0 = a \ln(n)$ we have
	\begin{align*}
		\Pr{\exists m \ge m_0: ~\sum_{k=0}^{m} \layer{k}[\rootVertex][\neighborOrderings_{X, \graph}][\graph] \ge c \left(\eulerE^{\varepsilon} \frac{n}{\volumeMeasure[\region]} \pwcc{\potential}\right)^{m}}
		\le (1 - \eulerE^{-\varepsilon/2}) \cdot \sum_{m = m_0}^{\infty} \eulerE^{- m \varepsilon / 2}
		\le n^{-2} ,
	\end{align*}
	where the last inequality is obtained by factoring out $\eulerE^{-m_0 \varepsilon/2}$ and using the fact that $m_0 = a \ln(n) \ge 4 \varepsilon^{-1} \ln(n)$.
	Finally, applying the union bound over the choice of root vertices $\rootVertex$ shows that $\kappa(B^{c}) \le 1/n$ and consequently $\kappa(B) \ge 1-1/n$, proving the theorem.
\end{proof}

We proceed by studying which notion of strong spatial mixing for graphs from $\canonicalDistribution{n}{\region}{\potential}$ we can obtain from \Cref{thm:connective_constant}.
To this end, we apply \Cref{thm:ssm_from_cc}, which yields the following result.

\begin{theorem} \label{thm:ssm_discretization}
	Suppose $\region \subseteq \pointProcessSpace$ is a bounded and measurable region with volume $\volumeMeasure[\region]>0$.
	Let $\potential: \pointProcessSpace^2 \to \R_{\ge 0} \cup \{\infty\}$ be a symmetric repulsive potential with $\generalizedTemperedness{\potential} < \infty$ and $\pwcc{\potential} > 0$, and let $\gppFugacity < \frac{\eulerE}{\pwcc{\potential}}$.
	There exists some $n_0 \in \bigTheta{\volumeMeasure[\region]}$ such that for all $n \ge n_0$ there is a family of graphs $\graphFamily \subseteq \graphs{n}$ with the following properties:
	\begin{enumerate}
		\item For $\graph \sim \canonicalDistribution{n}{\region}{\potential}$ it holds that
		\[
			\Pr{\graph \in \graphFamily} \ge 1 - \frac{1}{n} .
		\]
		\item Set $\fugacity = \gppFugacity \cdot \frac{\volumeMeasure[\region]}{n}$. There are constants $\alpha \ge 0, \beta > 0, a \ge 0$ independent $n$ and $\region$ such that for every $\graph \in \graphFamily$, all $v \in \vertices_{\graph}$, $S \subseteq \vertices_{\graph} \setminus \{v\}$ and feasible $\tau, \tau': S \to \{0, 1\}$ that only differ at distance $s \ge a \ln(n)$ from $v$ it holds that
		\[
			\absolute{\occupationRatio{\graph}{\fugacity}[v][\tau] - \occupationRatio{\graph}{\fugacity}[v][\tau']} \le \alpha \eulerE^{-\beta s} . \qedhere
		\]
	\end{enumerate}
\end{theorem}

\begin{proof}
	First, note that for $\gppFugacity < \frac{\eulerE}{\pwcc{\potential}}$ we can choose $\varepsilon > 0$ such that $\gppFugacity \le \frac{\eulerE^{1 - \varepsilon}}{\pwcc{\potential}}$.
	By \Cref{thm:connective_constant} we know that for $n_0 \in \bigTheta{\volumeMeasure[\region]}$ sufficiently large it holds that for all $n \ge n_0$ that there is a graph family $\graphFamily \subseteq\graphs{n}$ with connective constant bounded by $\connectiveConstant = \eulerE^{\varepsilon/2} \frac{n}{\volumeMeasure[\region]} \pwcc{\potential}$ such that $\graph \sim \canonicalDistribution{n}{\region}{\potential}$ is in $\graphFamily$ with probability at least $1 - \frac{1}{n}$.
	Thus, $\graphFamily$ satisfies the first requirement of our theorem.

	We proceed by establishing the second part of the theorem.
	To this end, let $\graphFamily$ and $\connectiveConstant$ be as above.
	Further, let $a, c$ be the constants given by \Cref{thm:connective_constant} and recall that they are independent of $n$ and $\volumeMeasure[\region]$.
	Observe that
	\[
		\fugacity = \gppFugacity \cdot \frac{\volumeMeasure[\region]}{n} \le \frac{\eulerE^{1 - \varepsilon}}{\pwcc{\potential}} \cdot \frac{\volumeMeasure[\region]}{n} \le \frac{\eulerE^{1 - \varepsilon/2}}{\connectiveConstant} .
	\]
	Thus, applying \Cref{thm:ssm_from_cc} proves our claim as soon as we show that $\frac{M}{L}$ for $M = \sinh^{-1}\left(\sqrt{\fugacity}\right)$ and $L = \frac{1}{2 \sqrt{\fugacity \cdot (\fugacity + 1)}}$ is uniformly bounded in $n$ and $\volumeMeasure[\region]$.
	Recalling that $\fugacity = \gppFugacity \cdot \frac{\volumeMeasure[\region]}{n}$, we see that $\frac{M}{L}$ is decreasing in $n$ and bounded by a constant independent of $\volumeMeasure[\region]$ as soon as $n_{0} \ge \gppFugacity \cdot \volumeMeasure[\region]$.
\end{proof}

\begin{remark}
	Note that from the proof of \Cref{thm:ssm_discretization} it actually even follows that the parameter $\alpha$ goes to~$0$ as $n$ increases.
	This is to be expected, since the occupation ratios go to zero as $\fugacity[n]$ decreases.
\end{remark}

We obtain the following corollary.

\begin{corollary} \label{cor:ssm_discretization}
	Suppose $\region \subseteq \pointProcessSpace$ is a bounded and measurable region with volume $\volumeMeasure[\region]>0$.
	Let $\potential: \pointProcessSpace^2 \to \R_{\ge 0} \cup \{\infty\}$ be a symmetric repulsive potential with $\generalizedTemperedness{\potential} < \infty$ and $\pwcc{\potential} > 0$, and let $\gppFugacity < \frac{\eulerE}{\pwcc{\potential}}$.
	There exists some $n_0 \in \bigTheta{\volumeMeasure[\region]}$ such that for all $n \ge n_0$ there is a family of graphs $\graphFamily \subseteq \graphs{n}$ with the following properties:
	\begin{enumerate}
		\item For $\graph \sim \canonicalDistribution{n}{\region}{\potential}$ it holds that
		\[
		\Pr{\graph \in \graphFamily} \ge 1 - \frac{1}{n} .
		\]
		\item The hard-core model with activity $\fugacity = \gppFugacity \cdot \frac{\volumeMeasure[\region]}{n}$ exhibits strong spatial mixing in the sense of \Cref{def:ssm} on $\graphFamily$ with decay rate independent of $n$ and $\region$. \qedhere
	\end{enumerate}
\end{corollary}

We finish this section with discussing some algorithmic consequences of the results above.
However, we keep the discussion informal since the results can be obtained from standard techniques.
We start with the observation that \Cref{thm:ssm_from_cc} does in fact not only imply strong spatial mixing on each graph in the family $\graphFamily$, but it also implies a notion of strong spatial mixing in the Weitz trees that were used for bounding the connective constant of $\graphFamily$.
More precisely, each of the trees exhibits strong spatial mixing in term of the occupation ratio of the root, given boundary conditions sufficiently far down the tree (see \cite{sinclair2017spatial}).
This can be turned into a deterministic approximation algorithm for the partition function of the hard-core model as argued in \cite{weitz2006counting,sinclair2017spatial} with running time $\size{\vertices_{\graph}}^{\bigO{\log(\connectiveConstant)}}$ for every graph in $\graph \in \graphFamily$.
Using \Cref{thm:gpp_concentration}, we may use this to obtain a randomized algorithm with running time $\volumeMeasure[\region]^{\bigO{\log(\volumeMeasure[\region])}}$ (i.e., quasi-polynomial in $\volumeMeasure[\region]$) for approximating $\gppPartitonFunction{\region}[\gppFugacity][\potential]$ for $\gppFugacity < \frac{\eulerE}{\pwcc{\potential}}$ via the following procedure:
We first draw a random graph $\graph \sim \canonicalDistribution{n}{\region}{\potential}$ together with its (random) vertex locations $X = X_1, \dots, X_n \in \region^{n}$ for $n$ sufficiently large to satisfy \Cref{thm:gpp_concentration} and \Cref{thm:ssm_discretization}.
We then use the algorithm given in \cite{weitz2006counting} together with the Weitz trees based on the neighborhood ordering $\neighborOrderings_{X, \graph}$ (see proof of \Cref{thm:connective_constant}) to approximate the hard-core partition function on $\graph$ at vertex activity $\gppFugacity \frac{\volumeMeasure[\region]}{n}$.
An analogous procedure can be used to approximately sample from $\GibbsPointProcess{\region}{\gppFugacity}{\potential}$ with similar running time.

%% file: main.bbl
\begin{thebibliography}{10}

\bibitem{anand2023perfect}
Konrad Anand, Andreas G{\"o}bel, Marcus Pappik, and Will Perkins.
\newblock Perfect sampling for hard spheres from strong spatial mixing.
\newblock In {\em Random}, 2023.

\bibitem{anand2022perfect}
Konrad Anand and Mark Jerrum.
\newblock Perfect sampling in infinite spin systems via strong spatial mixing.
\newblock {\em SIAM Journal on Computing}, 51(4):1280--1295, 2022.

\bibitem{anari2021entropic}
Nima Anari, Vishesh Jain, Frederic Koehler, Huy~Tuan Pham, and Thuy{-}Duong
  Vuong.
\newblock Entropic independence {II:} optimal sampling and concentration via
  restricted modified log-sobolev inequalities.
\newblock {\em CoRR}, abs/2111.03247, 2021.
\newblock URL: \url{https://arxiv.org/abs/2111.03247}.

\bibitem{anari2021spectral}
Nima Anari, Kuikui Liu, and Shayan~Oveis Gharan.
\newblock Spectral independence in high-dimensional expanders and applications
  to the hardcore model.
\newblock {\em SIAM Journal on Computing}, 0(0):FOCS20--1--FOCS20--37, 2021.
\newblock \href {https://doi.org/10.1137/20M1367696}
  {\path{doi:10.1137/20M1367696}}.

\bibitem{bacher2014explaining}
Andreas~K. Bacher, Thomas~B. Schr{\o}der, and Jeppe~C. Dyre.
\newblock Explaining why simple liquids are quasi-universal.
\newblock {\em Nature Communications}, 5(1):1--7, 2014.
\newblock \href {https://doi.org/10.1038/ncomms6424}
  {\path{doi:10.1038/ncomms6424}}.

\bibitem{baddeley2006case}
Adrian Baddeley, Pablo Gregori, Jorge Mateu, Radu Stoica, and Dietrich Stoyan.
\newblock {\em Case Studies in Spatial Point Process Modeling}, volume 185.
\newblock Springer, 2006.
\newblock \href {https://doi.org/10.1007/0-387-31144-0}
  {\path{doi:10.1007/0-387-31144-0}}.

\bibitem{berne1972gaussian}
Bruce~J Berne and Philip Pechukas.
\newblock Gaussian model potentials for molecular interactions.
\newblock {\em The Journal of Chemical Physics}, 56(8):4213--4216, 1972.

\bibitem{betsch2021uniqueness}
Steffen Betsch and G{\"u}nter Last.
\newblock On the uniqueness of gibbs distributions with a non-negative and
  subcritical pair potential.
\newblock {\em arXiv preprint arXiv:2108.06303}, 2021.

\bibitem{DBLP:journals/tcs/BringmannKL19}
Karl Bringmann, Ralph Keusch, and Johannes Lengler.
\newblock Geometric inhomogeneous random graphs.
\newblock {\em Theoretical Computer Science}, 760:35--54, 2019.
\newblock \href {https://doi.org/10.1016/j.tcs.2018.08.014}
  {\path{doi:10.1016/j.tcs.2018.08.014}}.

\bibitem{chen2022rapid}
Xiaoyu Chen, Weiming Feng, Yitong Yin, and Xinyuan Zhang.
\newblock Rapid mixing of glauber dynamics via spectral independence for all
  degrees.
\newblock In {\em 2021 IEEE 62nd Annual Symposium on Foundations of Computer
  Science (FOCS)}, pages 137--148, 2022.
\newblock \href {https://doi.org/10.1109/FOCS52979.2021.00022}
  {\path{doi:10.1109/FOCS52979.2021.00022}}.

\bibitem{chen2020rapid}
Zongchen Chen, Kuikui Liu, and Eric Vigoda.
\newblock Rapid mixing of glauber dynamics up to uniqueness via contraction.
\newblock In {\em 2020 IEEE 61st Annual Symposium on Foundations of Computer
  Science (FOCS)}, pages 1307--1318, 2020.
\newblock \href {https://doi.org/10.1109/FOCS46700.2020.00124}
  {\path{doi:10.1109/FOCS46700.2020.00124}}.

\bibitem{chen2021optimal}
Zongchen Chen, Kuikui Liu, and Eric Vigoda.
\newblock Optimal mixing of glauber dynamics: Entropy factorization via
  high-dimensional expansion.
\newblock In {\em Proceedings of the 53rd Annual ACM SIGACT Symposium on Theory
  of Computing (STOC)}, pages 1537--1550, 2021.
\newblock \href {https://doi.org/10.1145/3406325.3451035}
  {\path{doi:10.1145/3406325.3451035}}.

\bibitem{daley2008introduction}
Daryl~J. Daley and David Vere-Jones.
\newblock {\em An Introduction to the Theory of Point Processes. Volume II:
  General Theory and Structure}.
\newblock Springer, 2008.
\newblock \href {https://doi.org/10.1111/j.1751-5823.2008.00054_18.x}
  {\path{doi:10.1111/j.1751-5823.2008.00054_18.x}}.

\bibitem{efron1981jackknife}
Bradley Efron and Charles Stein.
\newblock The jackknife estimate of variance.
\newblock {\em The Annals of Statistics}, pages 586--596, 1981.
\newblock \href {https://doi.org/10.1214/aos/1176345462}
  {\path{doi:10.1214/aos/1176345462}}.

\bibitem{el2000line}
Fouad El~Azhar, Marc Baus, Jean-Paul Ryckaert, and Evert~Jan Meijer.
\newblock Line of triple points for the hard-core yukawa model: A computer
  simulation study.
\newblock {\em The Journal of Chemical Physics}, 112(11):5121--5126, 2000.

\bibitem{feng2022perfect}
Weiming Feng, Heng Guo, and Yitong Yin.
\newblock Perfect sampling from spatial mixing.
\newblock {\em Random Structures \& Algorithms}, 61(4):678--709, 2022.

\bibitem{friedrich2021algorithms}
Tobias Friedrich, Andreas G{\"o}bel, Maximilian Katzmann, Martin~S. Krejca, and
  Marcus Pappik.
\newblock Algorithms for hard-constraint point processes via discretization.
\newblock In {\em Computing and Combinatorics: 28th International Conference,
  COCOON 2022, Shenzhen, China, October 22--24, 2022, Proceedings}, pages
  242--254. Springer, 2023.

\bibitem{friedrich2021spectral}
Tobias Friedrich, Andreas G{\"o}bel, Martin Krejca, and Marcus Pappik.
\newblock A spectral independence view on hard spheres via block dynamics.
\newblock In {\em 48th International Colloquium on Automata, Languages, and
  Programming (ICALP)}, volume 198, pages 66:1--66:15, 2021.
\newblock \href {https://doi.org/10.4230/LIPIcs.ICALP.2021.66}
  {\path{doi:10.4230/LIPIcs.ICALP.2021.66}}.

\bibitem{garcia2000perfect}
Nancy~L. Garcia.
\newblock Perfect simulation of spatial processes.
\newblock {\em Resenhas do Instituto de Matem{\'a}tica e Estat{\'\i}stica da
  Universidade de S{\~a}o Paulo}, 4(3):283--325, 2000.

\bibitem{georgii1997stochastic}
Hans-Otto Georgii and Torsten K{\"u}neth.
\newblock Stochastic comparison of point random fields.
\newblock {\em Journal of Applied Probability}, 34(4):868--881, 1997.
\newblock \href {https://doi.org/10.2307/3215003} {\path{doi:10.2307/3215003}}.

\bibitem{guo2021perfect}
Heng Guo and Mark Jerrum.
\newblock Perfect simulation of the hard disks model by partial rejection
  sampling.
\newblock {\em Annales de l’Institut Henri Poincar{\'e} D}, 8(2):159--177,
  2021.
\newblock \href {https://doi.org/10.4171/AIHPD/99}
  {\path{doi:10.4171/AIHPD/99}}.

\bibitem{haggstrom1999characterization}
Olle H{\"a}ggstr{\"o}m, Marie-Colette~NM Van~Lieshout, and Jesper M{\o}ller.
\newblock Characterization results and markov chain monte carlo algorithms
  including exact simulation for some spatial point processes.
\newblock {\em Bernoulli}, pages 641--658, 1999.
\newblock \href {https://doi.org/10.2307/3318694} {\path{doi:10.2307/3318694}}.

\bibitem{helmuth2022correlation}
Tyler Helmuth, Will Perkins, and Samantha Petti.
\newblock Correlation decay for hard spheres via markov chains.
\newblock {\em The Annals of Applied Probability}, 32(3):2063--2082, 2022.

\bibitem{huber2012spatial}
Mark Huber.
\newblock Spatial birth—death swap chains.
\newblock {\em Bernoulli}, pages 1031--1041, 2012.

\bibitem{jansen2018gibbsian}
Sabine Jansen.
\newblock Gibbsian point processes, 2018.
\newblock URL: \url{https://www.math.lmu.de/~jansen/gibbspp.pdf}.

\bibitem{jenssen2022quasipolynomial}
Matthew Jenssen, Marcus Michelen, and Mohan Ravichandran.
\newblock Quasipolynomial-time algorithms for gibbs point processes.
\newblock {\em Combinatorics, Probability and Computing}, pages 1--15, 2023.

\bibitem{krioukov2010hyperbolic}
Dmitri Krioukov, Fragkiskos Papadopoulos, Maksim Kitsak, Amin Vahdat, and
  Mari{\'a}n Bogun{\'a}.
\newblock Hyperbolic geometry of complex networks.
\newblock {\em Physical Review E}, 82(3):036106, 2010.
\newblock \href {https://doi.org/10.1103/PhysRevE.82.036106}
  {\path{doi:10.1103/PhysRevE.82.036106}}.

\bibitem{li2013correlation}
Liang Li, Pinyan Lu, and Yitong Yin.
\newblock Correlation decay up to uniqueness in spin systems.
\newblock In {\em Proceedings of the twenty-fourth annual ACM-SIAM symposium on
  Discrete algorithms (SODA)}, pages 67--84, 2013.
\newblock URL: \url{https://dlnext.acm.org/doi/10.5555/2627817.2627822}.

\bibitem{lovasz2012large}
L{\'a}szl{\'o} Lov{\'a}sz.
\newblock {\em Large networks and graph limits}, volume~60.
\newblock American Mathematical Soc., 2012.

\bibitem{mcdiarmid1989method}
Colin McDiarmid.
\newblock On the method of bounded differences.
\newblock {\em Surveys in combinatorics}, 141(1):148--188, 1989.
\newblock \href {https://doi.org/10.1017/CBO9781107359949.008}
  {\path{doi:10.1017/CBO9781107359949.008}}.

\bibitem{mcdiarmid1998concentration}
Colin McDiarmid.
\newblock Concentration.
\newblock In {\em Probabilistic methods for algorithmic discrete mathematics},
  pages 195--248. Springer, 1998.
\newblock \href {https://doi.org/10.1007/978-3-662-12788-9_6}
  {\path{doi:10.1007/978-3-662-12788-9_6}}.

\bibitem{Metropolis}
Nicholas Metropolis, Arianna~W. Rosenbluth, Marshall~N. Rosenbluth, Augusta~H.
  Teller, and Edward Teller.
\newblock Equation of state calculations by fast computing machines.
\newblock {\em The Journal of Chemical Physics}, 21(6):1087--1092, 1953.
\newblock \href {https://doi.org/10.1063/1.1699114}
  {\path{doi:10.1063/1.1699114}}.

\bibitem{michelen2021potential}
Marcus Michelen and Will Perkins.
\newblock Potential-weighted connective constants and uniqueness of gibbs
  measures.
\newblock {\em CoRR}, abs/2109.01094, 2021.
\newblock URL: \url{https://arxiv.org/abs/2109.01094}.

\bibitem{michelen2022strong}
Marcus Michelen and Will Perkins.
\newblock Strong spatial mixing for repulsive point processes.
\newblock {\em Journal of Statistical Physics}, 189(1):9, 2022.

\bibitem{moller1989rate}
Jesper M{\o}ller.
\newblock On the rate of convergence of spatial birth-and-death processes.
\newblock {\em Annals of the Institute of Statistical Mathematics},
  41(3):565--581, 1989.
\newblock \href {https://doi.org/10.1007/BF00050669}
  {\path{doi:10.1007/BF00050669}}.

\bibitem{moller2007modern}
Jesper M{\o}ller and Rasmus~P Waagepetersen.
\newblock Modern statistics for spatial point processes.
\newblock {\em Scandinavian Journal of Statistics}, 34(4):643--684, 2007.
\newblock \href {https://doi.org/10.1111/j.1467-9469.2007.00569.x}
  {\path{doi:10.1111/j.1467-9469.2007.00569.x}}.

\bibitem{moller2003statistical}
Jesper M{\o}ller and Rasmus~Plenge Waagepetersen.
\newblock {\em Statistical inference and simulation for spatial point
  processes}.
\newblock CRC press, 2003.

\bibitem{rowlinson1989yukawa}
JS~Rowlinson.
\newblock The yukawa potential.
\newblock {\em Physica A: Statistical Mechanics and its Applications},
  156(1):15--34, 1989.

\bibitem{ruelle1999statistical}
David Ruelle.
\newblock {\em Statistical Mechanics: Rigorous Results}.
\newblock World Scientific, 1999.

\bibitem{sinclair2017spatial}
Alistair Sinclair, Piyush Srivastava, Daniel {\v{S}}tefankovi{\v{c}}, and
  Yitong Yin.
\newblock Spatial mixing and the connective constant: Optimal bounds.
\newblock {\em Probability Theory and Related Fields}, 168(1):153--197, 2017.
\newblock \href {https://doi.org/10.1007/s00440-016-0708-2}
  {\path{doi:10.1007/s00440-016-0708-2}}.

\bibitem{sinclair2014approximation}
Alistair Sinclair, Piyush Srivastava, and Marc Thurley.
\newblock Approximation algorithms for two-state anti-ferromagnetic spin
  systems on bounded degree graphs.
\newblock {\em Journal of Statistical Physics}, 155(4):666--686, 2014.
\newblock \href {https://doi.org/10.1007/s10955-014-0947-5}
  {\path{doi:10.1007/s10955-014-0947-5}}.

\bibitem{sinclair2013spatial}
Alistair Sinclair, Piyush Srivastava, and Yitong Yin.
\newblock Spatial mixing and approximation algorithms for graphs with bounded
  connective constant.
\newblock In {\em 2013 IEEE 54th Annual Symposium on Foundations of Computer
  Science}, pages 300--309. IEEE, 2013.

\bibitem{sly2012computational}
Allan Sly and Nike Sun.
\newblock The computational hardness of counting in two-spin models on
  $d$-regular graphs.
\newblock In {\em 2012 IEEE 53rd Annual Symposium on Foundations of Computer
  Science (FOCS)}, pages 361--369, 2012.
\newblock \href {https://doi.org/10.1109/FOCS.2012.56}
  {\path{doi:10.1109/FOCS.2012.56}}.

\bibitem{vstefankovivc2009adaptive}
Daniel {\v{S}}tefankovi{\v{c}}, Santosh Vempala, and Eric Vigoda.
\newblock Adaptive simulated annealing: A near-optimal connection between
  sampling and counting.
\newblock {\em Journal of the ACM}, 56(3):1--36, 2009.
\newblock \href {https://doi.org/10.1145/1516512.1516520}
  {\path{doi:10.1145/1516512.1516520}}.

\bibitem{van2000markov}
Marie-Colette~NM Van~Lieshout.
\newblock {\em Markov point processes and their applications}.
\newblock World Scientific, 2000.

\bibitem{vigoda2001note}
Eric Vigoda.
\newblock A note on the glauber dynamics for sampling independent sets.
\newblock {\em The Electronic Journal of Combinatorics}, 8(1):1--8, 2001.
\newblock \href {https://doi.org/10.37236/1552} {\path{doi:10.37236/1552}}.

\bibitem{weitz2006counting}
Dror Weitz.
\newblock Counting independent sets up to the tree threshold.
\newblock In {\em Proceedings of the thirty-eighth annual ACM symposium on
  Theory of computing}, pages 140--149, 2006.

\bibitem{yoshida1974liquid}
Takeshi Yoshida and Shiro Kamakura.
\newblock Liquid-solid transitions in systems of soft repulsive forces:
  Softness of potentials and a maximum in melting curves.
\newblock {\em Progress of Theoretical Physics}, 52(3):822--839, 1974.
\newblock \href {https://doi.org/10.1143/PTP.52.822}
  {\path{doi:10.1143/PTP.52.822}}.

\end{thebibliography}
